\documentclass[12pt]{article}
\usepackage{amsmath, amssymb}
\usepackage{graphicx}
\usepackage{bm}
\usepackage{enumerate}
\usepackage[colorlinks,citecolor=blue]{hyperref}
\usepackage{color}
\usepackage{appendix}
\usepackage{mathrsfs}
\usepackage{float}
\usepackage{bigstrut}
\usepackage{MnSymbol}
\usepackage{cleveref}
\usepackage{subfigure}
\usepackage{ulem}
\usepackage{booktabs}
\usepackage{ulem}
\usepackage{pdflscape}
\usepackage{multirow}
\usepackage{natbib}
\usepackage{ntheorem}
\usepackage{authblk}
\usepackage{url} 
\usepackage{enumitem}
\usepackage[top=1in,bottom=1.3in,left=1in,right=1in]{geometry}
\usepackage{xr}

\graphicspath{{figures/}}
\newcommand*{\QED}{\null\nobreak\hfill\ensuremath{\blacksquare}}%

\newcommand{\diff}{\mathrm{d}}
\newcommand{\reall}{\mathbb{R}}
\newcommand{\tang}{\mathrm{Tan}}
\newcommand{\expm}{\mathrm{Exp}}
\newcommand{\logm}{\mathrm{Log}}
\newcommand{\dsp}{\mathcal{W}_{2}^{ac}}
\newcommand{\wsp}{\mathcal{W}_{2}}
\newcommand{\lnorm}{\left|\left|\left|}
\newcommand{\rnorm}{\right|\right|\right|}
\newcommand{\tdomain}{\mathcal{T}}
\newcommand{\wdomain}{\mathcal{S}}

\newcommand{\dwsp}{\wsp(\wdomain)}
\newcommand{\Log}{\mathrm{Log}}
\newcommand{\Tan}{\mathrm{Tan}}
\newcommand{\blind}{1}

\theoremstyle{plain}
\theoremheaderfont{\scshape}
\theoremseparator{.}
\theorembodyfont{\itshape}
\newtheorem{theorem}{\indent\large Theorem}
\newtheorem{prop}[theorem]{\indent\large Proposition}
\newtheorem{lemma}[theorem]{\indent\large Lemma}
\newtheorem{remark}{\indent\large Remark}

\makeatletter
\newtheoremstyle{proof}%
{\item[\hskip\labelsep\theorem@headerfont\MakeUppercase ##1\theorem@separator]}%
{\item[\hskip \labelsep\theorem@headerfont\MakeUppercase ##1\ ##3\theorem@separator]}
\makeatother

\theoremstyle{proof}
\theoremheaderfont{\scshape}
\theorembodyfont{\upshape}
\theoremsymbol{\ensuremath{\square}}
\newtheorem{proof}{Proof}
\theoremseparator{}
\newtheorem{assum}{\indent}

\begin{document}

	\def\spacingset#1{\renewcommand{\baselinestretch}%
		{#1}\small\normalsize} \spacingset{1}
	
	\def\lin#1{\textcolor{red}{#1}}
	\def\zhou#1{\textcolor{red}{#1}}
	
	
	\if1\blind
	{
		\title{\bf Intrinsic Wasserstein Correlation Analysis}
		\author[1]{Hang Zhou}
		\author[2]{Zhenhua Lin}
		\author[1]{Fang Yao}
		\affil[1]{Department of Probability and Statistics, School of Mathematical Sciences and Center for Statistical Science, Peking University}
		\affil[2]{Department of Statistics and Applied Probability, National University of Singapore}
		\maketitle
	} \fi
	
	\if0\blind
	{
		\bigskip
		\bigskip
		\bigskip
		\begin{center}
			{\LARGE\bf Intrinsic Wasserstein Correlation Analysis}
		\end{center}
		\medskip
	} \fi
	
	\bigskip
	\begin{abstract}
	We develop a framework of canonical correlation analysis for distribution-valued functional data within the geometry of Wasserstein spaces. Specifically, we formulate an intrinsic concept of correlation between random distributions,  propose estimation methods based on functional principal component analysis (FPCA) and Tikhonov regularization, respectively, for the correlation and its corresponding weight functions, and  establish the minimax convergence rates of the estimators. The key idea is to extend the framework of tensor Hilbert spaces to distribution-valued functional data  to overcome the challenging issue raised by nonlinearity of Wasserstein spaces. The finite-sample performance of the proposed estimators is illustrated via  simulation studies, and the practical merit is demonstrated via a study on the association of distributions of brain activities between two brain regions.
		
	\end{abstract}
	
	\noindent%
	\vfill
	
	\newpage
	\spacingset{1} 
	\section{Introduction}\label{sec:intro}
	Thanks to rapid evolution of modern data collection  technologies, functional data emerge ubiquitously and the challenges of analyzing such data lead to a major line of research. For instance, various methodologies for multivariate data have been successfully extended to functional data, including functional principal components analysis (FPCA) \citep{yao2005jasa,hall2006}, linear regression  \citep{yao2005aos,hall2007,yuan2010,dou2012}, classification \citep{delaigle2012} and clustering \citep{james2003}. For a comprehensive treatment on functional data, we recommend the monographs  \cite{ramsay2005}, \cite{Ferraty2006}, \cite{Horvath2012}, \cite{hsing2015} and \cite{Kokoszka2017}. In addition, statistical analysis of functional data taking values in a nonlinear Riemannian manifold has gained increasing attention and been investigated by \cite{dai2018,lin2019,Dai2020} and \cite{lin2020}.
	
	In addition to manifold-valued data, probability distributions are nowadays commonly seen in practice, for example, arising from studies on mortality rates \citep{lin-metric-tv}, economics/housing \citep{chen2020}, healthcare \citep{lin-causal-inference}. The space of probability distributions defined in a common domain, referred to as Wasserstein space, is clearly not a linear space as a linear combination of two probability measures may not be a probability measure. In order to tackle the nonlinear structure of the Wasserstein space,  \cite{petersen2016} proposed a log quantile density (LQD) transformation to turn probability density functions to unconstrained functions. \cite{dai2021} adopted a square root transformation to map density functions into the positive orthant of a unit Hilbert sphere $\mathcal{S}^{\infty} $.  {However, none of these consider the more natural geometry that is compatible with optimal transport on the Wasserstein space. }    Since the Wasserstein space comes with a formal Riemannian structure \citep{ambrosio2008} that is compatible with optimal transport, it is natural to transform probability distributions via Riemannian logarithmic maps that have been well utilized \citep{lin2019,lin2020}. For instance, based on this idea, \cite{bigot2017} proposed a geodesic principal component analysis for data sampled from a Wasserstein space,  \cite{petersen2019bmka} studied Wasserstein covariance for multiple random densities, and \cite{chen2020} developed a class of regression models on Wasserstein space. 
	
	In this paper, we push further the frontier of statistical analysis on Wasserstein data into Wasserstein functional data that refer to functions taking values in a Wasserstein space. Such data, for example, could naturally arise from functional magnetic resonance imaging (fMRI) studies, where the distribution of brain signals in a region is longitudinally available for a period; see Figure \ref{fig:exa-WF} for an illustration. For statistical analysis of such data, in addition to the challenging issue of infinite dimensionality shared by the ordinary functional data analysis, a major challenge comes from the nonlinear nature of the Wasserstein space that creates  difficulties especially in modeling the covariance structure. Such nonlinearity is also presented in the Riemannian functional data analysis and is addressed by the intrinsic device of tensor Hilbert spaces proposed in \cite{lin2019}. However, although Wasserstein spaces have a geometric construction that is similar to the Riemannian structure, they are not Riemannian manifolds.  On one hand, when considering the measures defined on the real line, the Wasserstein space can be seen as the convex closed subset of $\mathcal{L}^{2}(0,1)$ formed by equivalence classes of quantile functions. Therefore, the  Fr\'echet mean can be expressed by the quantile functions and some regularity conditions on the Fr\'echet functional are no longer needed to ensure the existence and uniqueness of the empirical and population Fr\'echet mean. In addition, the  McCann’s interpolation defines a constant-speed geodesic on the Wasserstein space and the flatness property facilitates the theoretical analysis and asymptotic behavior of our proposed estimators.	On the other hand, since the tangent space at each point of Wasserstein space is an infinite-dimensional linear space, it is non-trivial to extend the framework in \cite{lin2019} proposed for (finite-dimensional) Riemannian manifolds to Wasserstein spaces. Given the aforementioned formal Riemannian structure of the Wasserstein space, we propose to circumvent the challenge by extending the device to Wasserstein functional data. 
		\begin{figure}[htbp]
		\centering
		\begin{minipage}[t]{0.48\textwidth }
			\centering
			\includegraphics[width=8.5cm]{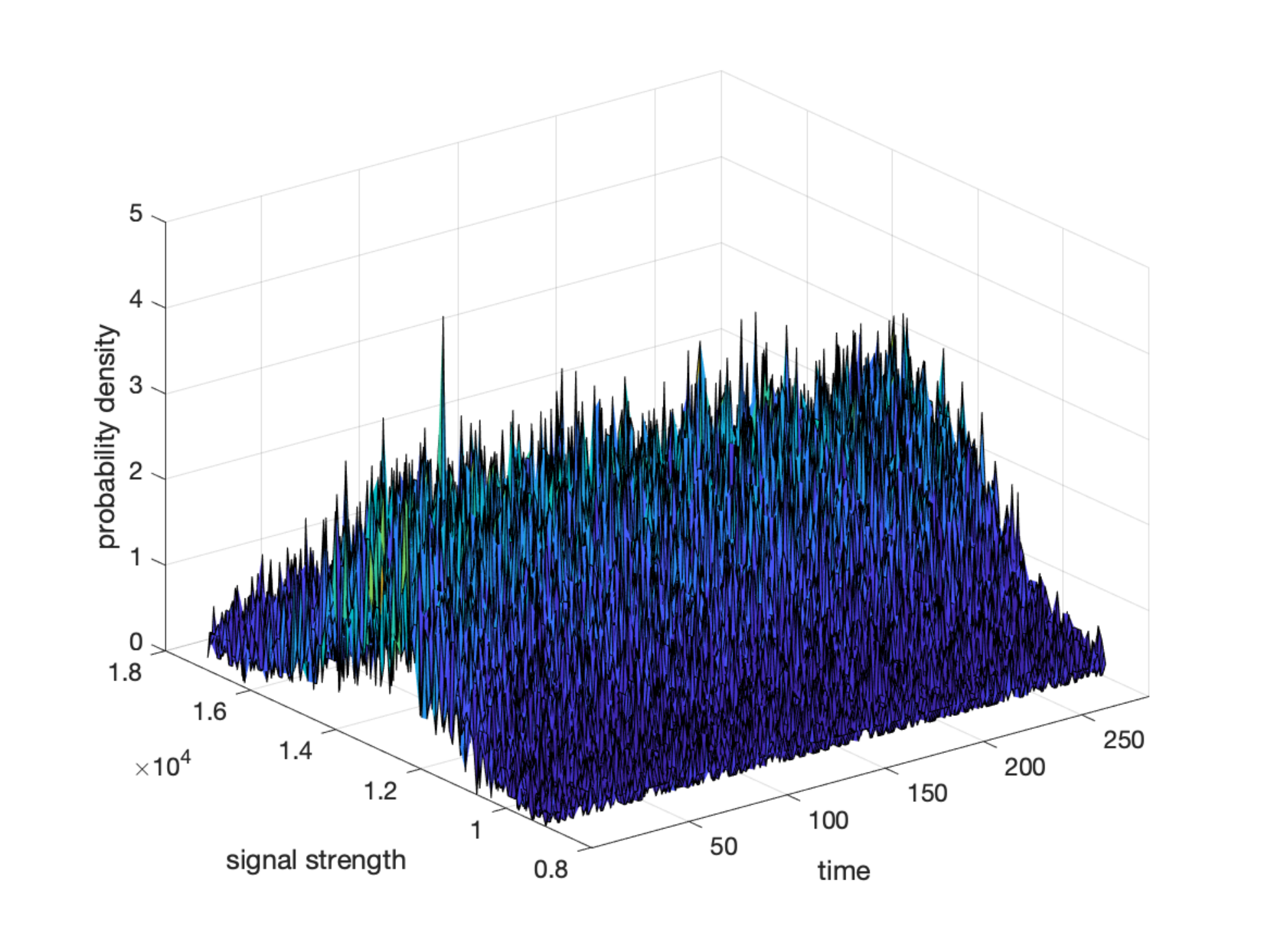}
		\end{minipage}
		\begin{minipage}[t]{0.48\textwidth }
			\centering
			\includegraphics[width=8.5cm]{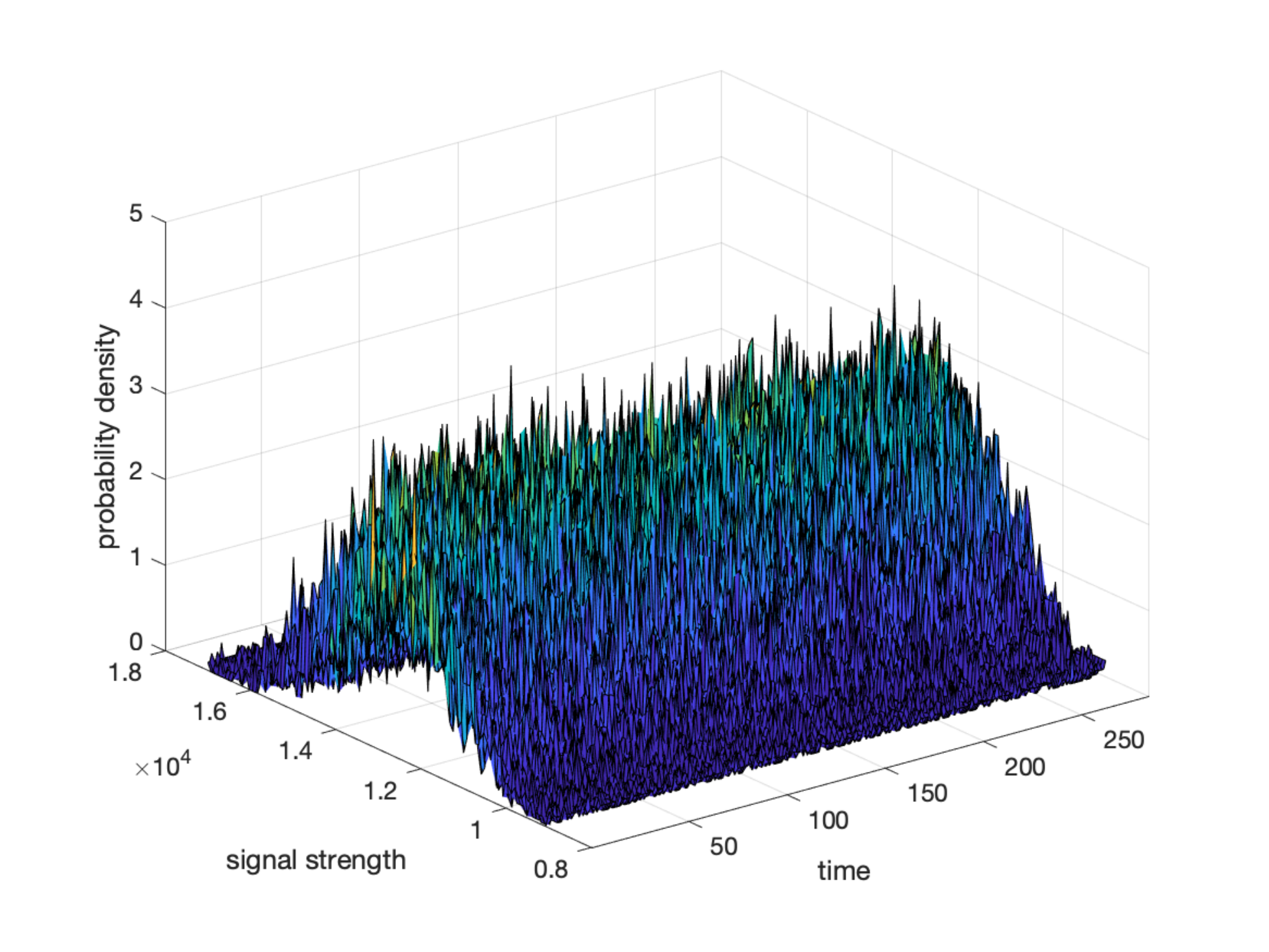}
		\end{minipage}
		\caption{The observed longitudinal densities of signal strength in Caudate nucleus (left) and Putamen area (right) from resting-state fMRI examinations.}
		\label{fig:exa-WF}
	\end{figure}

	Specifically, we investigate correlation analysis for Wasserstein functional data. Canonical correlation is one of the key tools for statistical analysis, and has been extensively studied for multivariate data and Euclidean functional data \citep{he2003,eubank2008,yang2011,lian2014, zhou2020}, but is yet to be explored for Wasserstein functional data. Our main  contribution is to formulate an intrinsic concept of correlation between random distributions, to propose an FPCA-based estimator and a Tikhonov regularized estimator, respectively, for the correlation and its corresponding weight functions, to establish the minimax convergence rates of the estimators, and to extend the framework of tensor Hilbert spaces to Wasserstein functional data in order to overcome the challenging issue raised by nonlinearity of Wasserstein spaces. In addition, our arguments for the minimax rate can  be straightforwardly extended to the  setting of Euclidean functional data, while existing works seem to lack of rigorous arguments; see Remark \ref{rem:minimax} in Appendix \ref{sec:prof-GCA} for details.
	

	The rest of the paper is organized as follows. We develop the foundational framework for Wasserstein space valued functional data in Section \ref{sec:WFDA}. Intrinsic Wasserstein correlation analysis is presented in Section \ref{sec:GCA}. The simulation studies are offered in Section \ref{sec:sim} and an application to fMRI dataset can be found in Section \ref{sec:data}. The proofs for the Theorems are collected together in the Appendix.

	\section{Wasserstein Functional Data}\label{sec:WFDA}
	
	In this section, we shall give a synopsis of the functional data valued in Wasserstein space. We first introduce the Wasserstein metric and geometry on a density class, which has similar structures to a Riemannian manifold. Based on the concept of tensor Hilbert space, we propose a new framework for functional data on Wasserstein space and discuss its properties. Then the  random elements perspective in tensor Hilbert space are investigated, including an estimation procedure for the mean surface and covariance operator. Finally, the asymptotic properties of the proposed estimators are studied. The proofs for the theorems in this section can be found in section \ref{sec:prof-WFDA}.
	
	\subsection{Wasserstein Metric and Geometry}\label{sec:intro-wsp}
	Let $\mathcal{W}_{2}(\reall)$ be the collection of probability measures on the real line $\reall$ with  finite second-order moments, that is,
	\begin{equation}\label{def:w2}
		\mathcal{W}_2(\reall)=\bigg\{\mu\in\mathcal{P}(\mathbb{R}): \int_\mathbb{R}|x|^{2} d\mu(x)<\infty\bigg\},
	\end{equation}
	where $\mathcal{P}(\mathbb{R}) $ is the set of probability measures on $\reall$. For $\mu\in\mathcal{W}_2(\reall)$ and a $\mu$-measurable map $T:\reall\rightarrow\reall$, the push-forward measure of $\mu$ through $T$ is defined by $(T\#\mu)(A)= \mu\{x \in \reall \mid T(x) \in A\}$, for $A\in \mathcal{B}(\reall)$, where $\mathcal B(\reall)$ denotes the Borel space of $\reall$. In the sequel, we use $F_{\mu}$ and $F^{-1}_{\mu}$ to denote the distribution function and the right-continuous quantile function of $\mu$, respectively. If $\mu$ is absolutely continuous to the Lebesgue measure, its density function is denoted by $f_{\mu}$. The Wasserstein distance between two measures ${\mu},{\nu}\in \mathcal{W}_{2}(\reall)$ is defined by 
	\begin{equation}\label{def:wass-dis}
		d({\mu},{\nu}):=\inf\bigg\{\int_{\reall^2}|x_{1}-x_{2}|^2\diff\gamma(x_1,x_2): \gamma\in\Gamma(\mu,\nu)\bigg\},
	\end{equation} 
	where $\Gamma(\mu,\nu)$ is the class of joint probability measures with marginal measures $\mu$ and $\nu$. The $\mathcal{W}_{2}(\reall)$ space endowed with the Wasserstein distance, denoted by the Wasserstein space $(\mathcal{W}_{2}(\reall),d) $, is a separable and complete metric space \citep{ambrosio2008}. 
	
	The minimization problem \eqref{def:wass-dis} is known as the Kantorovich’s formulation \citep{kantorovich2006} that is a relaxation of the Monge problem 
	\begin{equation}\label{def:monge}
		\inf\bigg\{\int_{\reall}\{T^{\ast}(u)-u\}^2\diff \mu(u),\text{ such that } T\#\mu=\nu\bigg\}.
	\end{equation}
	Unlike the  Kantorovich’s formulation, the Monge problem \eqref{def:monge} can be ill-posed when $\mu$ is a Dirac mass and $\nu$ has no atom, under which there is no transport map $T$ such that $T\#\mu=\nu$. If $\mu$ is absolutely continuous to Lebesgue measure, which implies $\mu$ has no atom and $F_{\mu}$ is continuous, the Monge problem \eqref{def:monge} is equivalent to \eqref{def:wass-dis} and has a unique solution  $T=F_{\nu}^{-1}\circ F_{\mu}:=T_{\mu}^{\nu}$ \citep{gangbo1996,petersen2019bmka}. This solution, called  the optimal transport map from $\mu$ to $\nu$, also induces a geodesic between $\mu$ and $\nu$;  here, a curve $\eta(t):I\rightarrow \mathcal W_2(\reall)$ parameterized by an interval $I\subset\reall$ is a geodesic if $d(\eta(t),\eta(t+\epsilon))=a\epsilon$ for a constant $a>0$ and for all $t\in I$ and all sufficiently small $\epsilon>0$. Specifically,  the McCann’s interpolation \citep{mccann1997}
	\begin{equation}\label{def:mcc}
		\mu_{t}=\left[\mathbf{id}+t\left(T_{\mu}^{\nu}-\mathbf{id}\right)\right] \# \mu: [0,1]\rightarrow \mathcal{W}_{2}(\reall)
	\end{equation}
	is a geodesic connecting $\mu$ to $\nu$ \citep{ambrosio2008},  where $\mathbf{id}$ denotes the identity map.
	
	Although $(\mathcal W_2(\reall),d)$ is not a Riemannian manifold, it can be endowed with a formal Riemannian structure, in which geometric concepts essential to statistical analysis, such as tangent spaces, Riemannian exponential maps and logarithmic maps can be defined \citep{bigot2017,chen2020}. This motivates us to extend the framework of \cite{lin2019} to the Wasserstein space, as follows. Define  
	\begin{equation}\label{def:Lp-spc}
		\begin{aligned}
			\mathcal{L}^{2}(\mu ; \reall):=&\bigg\{ T: \reall \rightarrow \reall\quad \mu \text {-measurable }: \int_{\reall }|T(x)|^2 \diff \mu(x)<+\infty\bigg\},
		\end{aligned}
	\end{equation}
	which is a separable Hilbert space for any probability measure $\mu$, where the inner product is given by $\langle T_{1},T_{2} \rangle_{\mu}=\int_{\reall}T_{1}(u)T_{2}(u)\diff \mu(u)$ for $T_{1},T_{2}\in \mathcal{L}^{2}(\mu ; \reall)$.
	The tangent space at $\mu$ is defined as the closure of $\tang_{\mu}^\circ=\left\{t\left(F^{-1}_{\nu}\circ F_{\mu} -\mathbf{id}\right): t>0, \nu \in \mathcal{W}_{2}(\mathbb{R})\right\}$ within ${\mathcal{L}^{2}(\mu;\mathbb{R})}$, i.e.,
	\begin{equation}\label{def:tan}
		\tang _{\mu}=\overline{\tang_{\mu}^\circ}^{\mathcal{L}^{2}(\mu;\mathbb{R})}.
	\end{equation}
	It follows from the definition that $\tang_{\mu}$ is a complete and separable subspace of  $\mathcal{L}^{2}(\mu,\reall)$. $\tang_{\mu} $ is also a linear space \citep[Chapter 2.3,][]{panaretos2020}, and thus a separable Hilbert space endowed with the inner product $\langle \cdot,\cdot \rangle_{\mu}$. The exponential map $\mathrm{Exp}_{\mu}:\rm{Tan}_{\mu}\rightarrow \mathcal{W}_{2}(\reall) $ and the corresponding logarithmic map $\rm{Log}_{\mu}: \mathcal{W}_{2}(\reall)\rightarrow\rm{Tan}_{\mu} $ at $\mu$ are defined by
	\begin{equation}\label{eq:logmap}
		\rm {Exp}_{\mu}(T)=(T+\mathbf{id})\#\mu \quad\text{ and }\quad \rm Log_{\mu}(\nu)=F_{\nu}^{-1}\circ F_{\mu}-\mathbf{id},
	\end{equation}
	respectively, 
	for $\rm T\in \rm Tan_{\mu}$ and $\mu,\nu \in \mathcal{W}_{2}(\reall)$. Unlike ordinary Riemannian exponential maps, the exponential map at $\mu$ defined above may not be a local homeomorphism between a neighborhood of $\mu$ and a neighborhood of the origin of $\rm{Tan}_\mu$, which shows that $\mathcal W_2(\reall)$ is not a genuine Riemannian manifold. Nevertheless, Theorem 2.2 in \cite{bigot2017} shows that the exponential map $\expm_{\mu}$ restricted to the image of the logarithmic map $\logm_{\mu}$ is an isometric homeomorphism with  $\logm_{\mu}$ being its inverse, and this is sufficient for our statistical analysis.
	
	As in most Wasserstein data the observed distributions have density functions supported in a compact domain $\wdomain$, in this paper we restrict our attention to $\dwsp$ that contains all probability measures  supported in $\wdomain$.
	
	\subsection{Distribution-valued Functional Data} 
	
	Let $X(t)$ be a random process indexed by $t\in\tdomain$ and taking values in $\dwsp$, where $\tdomain$ is a compact subset of $\reall$,  that is, for each $t\in\tdomain$, $X(t)$ is a random measure on $\wdomain$; in the statistical analysis, the random process $X$ serves as a prototype of the observed distribution-valued functional data. To quantify the first-order behaviors of $X$, as in \cite{dai2018,lin2019}, we utilize the concept of the Fr\'echet mean \citep{frechet1948}, which in our context is defined by 
	$$\mu(t)=\underset{p\in\dwsp}{\arg\min}F(p,t)$$
	with 
	\begin{equation}\label{def:intmeanF}
		F(p,t)=\mathbf{E}d^{2}(X(t),p), \quad p\in\dwsp,\,t\in\mathcal{T}.
	\end{equation}
	According to Lemma S2 of \cite{lin-causal-inference}, the Fr\'echet mean function exists and is unique for each $t\in\tdomain$.
	
	
	To characterize the second-order structure of $X$, we exploit the formal Riemannian structure of $\dwsp$ to consider the process $\Log_{\mu(t)}X(t)$ indexed by $t$ and taking values in the vector space $\Tan_{\mu(t)}$ for each $t\in\tdomain$. Like \cite{lin2019}, we treat the process $\Log_{\mu(\cdot)}X(\cdot)$, that is also denoted by $\Log_\mu X$ in the sequel for simplicity, as a vector field along the mean curve $\mu$, and then further view it as a random element in the space of vector fields along $\mu$,
	\begin{equation}\label{def:tensor}
		\mathscr{T}(\mu):=\left\{Z: Z(\cdot) \in \operatorname{Tan}_{\mu (\cdot)}, \int\langle Z(t), Z(t)\rangle_{\mu(t)} \mathrm{~d} t<\infty\right\}.
	\end{equation}
	For Riemannian functional data, i.e., data of functions taking value in a (finite-dimensional) Riemannian manifold, the space $\mathscr{T}(\mu)$, termed a tensor Hilbert space and shown by \cite{lin2019}, is a separable Hilbert space. Now we extend this result to the Wasserstein space $\dwsp$ by endowing $\mathscr{T}(\mu)$ with the inner product $\llangle Z_{1},Z_{2}\rrangle_{ \mu}:=\int \langle Z_{1}(t),Z_{2}(t) \rangle_{\mu(t)}\diff t$ and the induced norm $\|\cdot\|_{\mu}$, where we recall that  $\langle\cdot,\cdot\rangle_{\mu(t)}$ denotes the inner product in $\tang_{\mu(t)}$. 
	\begin{theorem}\label{thm:tensor}
		$\mathscr{T}(\mu)$ is a separable Hilbert space with the inner product $\llangle\cdot,\cdot\rrangle_{\mu}$.
	\end{theorem}
	
	
	Since $\mathcal{S}$ is compact, then $\dwsp$ is compact and thus $\mathbf{E}\left\|\logm_{\mu}X \right\|_{\mu}^2<\infty$, according to Theorem 7.4.2 in \cite{hsing2015}, $\logm_{\mu}X$ can be viewed as a random element in $\mathscr{T}(\mu)$. The auto-covariance operator $\mathbf{C}: \mathscr{T}(\mu) \rightarrow \mathscr{T}(\mu)$ for $X$ can be defined by
	\begin{equation}\label{def:covop}
		\llangle\mathbf{C} U, V\rrangle_{\mu}:=\mathbf{E}\left(\left\llangle\mathrm{Log} _{\mu} X, U\right\rrangle_{\mu}\left\llangle\mathrm{Log} _{\mu} X, V\right\rrangle_{\mu}\right),\quad \text { for } U, V \in \mathscr{T}(\mu).
	\end{equation}
	This operator is a nonnegative-definite trace-class operator with the following eigendecomposition \cite[Theorem 7.2.6,][]{hsing2015}
	\begin{equation}\label{def:covdec}
		\mathbf{C}=\sum_{k=1}^{\infty}\lambda_{k}\mathbf{\Phi}_{k}\otimes\mathbf{\Phi}_{k}
	\end{equation}
	with eigenvalues $\lambda_{1}>\lambda_{2}>\cdots>0$, that are assumed of multiplicity 1 without loss of generality, and orthonormal eigenelements $ \mathbf{\Phi}_{k}$ that form a complete orthonormal system for $\mathscr{T}(\mu) $. In addition, the  process $\mathrm{Log} _{\mu} X $ admits the following  Karhunen--Lo\`eve expansion 
	\begin{equation}\label{def:KLexp}
		\mathrm{Log}_{\mu}X= \sum_{k=1}^{\infty}\xi_{k}\mathbf{\Phi}_{k}
	\end{equation}
	with $\xi_{k}:=\llangle\mathrm{Log}_{\mu}X, \mathbf{\Phi}_{k} \rrangle_{\mu}$ being uncorrelated and centered random variables.

	Given a sample of independently and identically distributed (i.i.d.) copies $X_{1},\cdots,X_{n}$ of $X$,  the Fr\'echet mean function $\mu$ is estimated by its sample version 
	$$\hat{\mu}(t)=\underset{p\in\wsp}{\arg\min}\, F_{n}(p,t)$$
	with $$F_{n}(p,t)=\frac{1}{n}\sum_{i=1}^{n} d^{2}(X_{i}(t),p), \quad p\in\wsp,\,t\in\mathcal{T}. $$
	According to \cite{chen2020} and Lemma S2 of \cite{lin-causal-inference}, 
	\begin{equation}\label{def:muest}
		F_{\mu(t)}^{-1}=\mathbf{E}F_{X(t)}^{-1}\quad\text{ and }\quad F_{\hat{\mu}(t)}^{-1}=\frac{1}{n}\sum_{i=1}^{n}F_{X_{i}(t)}^{-1}\text{ for each }t\in\mathcal{T}.
	\end{equation}
	Similarly, the auto-covariance operator is estimated by its sample version
	\begin{equation}\label{def:covest}
		\hat{\mathbf{C}}=\frac{1}{n}\sum_{i=1}^{n} (\mathrm{Log}_{\hat\mu}X_{i})\otimes(\mathrm{Log}_{\hat\mu}X_{i}),
	\end{equation}
	which admits the eigendecomposion  $\hat{\mathbf{C}}=\sum_{k=1}^{\infty}\hat\lambda_{k}\mathbf{\hat\Phi}_{k}\otimes\mathbf{\hat\Phi}_{k} $ for $\hat\lambda_{1}>\hat\lambda_{2}>\cdots>0$ with the estimated eigenvalues $\hat\lambda_{k}$ and eigenfunctions $\mathbf{\hat\Phi}_{k}$. 
	
	To assess the estimation quality, for the Fr\'echet mean function, one may use the integrated squared error $\int d^2(\hat\mu(t),\mu(t))\diff t$. It turns out to be challenging to quantify the discrepancy between $\hat{\mathbf C}$ and $\mathbf C$, as when $\hat\mu$ and $\mu$ are not identical, the spaces $\mathscr{T}(\hat\mu)$ and $\mathscr{T}(\mu)$ are distinct Hilbert spaces. For Riemannian functional data, \cite{lin2019} addressed this problem by the parallel transport induced by the Levi--Civita connection that is intrinsic to the Riemannian manifold under consideration. Fortunately, as shown in \cite{ambrosio2008},   the Wasserstein space also has also a similar geometric structure that is defined via the Benamou--Brenier formula \citep{benamou2000}, which can be adopted  to $\dwsp$ as follows.

%
	First, we begin with the parallel transport of tangent vectors at an element of $\dwsp$ to another element. To this end, let $\nu$ and $\nu^\prime$ be elements in $\dwsp$,  a parallel transport operator can be defined between the entire Hilbert spaces $\mathcal{L}^{2}(\nu;\mathcal{S})$ and $\mathcal{L}^{2}(\nu^\prime;\mathcal{S})$ \citep{chen2020}, that is, $\mathrm{P}_{\nu}^{\nu^{\prime}}u:=u\circ F_{\nu}^{-1}\circ F_{\nu^{\prime}}$ for $u\in\mathcal{L}^{2}(\nu;\mathcal{S}) $, where $ F_{\nu}^{-1}$ and $F_{\nu^{\prime}} $ are the quantile function of $\nu$ and distribution function of $\nu'$. Assuming that $\nu$ is atomless, the parallel transport $\mathcal{P}_{\nu}^{\nu^{\prime}} $ from tangent space $\tang_{\nu}$ to $\tang_{\nu'}$ is defined by $\mathrm{P}_{\nu}^{\nu^{\prime}}$ restricted to $\tang_{\nu}$, i.e., $\mathcal{P}_{\nu}^{\nu^{\prime}}=\mathrm{P}_{\nu}^{\nu^{\prime}}|_{\tang_{\nu}} $.

%
	
	To extend the concept of parallel transport defined in \cite{chen2020} to tensor Hilbert spaces, let $\mu(\cdot),\mu'(\cdot),\nu(\cdot)$ and $\nu'(\cdot)$ be measurable curves on $\dwsp$. For $U\in \mathscr T(\mu)$, the parallel transport of $U$ from $\mu$ to $\mu^\prime$ is defined by $(\mathcal{P}_{\mu}^{\mu^\prime}U)(\cdot)=\mathcal{P}_{\mu(\cdot)}^{\mu^\prime(\cdot)}(U(\cdot))$. {Let $\mathfrak{B}(\mu,\nu)$ denotes the set of all bounded linear operators on $\mathscr{T}(\mu) $ mapping to $\mathscr{T}(\nu) $, which is a Banach space with the norm $|||\mathbf{A} |||_{\mathfrak{B}(\mu,\nu)}=\sup_{U\in\mathscr{T}(\mu),\|U\|_{\mu}=1}\|\mathbf{A}U\|_{\nu}$ \cite[Theorem 3.1.3]{hsing2015}.}  The operator $\mathcal{P}_{\mu}^{\nu}$ also gives rise to a mapping $\mathcal{P}_{\mathfrak{B}(\mu,\nu) }^{\mathfrak{B}(\mu',\nu')} $ from $\mathfrak{B}(\mu,\nu) $ to $\mathfrak{B}(\mu',\nu')$, defined by $\left(\mathcal{P}_{\mathfrak{B}(\mu,\nu) }^{\mathfrak{B}(\mu',\nu')}\mathbf{A}\right)V=\mathcal{P}_{\nu}^{\nu'}\mathbf{A}(\mathcal{P}_{\mu'}^{\mu}V ) $ for $\mathbf{A}\in\mathfrak{B}(\mu,\nu) $ and $V\in \mathscr{T}{(\mu')}$.
	\begin{prop}\label{prop:trans}
		Let  $\mu, \mu', \nu, \nu'$ be measurable curves on $\dwsp$ and assume for each $t\in\tdomain$, $\mu(t), \mu'(t), \nu(t), \nu'(t)$ are atomless; $U,U'\in \mathscr{T}(\mu)$, $V\in \mathscr{T}(\nu)$, $\mathbf{A}\in\mathfrak{B}(\mu,\nu)$, and $\mathbf{B}\in\mathfrak{B}(\mu',\nu')$.
		\begin{enumerate}[label=\textup{(\alph*)}]
			\item $\mathcal{P}_{\mu(t)}^{\nu(t)}$ is a unitary transportation form $\tang_{\mu(t)}$ to $\tang_{\nu(t)}$ and the adjoint operator of $\mathcal{P}_{\mu(t)}^{\nu(t)}$ is $\mathcal{P}_{\nu(t)}^{\mu(t)}$. 
			\item $\left\langle\mathcal{P}_{\mu(t)}^{\nu(t)} u,v\right\rangle_{\nu(t)} =\left\langle u,\mathcal{P}_{\nu(t)}^{\mu(t)} v\right\rangle_{\mu(t)}$ and $\left\|\mathcal{P}_{\mu(t)}^{\nu(t)} u-v \right\|_{\nu(t)}=\left\|u-\mathcal{P}_{\nu(t)}^{\mu(t)} v \right\|_{\mu(t)} $ for $u\in \tang_{\mu(t)}$ and $v\in\tang_{\nu(t)}$.
			\item $\llangle U, U'\rrangle_{\mu}=\left\llangle\mathcal{P}_{\mu}^{\nu} U,\mathcal{P}_{\mu}^{\nu} U'\right\rrangle_{\nu} $, $\left\llangle\mathcal{P}_{\mu}^{\nu} U,V\right\rrangle_{\nu} =\left\llangle U,\mathcal{P}_{\nu}^{\mu} V\right\rrangle_{\mu}$ and $\left\|\mathcal{P}_{\mu}^{\nu} U-V \right\|_{\nu}=\left\|U-\mathcal{P}_{\nu}^{\mu} V \right\|_{\mu} $.
			\item $\mathcal{P}_{\nu}^{\nu'}(\mathbf{A}U)=\left(\mathcal{P}_{\mathfrak{B}(\mu,\nu) }^{\mathfrak{B}(\mu',\nu')}\mathbf{A}\right)(\mathcal{P}_{\mu}^{\mu'} U) $.
			\item $\lnorm\mathcal{P}_{\mathfrak{B}(\mu,\nu) }^{\mathfrak{B}(\mu',\nu')} \mathbf{A}-\mathbf{B}\rnorm_{\mathfrak{B}(\mu',\nu') }=\lnorm\mathbf{A}-\mathcal{P}_{\mathfrak{B}(\mu',\nu') }^{\mathfrak{B}(\mu,\nu)} \mathbf{B}\rnorm_{\mathfrak{B}(\mu,\nu) }$.
			\item $\mathcal{P}_{\mathfrak{B}(\mu,\nu)}^{\mathfrak{B}(\mu',\nu') }\sum_{k}c_{k}\mathbf{\Phi}_{\mu,k}\otimes\mathbf{\Phi}_{\nu,k}=\sum_{k }c_{k}\left(\mathcal{P}_{\mu }^{\mu'}\mathbf{\Phi}_{\mu,k}\right)\otimes \left(\mathcal{P}_{\nu }^{\nu'}\mathbf{\Phi}_{\nu,k}\right) $, where $c_{k}$ are scalar constants, $\mathbf{\Phi}_{\mu,k}\in \mathscr{T}(\mu)$, and $\mathbf{\Phi}_{\nu,k}\in \mathscr{T}(\nu) $.
		\end{enumerate}
	\end{prop}

	In the above we exploit  the geometry of the Wasserstein space to develop the parallel transport of elements  and linear operators on tensor Hilbert spaces. This contrasts with the mechanism adopted in \cite{lin2019} for Riemannian manifolds in which the Levi--Civita connection exists and can be leveraged. The first two statements are direct results of Proposition 1 in \cite{chen2020} and the remaining assertions can be checked by  similar arguments in \cite{lin2019}; thus we omit the proof.

Now we are ready to quantify the discrepancy between objects of the same kind in $\mathscr{T}(\mu) $ and $\mathscr{T}(\hat\mu) $ by utilizing the above parallel transport.
\begin{theorem}\label{thm:meancov}
	Assume for each $t\in\tdomain$, $\mu(t)\in\wsp$ is atomless.  
	\begin{enumerate}[label=\textup{(\alph*)}]
		\item $\sqrt{n}\mathrm{Log}_{\mu }\hat{\mu} $ converges in distribution to a Gaussian measure on the tensor Hilbert space $\mathscr{T}(\mu ) $.
		\item $\sup_{t\in\mathcal{T}}d^{2}(\mu (t),\hat\mu (t))=O_{p}(n^{-1})$ and $\int_{\mathcal{T}}d^{2}(\mu (t),\hat\mu (t))\mathrm{d}t=O_{p}(n^{-1})$.
		\item $\left|\left|\left|\mathcal{P}_{\mathfrak{B(\hat\mu ,\hat\mu )} }^{\mathfrak{B(\mu ,\mu )} }\mathbf{\hat C} -\mathbf{C} \right|\right|\right|^2_{\mathfrak{B(\mu ,\mu )}}=O_{p}(n^{-1})$ and $\sup_{k\geqslant1}|\hat\lambda_{k}-\lambda_{k} |= O_{p}(n^{-1})$. 
		\item Let $\eta_{j}=(1/2)\inf_{j\neq k }|\lambda_{k}-\lambda_{j} |$ and $\Delta =\mathcal{P}_{\mathfrak{B(\hat\mu ,\hat\mu )} }^{\mathfrak{B(\mu ,\mu )} }\mathbf{\hat C} -\mathbf{C} $. If $\lambda_{j}\sim j^{-a }$ and $\eta_{j}\sim j^{-(a +1)}$, then for all $j$ such that $|||\Delta |||_{\mathfrak{B(\mu ,\mu )}}<\eta_j/2$ and a constant $C$, we have
		$$\mathbf{E}\left\|\mathcal{P}_{\hat\mu }^{\mu }\mathbf{\hat\Phi} _{j}-\mathbf{\Phi}_{j} \right\|_{\mathscr{T}(\mu )}^2\leqslant C\frac{j^{2}}{n}. $$
	\end{enumerate}
\end{theorem}

Part (a), (b) and (c) in Theorem \ref{thm:meancov} show that the convergence rates for the mean and covariance estimators are root-$n$, which is consistent with the classic results for fully observed Euclidean \citep{hall2006,hall2007} and Riemannian manifold functional data \citep{lin2019}. In view of the  properties for $\dwsp$,  some regularity conditions for the Fr\'{e}chet functional  are no longer needed to ensure the existence of the population and empirical Fr\'{e}chet mean.  Furthermore, due to the flatness of the Wasserstein space, the high order terms in the Taylor expansion are  vanished, which facilitates the theoretical derivation. 
   The last statement of Theorem \ref{thm:meancov} extends the classic result in eigenfunctions for fully observed Euclidean functional data, which is essential in most FPCA-based methods, especially related to  regression problems \citep{hall2006,hall2007,dou2012}. We stress that the component number $j$ is not fixed and could diverge slowly   with $n$, and this rate is optimal in the minimax sense \citep{wahl2020}.


	\section{Intrinsic Wasserstein Correlation Analysis}\label{sec:GCA}
	With the preparation of the groundwork for Wasserstein functional data, we are ready to discuss the correlation analysis between two sets of Wasserstein functional data.
	\subsection{Wasserstein Correlation}
	Let $X$ and $Y$ be two $\dwsp$-valued random processes with mean functions $\mu_{X},\mu_{Y}$ and auto-covariance operators $\mathbf{C}_{X},\mathbf{C}_{Y}$, respectively.
	The cross-covariance operator $\mathbf{C}_{XY}:\mathscr{T}(\mu_{Y})\longmapsto\mathscr{T}(\mu_{X})$ for $X $ and $Y$ is defined as 
	\begin{equation}\label{def:crosscov}
		\llangle\mathbf{C}_{XY} V, U\rrangle_{\mu_{X}}:=\mathbf{E}\left(\left\llangle\mathrm{Log} _{\mu_{Y}} Y, V\right\rrangle_{\mu_{Y}}\left\llangle\mathrm{Log} _{\mu_{X}} X, U\right\rrangle_{\mu_{X}}\right) \quad \text { for } V\in\mathscr{T}(\mu_{Y}) , U \in \mathscr{T}(\mu_{X}),
	\end{equation}
and $\mathbf{C}_{YX}$ is defined analogously. 
	We then define the intrinsic Wasserstein correlation between $X$ and $Y$ as
	\begin{equation}\label{def:iWCA}
		\rho=\max_{U\in\mathscr{T}(\mu_{X}):\langle U,\mathbf{C}_{X}U
			\rangle_{\mu_{X}}=1\atop V\in\mathscr{T}(\mu_{Y}):\langle V,\mathbf{C}_{Y}V
			\rangle_{\mu_{Y}}=1 } \langle U,\mathbf{C}_{XY}V \rangle_{\mu_{X}},
	\end{equation}
which generalizes the canonical correlation for classic multivariate data and Euclidean functional data \citep{he2003,lian2014}. Unlike the latter two types of data, functions valued in Wasserstein space  are nonlinear, and such nonlinearity is overcome by the device of tensor Hilbert space introduced in Section \ref{sec:WFDA}. Note that  the maximization problem \eqref{def:iWCA} is equivalent to finding $U$ in $\mathscr{T}(\mu_{X})$ and $V$ in  $\mathscr{T}(\mu_{Y})$ to maximize the correlation between $\llangle U, \logm_{\mu_{X}}X \rrangle_{\mu_{X}}$ and $\llangle V, \logm_{\mu_{Y}}Y \rrangle_{\mu_{Y}}$. 

In the case of multivariate data, the solution to canonical correlation analysis is reduced to singular value decomposition of $\mathbf{C}_{X}^{-1/2}\mathbf{C}_{XY}\mathbf{C}_{Y}^{-1/2}$, which could not be applied to functional data because both $\mathbf C_X$ and $\mathbf C_Y$ are infinite-dimensional compact operators and thus have an invertibility issue. By Theorem 7.2.10 in \cite{hsing2015}, there exists an operator $\mathbf{R}_{XY}\in \mathfrak{B}(\mu_{Y},\mu_{X} )$ with $\lnorm\mathbf{R}_{XY} \rnorm\leqslant 1 $ such that $\mathbf{C}_{XY}=\mathbf{C}_{X}^{1/2}\mathbf{R}_{XY} \mathbf{C}_{Y}^{1/2}$. This suggests that the operator $\mathbf{C}_{X}^{-1/2} \mathbf{C}_{XY}\mathbf{C}_{Y}^{-1/2}$ is definable on the range of $\mathbf{C}_{Y}^{1/2}$. To formulate this idea and link it to the optimization problem \eqref{def:iWCA}, we first  recall the Karhunen--Lo$\grave{\mathrm{e}}$ve expansions for $X$ and $Y$, given by
$$\mathrm{Log}_{\mu_{X}}X= \sum_{k=1}^{\infty}\xi_{k}\mathbf{\Phi}_{X,k},\quad\mathrm{Log}_{\mu_{Y}}Y= \sum_{k=1}^{\infty}\eta_{k}\mathbf{\Phi}_{Y,k},$$
where $\{\mathbf{\Phi}_{X,k}\}_{k=1}^{\infty} $ and $\{\mathbf{\Phi}_{Y,k}\}_{k=1}^{\infty} $ are respectively the eigenbases of $\mathbf C_X$ and $\mathbf C_Y$, and $\xi_{j},\eta_{j}$ are principal component scores respectively with variance $\lambda_{X,j}$ and $\lambda_{Y,j}$. It is then seen that  $\mathbf{C}_{YX}$ and $\mathbf{C}_{XY}$ can be expressed as 
\begin{align*}
	\mathbf{C}_{YX}  =\sum_{j_{1}=1}^{\infty}\sum_{j_{2}=1}^{\infty}\gamma_{j_{1}j_{2}}\mathbf{\Phi}_{X,j_{1}}\otimes\mathbf{\Phi}_{Y,j_{2}}, \quad
	\mathbf{C}_{XY}  =\sum_{j_{1}=1}^{\infty}\sum_{j_{2}=1}^{\infty}\gamma_{j_{1}j_{2}}\mathbf{\Phi}_{Y,j_{2}}\otimes\mathbf{\Phi}_{X,j_{1}}
\end{align*}
with $\gamma_{j_{1}j_{2}}=\mathbf{E}\{\xi_{j_{1}}\eta_{j_{2}}\}$.
Now we impose the following assumption on the interplay among $\gamma_{j_1j_2}$, $\lambda_{X,j_1}$ and $\lambda_{Y,j_2}$.
	\begin{assum}[\textbf{B.0}]\label{asm:b0}
		$\sum_{j_{1},j_{2}}^{\infty}\frac{ \gamma_{j_{1}j_{2}}^2}{\lambda_{X,j_{1}}^2\lambda_{Y,j_{2}} }<\infty $ and $\sum_{j_{1},j_{2}}^{\infty}\frac{ \gamma_{j_{1}j_{2}}^2}{\lambda_{X,j_{1}}\lambda_{Y,j_{2}}^2 }<\infty $.
	\end{assum}

The above assumption is the same as Condition 4.5 in \cite{he2003} and requires that cross-covariance operator of $X$ and $Y$  be  aligned with the eigenfunctions of $\mathrm{Log}_{\mu_{X}}X$ and $\mathrm{Log}_{\mu_{Y}}Y$; such a requirement is  commonly adopted in FPCA-based functional regression models \citep{hall2007,dou2012}. Under this assumption, the following proposition, inspiring estimators that are proposed in the next section,  asserts the boundedness of the operator $\mathbf{C}_{X}^{-1/2} \mathbf{C}_{XY}\mathbf{C}_{Y}^{-1/2}$ and provides a solution to the maximization problem  \eqref{def:iWCA}. The proposition can be  checked by arguments similar to that of \cite{he2003} and \cite{lian2014}, thus  its proof is omitted.

\begin{prop}\label{prop:exist-GCA}
	Under Assumption \hyperref[asm:b0]{\textup{B.0}}, $\mathbf{C}_{X}^{-1/2} \mathbf{C}_{XY}\mathbf{C}_{Y}^{-1/2}$ and $\mathbf{C}_{X}^{-1} \mathbf{C}_{XY}\mathbf{C}_{Y}^{-1/2}$ are Hilbert--Schmidt operators defined on $\mathscr{T}(\mu_{Y})$. The maximum in \eqref{def:iWCA} is achieved for the weight functions $U\in\mathscr{T}(\mu_{X}) $ and $V\in\mathscr{T}(\mu_{Y}) $ with maximum $\rho=\sqrt{\alpha}$, where $(\alpha,U)$ is the first eigenpair of $\mathbf{C}_{X}^{-1} \mathbf{C}_{XY} \mathbf{C}_{Y}^{-1} \mathbf{C}_{YX} $, and $V= \mathbf{C}_{Y}^{-1} \mathbf{C}_{YX}U/\|\mathbf{C}_{Y}^{-1/2} \mathbf{C}_{YX}U \|_{\mu_{Y}}$.	
\end{prop}

	\subsection{Estimation and Theoretical Properties}
	Given a random sample of functions $\{(X_{i},Y_{i})\}_{i=1}^{n}$ of $(X,Y)$,  the mean functions for $X$ and $Y$ are estimated by $\hat\mu_X$ and $\hat\mu_Y$ that are respectively represented by 
	$$
	F_{\hat{\mu}_{X}(t)}^{-1}=\frac{1}{n}\sum_{i=1}^{n}F_{X_{i}(t)}^{-1}\text{ and }F_{\hat{\mu}_{Y}(t)}^{-1}=\frac{1}{n}\sum_{i=1}^{n}F_{Y_{i}(t)}^{-1}\text{ for each }t\in\mathcal{T}. 
	$$
	The estimators for autocovariance operators of $X$ and $Y$ are 
	$$
	\hat{\mathbf{C}}_{X}=\frac{1}{n}\sum_{j=1}^{n} (\mathrm{Log}_{\hat\mu_{X}}X_{i})\otimes(\mathrm{Log}_{\hat\mu_{X}}X_{i})\text{ and }	\hat{\mathbf{C}}_{Y}=\frac{1}{n}\sum_{j=1}^{n} (\mathrm{Log}_{\hat\mu_{Y}}Y_{i})\otimes(\mathrm{Log}_{\hat\mu_{Y}}Y_{i}), 
	$$
	respectively. In addition, they admit the following decompositions,
	\begin{equation}\label{def:sampleig}
		\hat{\mathbf{C}}_{X}=\sum_{k=1}^{\infty}\hat\lambda_{X,k}\mathbf{\hat\Phi}_{X,k}\otimes\mathbf{\hat\Phi}_{X,k} \text{ and }\hat{\mathbf{C}}_{Y}=\sum_{k=1}^{\infty}\hat\lambda_{Y,k}\mathbf{\hat\Phi}_{Y,k}\otimes\mathbf{\hat\Phi}_{Y,k}, 
	\end{equation}
	where $(\hat\lambda_{X,k}\mathbf{\hat\Phi}_{X,k}) $ and $(\hat\lambda_{Y,k},\mathbf{\hat\Phi}_{Y,k}) $ serve as estimators for $(\lambda_{X,k},\mathbf{\Phi}_{X,k}) $ and  $(\lambda_{Y,k},\mathbf{\Phi}_{Y,k}) $, respectively. Similarly,  the cross covariance operators between $X$ and $Y$ are estimated by
	$$\hat{\mathbf{C}}_{YX}=\frac{1}{n}\sum_{i=1}^{n} (\mathrm{Log}_{\hat\mu_{X}}X_{i} )\otimes(\mathrm{Log}_{\hat\mu_{Y}}Y_{i} )\text{ and }\hat{\mathbf{C}}_{XY}=\frac{1}{n}\sum_{i=1}^{n} (\mathrm{Log}_{\hat\mu_{Y}}Y_{i} )\otimes(\mathrm{Log}_{\hat\mu_{X}}X_{i} ). $$
	
	Based on the above estimators, we construct an estimator of  $\mathbf{C}_{X}^{-1} \mathbf{C}_{XY} \mathbf{C}_{Y}^{-1} \mathbf{C}_{YX}$, denoted   by  $\hat{\mathbf{C}}_{X,k_{X}}^{-1}\hat{\mathbf{C}}_{XY}\hat{\mathbf{C}}_{Y,k_{Y}}^{-1}\hat{\mathbf{C}}_{YX}$, where $\hat{\mathbf{C}}_{X,k_{X}}^{-1}=\sum_{j=1}^{k_{X}}\hat\lambda_{X,j}^{-1}\hat{\mathbf{\Phi}}_{X,j}$, $\hat{\mathbf{C}}_{Y,k_{Y}}^{-1}=\sum_{j=1}^{k_{Y}}\hat\lambda_{Y,j}^{-1}\hat{\mathbf{\Phi}}_{Y,j}$, and $k_{X},k_{Y}\in \mathbb{N}^{+}$ are two tuning parameters. The truncation of $\hat{\mathbf{C}}_{X,k_{X}}^{-1}$ and $\hat{\mathbf{C}}_{X,k_{X}}^{-1}$ at respectively finite levels $k_X$ and $k_Y$ serves as a way of regularization that is  needed to address the invertibility issue of infinite-dimensional compact operators \citep{yao2005aos,hall2007,dou2012}. Then, the estimator of $U$, denoted by $\hat{U}$, is the eigenfunction of $\hat{\mathbf{C}}_{X,k_{X}}^{-1}\hat{\mathbf{C}}_{XY}\hat{\mathbf{C}}_{Y,k_{Y}}^{-1}\hat{\mathbf{C}}_{YX} $ associated with its largest eigenvalue $\hat\alpha$, and the estimators for $V,\rho$ are defined by $\hat{V}= \hat{\mathbf{C}}_{Y,k_{Y}}^{-1} \hat{\mathbf{C}}_{YX}\hat{U}/\|\hat{\mathbf{C}}_{Y,k_{Y}}^{-1/2} \hat{\mathbf{C}}_{YX}\hat{U} \|_{\hat\mu_{Y}},\hat\rho=\sqrt{\hat\alpha}$, respectively.
	
	Alternatively, we may utilize Tikhonov regularization \citep{hall2007} to estimate $\mathbf{C}_{X}^{-1} \mathbf{C}_{XY} \mathbf{C}_{Y}^{-1} \mathbf{C}_{YX}$ by  $(\hat{\mathbf{C}}_{X}+\epsilon_{X}\hat{\mathbf{id}}_{X})^{-1}\hat{\mathbf{C}}_{XY}(\hat{\mathbf{C}}_{Y}+\epsilon_{Y}\hat{\mathbf{id}}_{Y})^{-1}\hat{\mathbf{C}}_{YX}$, where $\epsilon_{X},\epsilon_{Y}$ are positive tuning parameters and $ \hat{\mathbf{id}}_{X},\hat{\mathbf{id}}_{Y}$ are the identity operators, respectively, on $\mathscr{T}(\hat\mu_{X})$ and $\mathscr{T}(\hat\mu_{Y})$. Then $U$ is estimated by the eigenfunction $\tilde{U}$ of $(\hat{\mathbf{C}}_{X}+\epsilon_{X}\hat{\mathbf{id}}_{X})^{-1}\hat{\mathbf{C}}_{XY}(\hat{\mathbf{C}}_{Y}+\epsilon_{Y}\hat{\mathbf{id}}_{Y})^{-1}\hat{\mathbf{C}}_{YX}$ associated with the largest eigenvalue $\tilde{\alpha}$,  $V$ is estimated by $\tilde{V}= (\hat{\mathbf{C}}_{Y}+\epsilon_{Y}\hat{\mathbf{id}}_{Y})^{-1}\hat{\mathbf{C}}_{YX}\tilde{U}/\|(\hat{\mathbf{C}}_{Y}+\epsilon_{Y}\hat{\mathbf{id}}_{Y})^{-1/2}\hat{\mathbf{C}}_{YX}\tilde{U}\|_{\hat\mu_{Y}}$, and $\rho$ is estimated by $\tilde\rho=\sqrt{\tilde\alpha}$. 
	
	To study the theoretical properties of the estimators $\hat U$ and $\hat V$, we require the following assumption to  utilize the parallel transportation operators defined in Section \ref{sec:WFDA}.
	\begin{assum}[\textbf{A.1}]\label{asm:a1}
		For each $t\in\tdomain$, $\mu_{X}(t)$ and $\mu_{Y}(t) $ are atomless.
	\end{assum}
	In addition, we assume the eigenspace for the largest eigenvalue of $\mathbf{C}_{X}^{-1} \mathbf{C}_{XY} \mathbf{C}_{Y}^{-1} \mathbf{C}_{YX} $ has multiplicity one to ensure the uniqueness of the solution to \eqref{def:iWCA}. We also make the following assumptions. 
	\begin{assum}[$\textbf{A.2}$ ]\label{asm:a2}
		$\mathbf{E}\xi_{j}^4\leqslant C\lambda_{X,j}^{2}$ and $\mathbf{E}\eta_{j}^4\leqslant C\eta_{Y,j}^{2}$ for a positive constant $C$.
	\end{assum}
	
	\begin{assum}[$\textbf{B.1}$ ]\label{asm:b1}
		There exist positive constants $c,C$ such that $Cj^{-a_{X}}\geqslant\lambda_{X,j}\geqslant\lambda_{X,j+1}+cj^{-a_{X}-1}$ and $Cj^{-a_{Y}}\geqslant\lambda_{Y,j}\geqslant\lambda_{Y,j+1}+cj^{-a_{Y}-1}$ for some $a_{X},a_{Y}>1$ and each $j\geqslant1$. 
	\end{assum}
	\begin{assum}[$\textbf{B.2}$ ]\label{asm:b2}
		$\sum_{j_{2}}\gamma_{j_{1}j_{2}}^2\leqslant C j_{1}^{-2a_{X}-2b_{X}},\sum_{j_{1}}\gamma_{j_{1}j_{2}}^2\leqslant C j_{2}^{-2a_{Y}-2b_{Y}} , k_{X}\asymp n^{1/(a_{X}+2b_{X})}, k_{Y}\asymp n^{1/(a_{Y}+2b_{Y})} $, for some constants $b_{X}>a_{X}/2+1,b_{Y}>a_{Y}/2+1$ and a positive constant $C$.
	\end{assum}
	Assumption \hyperref[asm:a2]{$\textup{A.2}$} is common in the FPCA literature \citep{hall2007}. Assumptions \hyperref[asm:b1]{B.1} and \hyperref[asm:b2]{B.2} define a class of random processes $X$ and $Y$ for which we are able to establish the minimax rate for the proposed estimators; similar conditions  have been adopted in the literature of functional linear regression \citep{hall2007,dou2012}. The following theorem presents an upper bound on the convergence rate of estimated weight functions $\hat U$ and $\hat V$.
	
	\begin{theorem}\label{thm:fpc}
		Under assumptions \hyperref[asm:a1]{\textup{A.1}}, \hyperref[asm:a2]{\textup{A.2}}, \hyperref[asm:b1]{\textup{B.1}} and \hyperref[asm:b2]{\textup{B.2}}, we have 
		$$\|\mathcal{P}_{\hat\mu_{X}}^{\mu_{X}} \hat{U}-U\|_{\mu_{X}}^2+ \|\mathcal{P}_{\hat\mu_{Y}}^{\mu_{Y}} \hat{V}-V\|_{\mu_{Y}}^2=O_{p}\left(\max\left\{n^{-(2b_{X}-1)/(a_{X}+2b_{X})},n^{-(2b_{Y}-1)/(a_{Y}+2b_{Y})} \right\}\right).$$
	\end{theorem}
	
	The convergence rate in Theorem \ref{thm:fpc} is in accordance to  the rate in classic functional regression problems \citep{hall2007,yuan2010,dou2012}, as well as the rate for non-functional Wasserstein regression \citep{chen2020}. This is not surprising, since canonical correlation analysis is intimately related to two regression problems, in our context, one in which  $\logm_{\mu_{X}}X$ is regressed on $\logm_{\mu_{Y}}Y$ and the other in which $\logm_{\mu_{Y}}Y$ is  regressed on $\logm_{\mu_{X}}X$. This is slightly different from the attained rate of the functional linear regression involving the Riemannian manifold \citep{lin2019,lin2021}, as the nonlinear structure does not affect the convergence rate due to the flatness of the geodesic in Wasserstein spaces. 
	
	To study the asymptotic properties of the Tikhonov estimators $\tilde U$ and $\tilde V$, we require the following assumptions.
	\begin{assum}[$\textbf{B.1}^{'}$ ]\label{asm:b1'}
		There exist positive constants $C$, $a_X$ and $a_Y$ such that $\lambda_{X,j}\leqslant Cj^{-a_{X}}$ and $\lambda_{Y,j}\leqslant Cj^{-a_{Y}}$ for all $j\geqslant1$. 
	\end{assum}
	\begin{assum}[$\textbf{B.2}^{'}$ ]\label{asm:b2'}
		For some constants  $b_{X}>a_{X}-1/2,b_{Y}>a_{Y}-1/2$, 
		$\sum_{j_{2}}\gamma_{j_{1}j_{2}}^2\leqslant C j_{1}^{-2a_{X}-2b_{X}}$, $\sum_{j_{1}}\gamma_{j_{1}j_{2}}^2\leqslant C j_{2}^{-2a_{Y}-2b_{Y}}$, $\epsilon_{X}\asymp n^{-a_{X}/(a_{X}+2b_{X})}$, and $\epsilon_{Y}\asymp n^{-a_{Y}/(a_{Y}+2b_{Y})} $.
	\end{assum}
The following theorem provides an upper bound on the convergence rate of both  $\tilde{U}$ and $\tilde V$. 
	\begin{theorem}\label{thm:ridge}
		Under assumptions \hyperref[asm:a1]{\textup{A.1}}, \hyperref[asm:a2]{\textup{A.2}}, \hyperref[asm:b1']{$\textup{B.1}^{'}$} and \hyperref[asm:b1']{$\textup{B.2}^{'}$}, we have 
		$$\|\mathcal{P}_{\hat\mu_{X}}^{\mu_{X}} \tilde{U}-U\|_{\mu_{X}}^2+ \|\mathcal{P}_{\hat\mu_{Y}}^{\mu_{Y}} \tilde{V}-V\|_{\mu_{Y}}^2=O_{p}\left(\max\left\{n^{-(2b_{X}-1)/(a_{X}+2b_{X})},n^{-(2b_{Y}-1)/(a_{Y}+2b_{Y})} \right\}\right).$$
	\end{theorem}
	
	Now we show that the bounds in Theorems \ref{thm:fpc} and \ref{thm:ridge} are tight. To this end, for two mean surfaces $\mu_{X}$ and $\mu_{Y}$,  recall that $\logm_{\mu_{X}}X $ and $\logm_{\mu_{Y}}Y $ admit the expansions $\mathrm{Log}_{\mu_{X}}X=\sum_{j}\xi_{j}\mathbf{\Phi}_{X,j} $ and $\mathrm{Log}_{\mu_{Y}}Y=\sum_{j}\eta_{j}\mathbf{\Phi}_{Y,j} $ for the orthogonal bases $\{\mathbf{\Phi}_{X,j}\}_{j=1}^{\infty}\subset\mathscr{T}(\mu_{X}) $ and $\{\mathbf{\Phi}_{Y,j}\}_{j=1}^{\infty}\subset\mathscr{T}(\mu_{Y}) $. Let ${P}_{XY}$ be the distribution of $(\logm_{\mu_{X}}X,\logm_{\mu_{Y}}Y)$ and define the family
	\begin{align*}
		\mathcal{F}(C,a,b):=&\left\{ P_{XY}:\sum_{i,j=1} ^{\infty}\frac{\gamma_{ij}^2}{\lambda_{X,j_{1}}^2\lambda_{Y,j_{2}}}\leqslant C,\sum_{i,j=1} ^{\infty}\frac{\gamma_{ij}^2}{\lambda_{X,j_{1}}\lambda_{Y,j_{2}}^2}\leqslant C, C^{-1}j^{-a}\leqslant \lambda_{X,j}\leqslant Cj^{-a}, \right.\\
		&\left.C^{-1}j^{-a}\leqslant \lambda_{Y,j}\leqslant Cj^{-a} , \sum_{j}\gamma_{ij}^2\leqslant Ci^{-2a-2b},\sum_{i}\gamma_{ij}^2\leqslant Cj^{-2a-2b}\right\}.
	\end{align*} 
The following result establishes  a lower bound on the convergence rate of an estimator of $(U,V)$ and can be derived by similar arguments in \cite{lian2014}; we omit the proof. The bound matches the upper bound in Theorems \ref{thm:fpc} and \ref{thm:ridge} and thus implies the optimality of the estimators $(\hat U,\hat V)$ and $(\tilde U,\tilde V)$ in the minimax sense.
	\begin{theorem}\label{thm:lowerbnd}
		Suppose that $(X_1,Y_1),\ldots (X_n,Y_n)$ form a random sample of $(X,Y)$ with $(\Log_{\mu_X}X,\Log_{\mu_Y}Y)$ following the distribution  ${P}_{XY}\in \mathcal{F}(C,a,b,\mu)$, and that $(U,V)$ maximizes \eqref{def:iWCA}. Then
		$$\lim_{c\rightarrow0}\liminf_{n\rightarrow\infty}\inf_{(\check{U},\check{V})}\sup_{{P}_{XY}\in \mathcal{F}(C,a,b) }P_{XY}\left(\|\check{U}-U\|^{2}_{\mu_{X}}+\|\check{V}-V\|^{2}_{\mu_{Y}} \geqslant cn^{-\frac{2b-1}{a+2b}}\right)=1, $$
		 where $(\hat{U},\hat{V})$ denotes an estimator of $(U,V)$ based on the data $(X_1,Y_1),\ldots (X_n,Y_n)$.
	\end{theorem}

	\section{Simulation Studies}\label{sec:sim}
To illustrate the numerical behavior of the proposed methods, we set $\mu_{X}(t)$ to the Beta distribution with parameters $(2+t,3-(t^2+t)/2)$ and $\mu_{Y}(t)$ to the Beta distribution with parameters $(3-t,2+ (t^2+t)/2)$ for each $t\in[0,1]$ as the mean surfaces. We consider the set of orthonormal functions
	$$
	\phi_{j}(x)=\sqrt{2} \sin ( \pi j x), \quad \text { for } x \in[0,1], \text { and } j \in \mathbb{N}_{+}.
	$$
	Then $\mathbf{\Phi}_{X,j}(x,t)=\phi_{j}\circ F_{X,t}(x)$ and $\mathbf{\Phi}_{Y,j}(x,t)=\phi_{j}\circ F_{Y,t}(x) $, $j=1,2,\ldots$, form an orthonormal basis for $\mathscr{T}(\mu_{X})$ and $\mathscr{T}(\mu_{Y})$, where $F_{X,t}$ and $F_{Y,t}$ denote the distribution function of $\mu_{X}(t)$ and $\mu_{Y}(t)$, respectively. 
	
Write $\mathrm{Log}_{\mu_{X}}X_{i}(t)=\sum_{j=1}^{\infty}\xi_{ij} \mathbf{\Phi}_{X,j}=\sum_{j=1}^{\infty}\xi_{ij}\phi_{j}\circ F_{X,t}(x) $, where $\xi_{ij}$ are uncorrelated random variables with zero mean such that $\sum_{j=1}^{\infty}\xi_{ij}^2<\infty$ almost surely. To guarantee $\sum_{j=1}^{\infty}\xi_{ij}\phi_{j}\circ F_{X,t}(x)\in \mathrm{Log}_{\mu_{X}}\dsp([0,1])$, where $\dsp([0,1])$ denotes the set of absolutely continuous measures on $[0,1]$, it suffices to require \begin{equation}\label{eq:sim-1}
		\sum_{j=1}^{\infty}\xi_{ij}\phi'_{j}(F_{X,t}(x) )f_{X,t} +1\geqslant0\text{ for all }x\in[0,1] \text{ and }t\in\mathcal{T} ,
	\end{equation}
	where $f_{X,t} $ is the density function of $F_{X,t}$.
	Condition \eqref{eq:sim-1} is satisfied, e.g., when $\xi_{ij}\leqslant v_{j}/(\sup_{x\in[0,1]}|\phi'_{j}(x) |\sup_{t,x\in[0,1]}f_{X,t}(x)\sum_{j=1}^{\infty}v_{j} )$, where $\{v_{j}\}_{j=1}^{\infty}$ is a non-negative sequence of constants such that $\sum_{j=1}^{\infty }v_{j}<\infty $, e.g., $v_j=a^{-j}$ for a given $a>1$.
	
	Taking $K=20$, $v_{j}=2^{-j}$, $V_{j}:=\sup_{x\in[0,1]}|\phi'_{j}(x)|=\sqrt{2}\pi j$ and $M=1.78> \sup_{x,t\in[0,1]}f_{X,t}(x)=\sup_{x,t\in[0,1]}f_{Y,t}(x) $, we set  $X_{i}(t)=\mathrm{Exp}_{\mu_{X}}(\sum_{k=1}^{K}\xi_{ik}\mathbf{\Phi}_{X,k} ) $ and $Y_{i}(t)=\mathrm{Exp}_{\mu_{Y}}(\sum_{k=1}^{K}\eta_{ik}\mathbf{\Phi}_{Y,k} ) $, considering the following two types of scores $\xi_{ij}$ and $\eta_{ij}$.
	\begin{itemize}
		\item Case 1 (Truncated normal): We sample $\xi_{ik}\sim v_{k}(V_{k}M)^{-1}\theta_{ik}$ with $\theta_{ik}\sim\text{TN}_{[-1,1]}(0,1)$ independently for $i=1,2,\cdots,n$ and $k=1,2,\cdots,K$, and $\eta_{ik}\sim v_{k}(V_{k}M)^{-1}\vartheta_{ik}$ with $\vartheta_{ik}\sim\text{TN}_{[-1,1]}(0,1)$, except that $\eta_{i2}=0.5(\xi_{i1}+\xi_{i2} )+\sigma\mathbf{E}(\xi_{i1}^2+\xi_{i2}^2)\vartheta_{i2}$,  where $\text{TN}_{[-1,1]}(0,1) $ denotes the Gaussian distribution $N(0,1)$ truncated on $[-1,1]$, and $\sigma$ is a constant representing the noise level. 
		\item Case 2 (Uniform): We sample $\xi_{ik}\sim \text{Unif}[-v_{k}(V_{k}M )^{-1},v_{k}(V_{k}M )^{-1} ]$ independently for $i=1,2,\cdots,n$ and $k=1,2,\cdots,K$, and $\eta_{ik}\sim \text{Unif}[-v_{k}(V_{k}M )^{-1},v_{k}(V_{k}M )^{-1} ]$, except that $\eta_{i2}=0.5(\xi_{i1}+\xi_{i2} )+\sigma\mathbf{E}(\xi_{i1}^2+\xi_{i2}^2)\vartheta_{i2}$ with $\vartheta_{i2}\sim\text{Unif}[-1,1]$. 
	\end{itemize}
	In this construction we have $U\propto \mathbf{\Phi}_{X,1}+\mathbf{\Phi}_{X,2} $, $V\propto\mathbf{\Phi}_{Y,2}$ and $\rho^2=0.5^2/(0.5^2+\sigma^2)$.
	
		\begin{figure}[htbp]
		\centering
		\caption{Absolute error for $\hat\rho-\rho$ ($\tilde\rho-\rho$, respectively) and IMSE for $\hat{U},\hat{V} $ ($\tilde{U},\tilde{V}  $, respectively) on different tuning parameters for the FPCA (left column) and Tikhonov (right column) methods by the average of 200 Monte Carlo replicates with noise level $\sigma=0.05$ in Case 1.}
		\includegraphics[width=\textwidth]{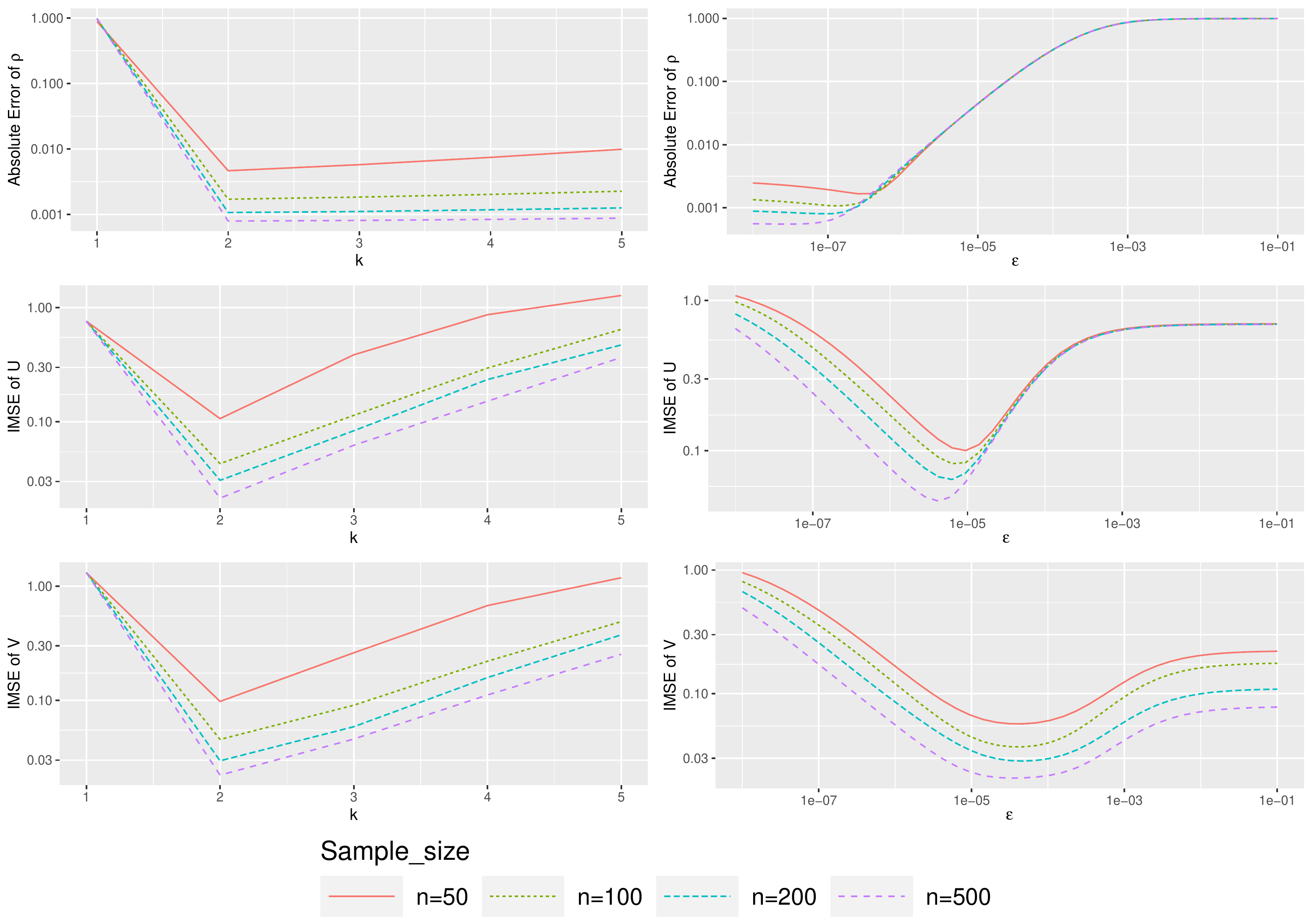}	
		\label{fig:0.05case1}	
		\end{figure}
		\begin{table}[htbp]
			\caption{Absolute Error of $\hat\rho-\rho$ ($\tilde\rho-\rho$, respectively) for different noise levels with tuning parameters chosen by five-fold CV\label{tab:CV-rho}}
			\vspace{0.1in}
			\centering
			\resizebox{\textwidth}{!}{%
    \begin{tabular}{cccccc|cccc}
    \hline
    \multirow{2}[4]{*}{} &       & \multicolumn{4}{c}{FPCA}      & \multicolumn{4}{c}{Tikhonov} \bigstrut\\
\cline{2-10}          & $\sigma$ & n=50  & n=100 & n=200 & n=500 & n=50  & n=100 & n=200 & n=500 \bigstrut\\
    \hline
    \multirow{5}[10]{*}{Case 1} & 0.05  & 1.8921E-03 & 1.0968E-03 & 8.2418E-04 & 5.6042E-04 & 3.8412E-02 & 1.2611E-02 & 6.8533E-03 & 4.7028E-03 \bigstrut\\
\cline{2-10}          & 0.1   & 7.3811E-03 & 4.3114E-03 & 3.2315E-03 & 2.1716E-03 & 4.4891E-02 & 2.1650E-02 & 1.4277E-02 & 9.2885E-03 \bigstrut\\
\cline{2-10}          & 0.2   & 2.6169E-02 & 1.5551E-02 & 1.1578E-02 & 7.4486E-03 & 7.4891E-02 & 3.5729E-02 & 2.5251E-02 & 1.7177E-02 \bigstrut\\
\cline{2-10}          & 0.3   & 4.8671E-02 & 2.8719E-02 & 2.1629E-02 & 1.3469E-02 & 1.1181E-01 & 4.9525E-02 & 3.6763E-02 & 2.2418E-02 \bigstrut\\
\cline{2-10}          & 0.5   & 8.4157E-02 & 4.9888E-02 & 3.7398E-02 & 2.1870E-02 & 1.5973E-01 & 7.4893E-02 & 5.1280E-02 & 2.7180E-02 \bigstrut\\
    \hline
    \multirow{5}[10]{*}{Case 2} & 0.05  & 1.8223E-03 & 1.0492E-03 & 7.8767E-04 & 5.1942E-04 & 2.8798E-02 & 1.2771E-02 & 6.7532E-03 & 4.6829E-03 \bigstrut\\
\cline{2-10}          & 0.1   & 7.1250E-03 & 4.0953E-03 & 3.0922E-03 & 2.0080E-03 & 4.4630E-02 & 2.1645E-02 & 1.3985E-02 & 9.2176E-03 \bigstrut\\
\cline{2-10}          & 0.2   & 2.5174E-02 & 1.4878E-02 & 1.1056E-02 & 6.9673E-03 & 7.4311E-02 & 3.5712E-02 & 2.4669E-02 & 1.7008E-02 \bigstrut\\
\cline{2-10}          & 0.3   & 4.7609E-02 & 2.8180E-02 & 2.0961E-02 & 1.2753E-02 & 1.0168E-01 & 4.9744E-02 & 3.2470E-02 & 2.1946E-02 \bigstrut\\
\cline{2-10}          & 0.5   & 8.2501E-02 & 4.9165E-02 & 3.6464E-02 & 2.1097E-02 & 1.3788E-01 & 7.5757E-02 & 4.7427E-02 & 2.6488E-02 \bigstrut\\
    \hline
    \end{tabular}
    }%
			
		\end{table}%

	
	We first inspect how the estimation error responds to the tuning parameters for  FPCA and Tikhonov methods,  where we set the same tuning parameters for $X$ and $Y$ for simplicity, i.e., $k_{X}=k_{Y}=k$ and $\epsilon_{X}=\epsilon_{Y}=\epsilon$, respectively. Four sample sizes from 50 to 500 are considered, and each simulation is repeated 200 times independently. We use the absolute error  $|\hat\rho-\rho|$ to quantify the estimation error for $\hat\rho$ and integrated mean squared error $\mathrm{IMSE}(\hat U)= \{\mathbb{E}\|\hat{U}-U \|_{\mu}^{2}\}^{1/2}$ for $\hat{U}$; estimation errors for  $\hat V$, $\tilde \rho$, $\tilde U$ and $\tilde V$ are quantified analogously. From the results in Figure \ref{fig:0.05case1} we see that the estimators converge rapidly  as the the sample size $n$ increases. In addition, we observe that the optimal tuning parameters for $\rho$, $U$ and $V$ are rather different for Tikhonov estimators, e.g., the $\epsilon$ that minimizes $|\hat\rho-\rho|$ may not be the optimal choice for estimating the weight functions $U$ and $V$. By contrast, the optimal truncation parameter $k_{X}$ ($k_{Y}$) for $\rho,U $  and $V$ seems to have little difference.
	
	In practice, we propose to choose the tuning parameters via  a   $\kappa$-fold cross-validation approach with $\kappa=5$, as follows. {We split the dataset into $\kappa$ partitions of roughly even size. Taking the FPCA estimation for example, let $\hat{U}^{-l}_{k} $ and $\hat{V}^{-l}_{k} $ be the FPCA estimators of the weight functions with the tuning (truncation) parameter $k$ obtained without using data from the $l$-th partition. The cross-validation score of $k$ is defined based on the squared Pearson correlation  \citep{leurgans1993}
	\begin{equation}\label{eq:CV}
		\mathrm{CV}( k )=\left[\frac{n\left(\sum_{i=1}^{n}x_{i, k }y_{i, k }\right)-\left( \sum_{i=1}^{n}x_{i, k }\right)\left( \sum_{i=1}^{n}y_{i, k }\right)  }{\sqrt{\left\{n\sum_{i=1}^{n}x_{i, k }^{2}-\left( \sum_{i=1}^{n}x_{i, k }\right)^2 \right\}\left\{n\sum_{i=1}^{n}y_{i, k }^{2}-\left( \sum_{i=1}^{n}y_{i, k }\right)^2 \right\} }}\right]^2
	\end{equation}
	with
	$$x_{i,k}=\sum_{l=1}^{\kappa}\langle\mathrm{Log}_{\hat\mu_{X}}X_{1}, \hat{U}^{-l}_{k}\rangle \mathbf{1}_{i\in I_{l}}\text{ and }y_{i,k}=\sum_{k=1}^{\kappa}\langle\mathrm{Log}_{\hat\mu_{Y}}Y_{1}, \hat{V}^{-l}_{ k }\rangle \mathbf{1}_{i\in I_{l}},$$
where $I_{l}$ denotes the set of indices of data in the $l$-th partition. Then we  choose the value of $k$ that maximizes $\mathrm{CV}(k)$. Selection for the tuning parameters associated with the Tikhonov estimators $\tilde U$ and $\tilde V$ can be performed analogously.

	\begin{table}[htbp]
			
			\centering
			\caption{IMSE of $\hat U$ ($\tilde U$, respectively) for different noise levels with tuning parameters chosen by five-fold CV\label{tab:CV-f}}
			\vspace{0.1in}
			\resizebox{\textwidth}{!}{%
			 \begin{tabular}{cccccc|cccc}
    \hline
    \multirow{2}[4]{*}{} &       & \multicolumn{4}{c}{FPCA}      & \multicolumn{4}{c}{Tikhonov} \bigstrut\\
\cline{2-10}          & $\sigma$ & n=50  & n=100 & n=200 & n=500 & n=50  & n=100 & n=200 & n=500 \bigstrut\\
    \hline
    \multirow{5}[10]{*}{Case 1} & 0.05  & 0.1281  & 0.0900  & 0.0850  & 0.0481  & 0.2536  & 0.2580  & 0.3033  & 0.2498  \bigstrut\\
\cline{2-10}          & 0.1   & 0.1935  & 0.1446  & 0.1377  & 0.0863  & 0.2846  & 0.2519  & 0.2977  & 0.2888  \bigstrut\\
\cline{2-10}          & 0.2   & 0.2979  & 0.2194  & 0.2303  & 0.1417  & 0.3538  & 0.3154  & 0.3156  & 0.2481  \bigstrut\\
\cline{2-10}          & 0.3   & 0.3806  & 0.2880  & 0.2747  & 0.1873  & 0.4318  & 0.3783  & 0.3328  & 0.2551  \bigstrut\\
\cline{2-10}          & 0.5   & 0.5237  & 0.3816  & 0.3429  & 0.2808  & 0.5322  & 0.4876  & 0.3912  & 0.3067  \bigstrut\\
    \hline
    \multirow{5}[10]{*}{Case 2} & 0.05  & 0.1312  & 0.0870  & 0.0847  & 0.0492  & 0.2546  & 0.2537  & 0.3006  & 0.2577  \bigstrut\\
\cline{2-10}          & 0.1   & 0.1994  & 0.1419  & 0.1447  & 0.0828  & 0.2928  & 0.2494  & 0.3037  & 0.2966  \bigstrut\\
\cline{2-10}          & 0.2   & 0.2938  & 0.2124  & 0.2278  & 0.1460  & 0.3562  & 0.3131  & 0.3199  & 0.2497  \bigstrut\\
\cline{2-10}          & 0.3   & 0.3789  & 0.2811  & 0.2744  & 0.1911  & 0.4279  & 0.3723  & 0.3359  & 0.2582  \bigstrut\\
\cline{2-10}          & 0.5   & 0.5178  & 0.3867  & 0.3457  & 0.2764  & 0.5241  & 0.4803  & 0.3853  & 0.3030  \bigstrut\\
    \hline
    \end{tabular}}
	\end{table}%
	\begin{table}[htbp]
	\centering
	\caption{IMSE of $\hat V$ ($\tilde V$, respectively) for different noise levels with tuning parameters chosen by 5-fold CV}\vspace{0.1in}
    \resizebox{\textwidth}{!}{%
\begin{tabular}{cccccc|cccc}
    \hline
    \multirow{2}[4]{*}{} &       & \multicolumn{4}{c}{FPCA}      & \multicolumn{4}{c}{Tikhonov} \bigstrut\\
\cline{2-10}          & $\sigma$ & n=50  & n=100 & n=200 & n=500 & n=50  & n=100 & n=200 & n=500 \bigstrut\\
    \hline
    \multirow{5}[10]{*}{Case 1} & 0.05  & 0.1095  & 0.0774  & 0.0621  & 0.0405  & 0.1943  & 0.1982  & 0.2431  & 0.1849  \bigstrut\\
\cline{2-10}          & 0.1   & 0.1554  & 0.1177  & 0.0989  & 0.0657  & 0.2003  & 0.1907  & 0.2361  & 0.2316  \bigstrut\\
\cline{2-10}          & 0.2   & 0.2111  & 0.1754  & 0.1661  & 0.1048  & 0.2255  & 0.2268  & 0.2359  & 0.2020  \bigstrut\\
\cline{2-10}          & 0.3   & 0.2640  & 0.2340  & 0.2061  & 0.1360  & 0.2549  & 0.2622  & 0.2459  & 0.1997  \bigstrut\\
\cline{2-10}          & 0.5   & 0.3869  & 0.3268  & 0.2753  & 0.2187  & 0.2860  & 0.3360  & 0.2939  & 0.2608  \bigstrut\\
    \hline
    \multirow{5}[10]{*}{Case 2} & 0.05  & 0.1133  & 0.0769  & 0.0645  & 0.0433  & 0.2040  & 0.1955  & 0.2387  & 0.1925  \bigstrut\\
\cline{2-10}          & 0.1   & 0.1568  & 0.1114  & 0.1074  & 0.0663  & 0.2082  & 0.1857  & 0.2414  & 0.2395  \bigstrut\\
\cline{2-10}          & 0.2   & 0.2150  & 0.1674  & 0.1632  & 0.1070  & 0.2271  & 0.2241  & 0.2410  & 0.2041  \bigstrut\\
\cline{2-10}          & 0.3   & 0.2677  & 0.2274  & 0.2062  & 0.1422  & 0.2578  & 0.2563  & 0.2504  & 0.2029  \bigstrut\\
\cline{2-10}          & 0.5   & 0.3847  & 0.3287  & 0.2829  & 0.2085  & 0.2922  & 0.3287  & 0.2917  & 0.2571  \bigstrut\\
    \hline
    \end{tabular}}
\label{tab:CV-g}%
	\end{table}

	Tables \ref{tab:CV-rho} to \ref{tab:CV-g}  report the absolute error  $|\hat\rho-\rho|$ ($|\tilde{\rho}-\rho|$, respectively) and the IMSE for estimating $U$ and $V$, where the tuning parameters are selected  to optimize the cross-validation score \eqref{eq:CV}. From Table \ref{tab:CV-rho}, we see that the correlation $\rho$ is estimated well by both  methods and   converges rapidly as the sample size increases. When inspecting the absolute error of $\hat\rho-\rho$ ($\tilde\rho- \rho$, respectively) in Figure \ref{fig:0.05case1} and Table \ref{tab:CV-rho} closely, one may find that Figure \ref{fig:0.05case1} shows similar scales between two methods over ranges of tuning parameters, while Table \ref{tab:CV-rho} shows that estimation errors by FPCA method are clearly lower than Tikhonov methods with tuning parameters chosen by five-fold CV. This might be explained by Figure \ref{fig:CVselect}, which presents the selected tuning parameters in both methods. For the FPCA method, we see that most of the selected truncation numbers  are concentrated in $2$, that is, exactly the optimal tuning parameter. Figure \ref{fig:0.05case1} also indicates that the optimal range  for $\epsilon$ is around $[10^{-8},10^{-6}]$ (the optimal choice for $\epsilon$ is decreasing as the sample size grows, which is consistent with our theoretical results), but as shown in Figure \ref{fig:CVselect}, the majority of the selected epsilons do not fall in this range. Although as the sample size increases, the selected  $\epsilon$ gradually becomes smaller and approaches the optimal one, there is still a substantial part outside the optimal range, which makes the Tikhonov method's performance worse than the FPCA method. In addition,  the CV method that we propose is aiming to select the best correlation $\rho$, and thus the  tuning parameters selected to optimize the estimation of $\rho$ might not be optimal for the weight functions, which can be observed by comparing the minimum points attained by the Tikhonov method in Figure \ref{fig:0.05case1}. {This might partially explain why the Tikhonov method yields larger estimation errors for the weight functions than the FPCA method, as exhibited in Tables \ref{tab:CV-f} and \ref{tab:CV-g}}.} Additional figures can be found in Appendix \ref{sec:sim-add}.
			\begin{figure}[htbp]
		\centering
		\caption{Histograms for selected tuning parameters in FPCA (left column) and Tikhonov (right column) methods by five-fold CV with noise level $\sigma=0.05$ in Case 1.  }	
		\includegraphics[width=\textwidth]{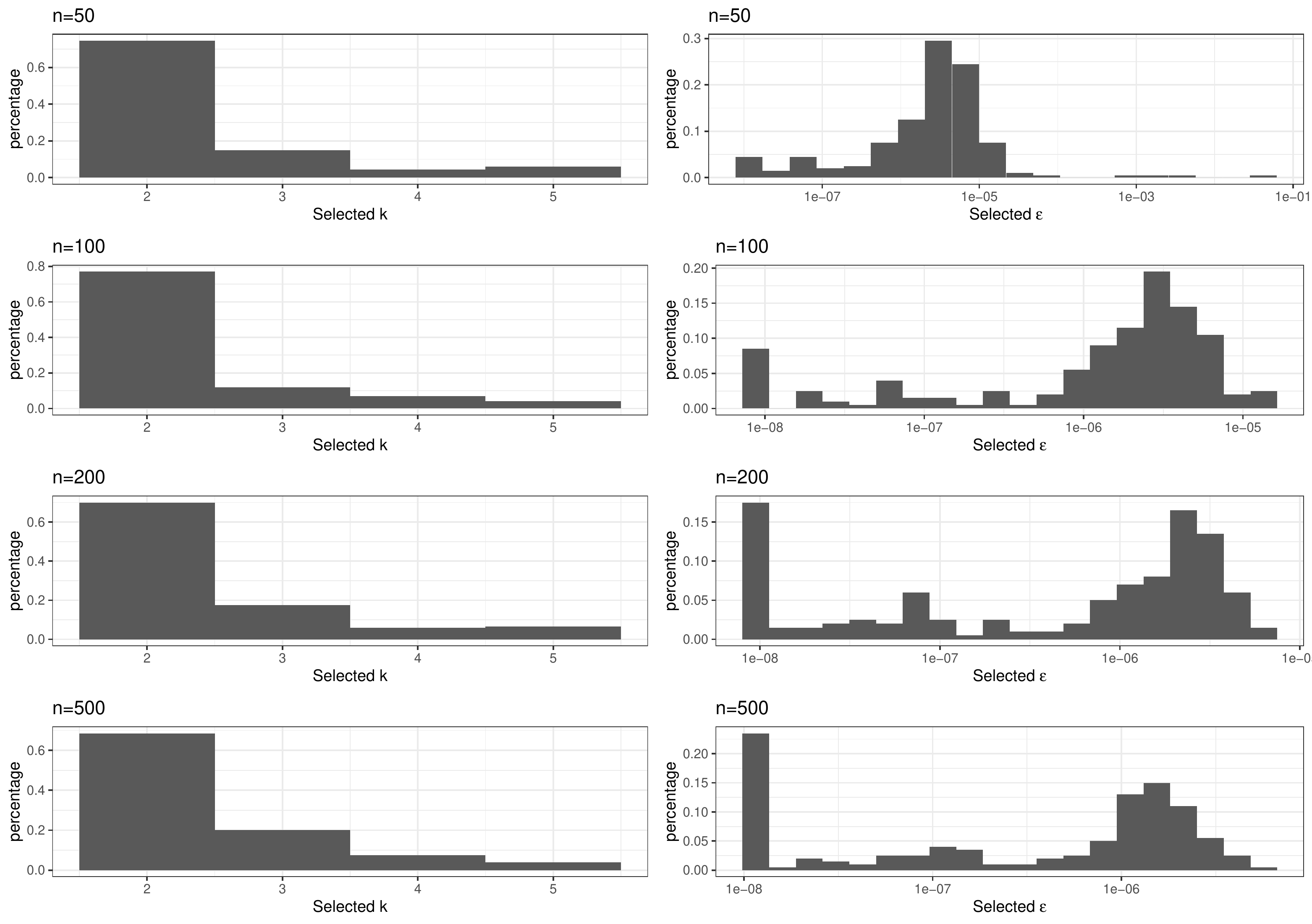}
		\label{fig:CVselect}	
		\end{figure}	

	\section{Application}\label{sec:data}
	We apply the proposed Wasserstein correlation analysis to study functional connectivities between different brain regions by using data from the HCP 1200 subjects release \citep{van2013}. Specifically, we focus on the correlation of longitudinally measured distributions of  the signal strength between two specific areas in the brain.
	
	The dataset we use consists of $n=209$ subjects who are healthy young adults and have been scanned for both a resting-state fMRI (rsfMRI) image and a  task-evoked fMRI (tfMRI) image related to fine motor skills. The rsfMRI data were acquired in four runs of approximately 15 minutes each, two runs in one session and two in the other session. During data acquisition, subjects were instructed to keep their eyes open and fixed  on a projected bright cross-hair on a dark background that was presented in a darkened room. In the tfMRI data acquisition, subjects were presented with visual cues that asked them to either tap their left or right fingers, squeeze their left or right toes, or move their tongue to map motor areas. During the experiment, the brain of each subject was scanned and the neural activities were recorded at 284 equal-spaced time points. More details of the experiment and data acquisition can be found in the reference manual of WU-Minn HCP 1200 Subjects Data Release.
%
%

	We consider the Putamen and Caudate nucleus areas 
	which are known to be related to motor skills in medical literature. At each time $t$,  $X_{i}^{rest}(t)$ and $X_{i}^{motor}(t)$ denote the distributions of the signal strength in the left Caudate nucleus of the $i$-th subject for rsfMRI and tfMRI, respectively, and $Y_{i}^{rest}(t)$ and $Y_{i}^{motor}(t)$ represent the distributions of the signal strength in the  Putamen area. The proposed Wasserstein correlation analysis is applied to investigate the correlation between $X_{i}^{rest}(t)$ and $Y_{i}^{rest}(t)$, as well as the correlation between $X_{i}^{motor}(t)$ and $Y_{i}^{motor}(t)$. The tuning parameters are selected by 5-fold CV \eqref{eq:CV}.
\begin{table}[htbp]
  \caption{Top five largest correlations obtained by FPCA and Tikhonov methods}
    \centering\vspace{0.1in}
    \begin{tabular}{ccccccc}
    \hline
    $\rho$   &       & 1     & 2     & 3     & 4     & 5 \bigstrut\\
    \hline
    \multirow{2}[4]{*}{FPCA} & rfMRI & 0.6966  & 0.3505  & 0.2970  & 0.2593  & 0.2474  \bigstrut\\
\cline{2-7}          & tfMRI & 0.7476  & 0.4448  & 0.4108  & 0.3093  & 0.2809  \bigstrut\\
    \hline
    \multirow{2}[4]{*}{Tikhonov} & rfMRI & 0.7250  & 0.3246  & 0.2510  & 0.2160  & 0.1840  \bigstrut\\
\cline{2-7}          & tfMRI & 0.7457  & 0.3450  & 0.2960  & 0.2135  & 0.1300  \bigstrut\\
    \hline
    \end{tabular}%
  \label{tab:LLcca}%
\end{table}%

	From Table \ref{tab:LLcca} that reports the top five correlations, we observe that FPCA and Tikhonov methods yield similar pattern that the correlation between these two areas in the tfMRI images is larger relative to its rsfMRI counterpart, while difference seems more pronounced by FPCA method.  From Figure \ref{fig:weightfun1}, {we see that compared with rsfMRI,  fluctuation of weight functions along time in tfMRI is more intensive, which suggests that the association between the two brain regions is more dynamic during a motor task.} 
	 In summary, via the proposed method we find that the correlations for distributions of the signal strength between the Putamen area and Caudate nucleus increase during the motor task in contrast to the resting state, moreover, the weight functions in tfMRI express a more active pattern than those in rsfMRI. 
	
				\begin{figure}[htbp]
		\centering
		\caption{Heat-map for the estimated weight functions $\hat U$ (left) and $\hat V$ (right) for the rsfMRI data (top) and  tfMRI data (bottom).  }
		\centering
		\subfigure {
			\includegraphics[width=7.8cm]{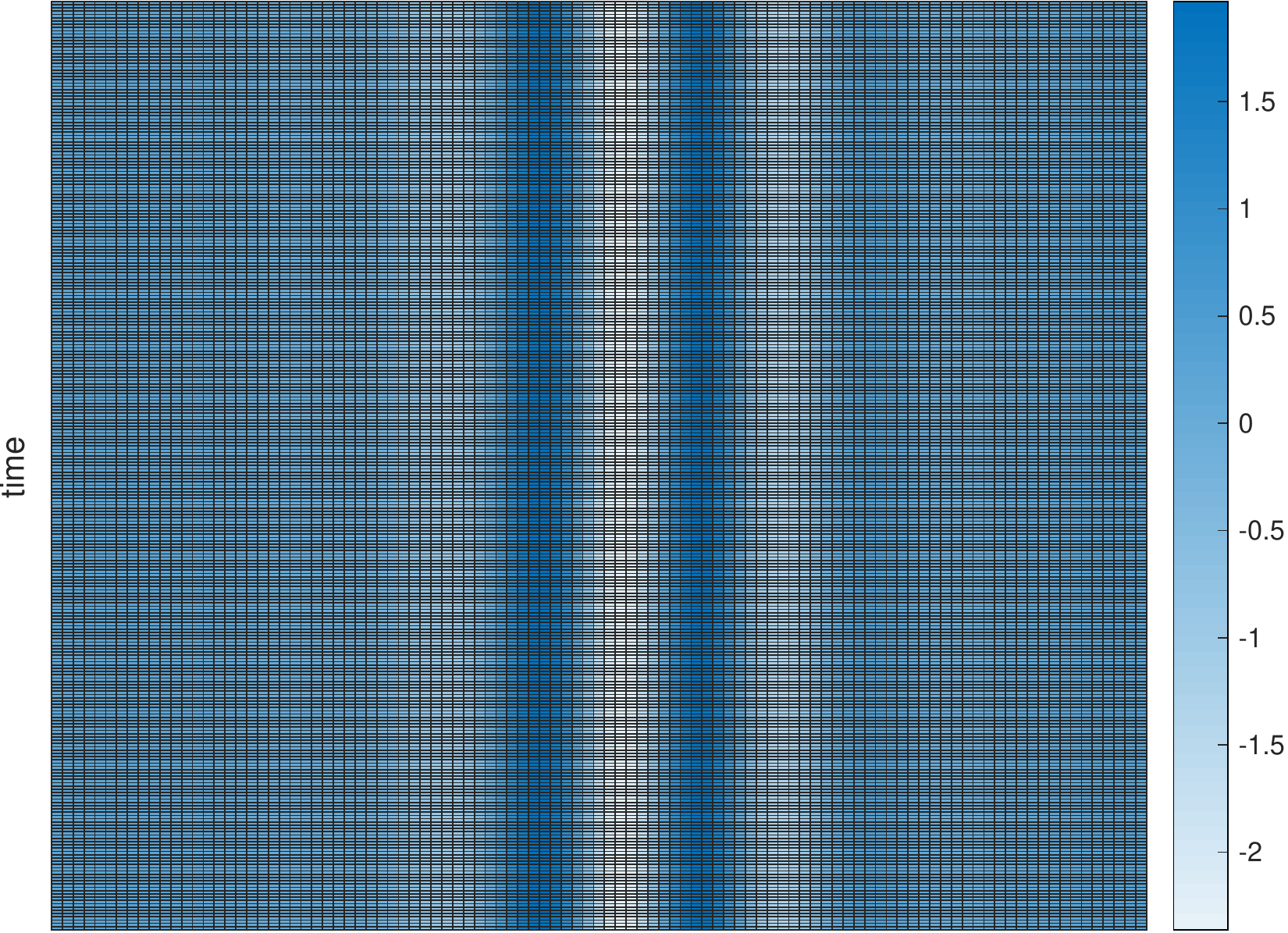}
		}
		\subfigure {
			\includegraphics[width=7.8cm]{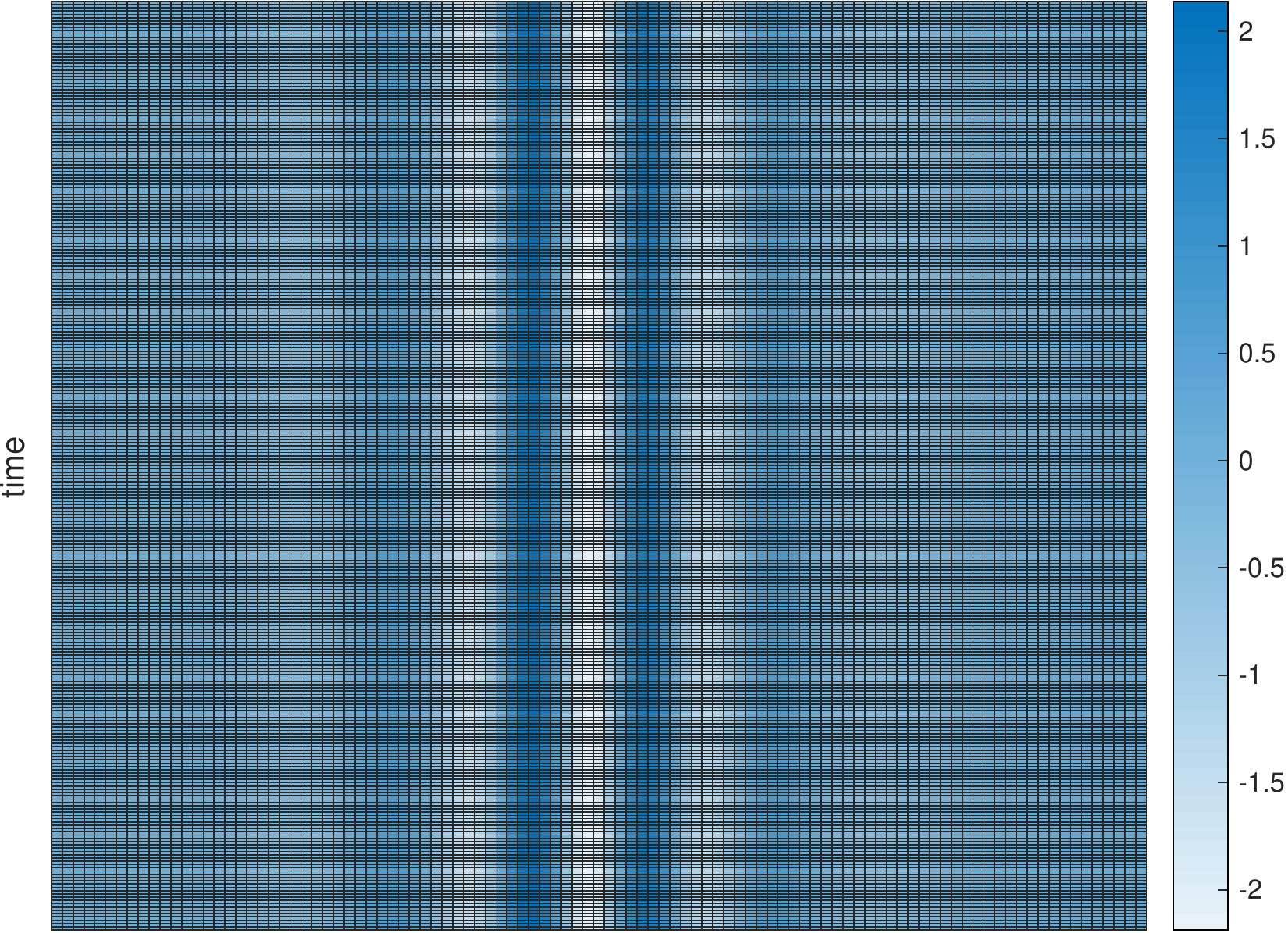}
		}
	
		\subfigure{
			\includegraphics[width=7.8cm]{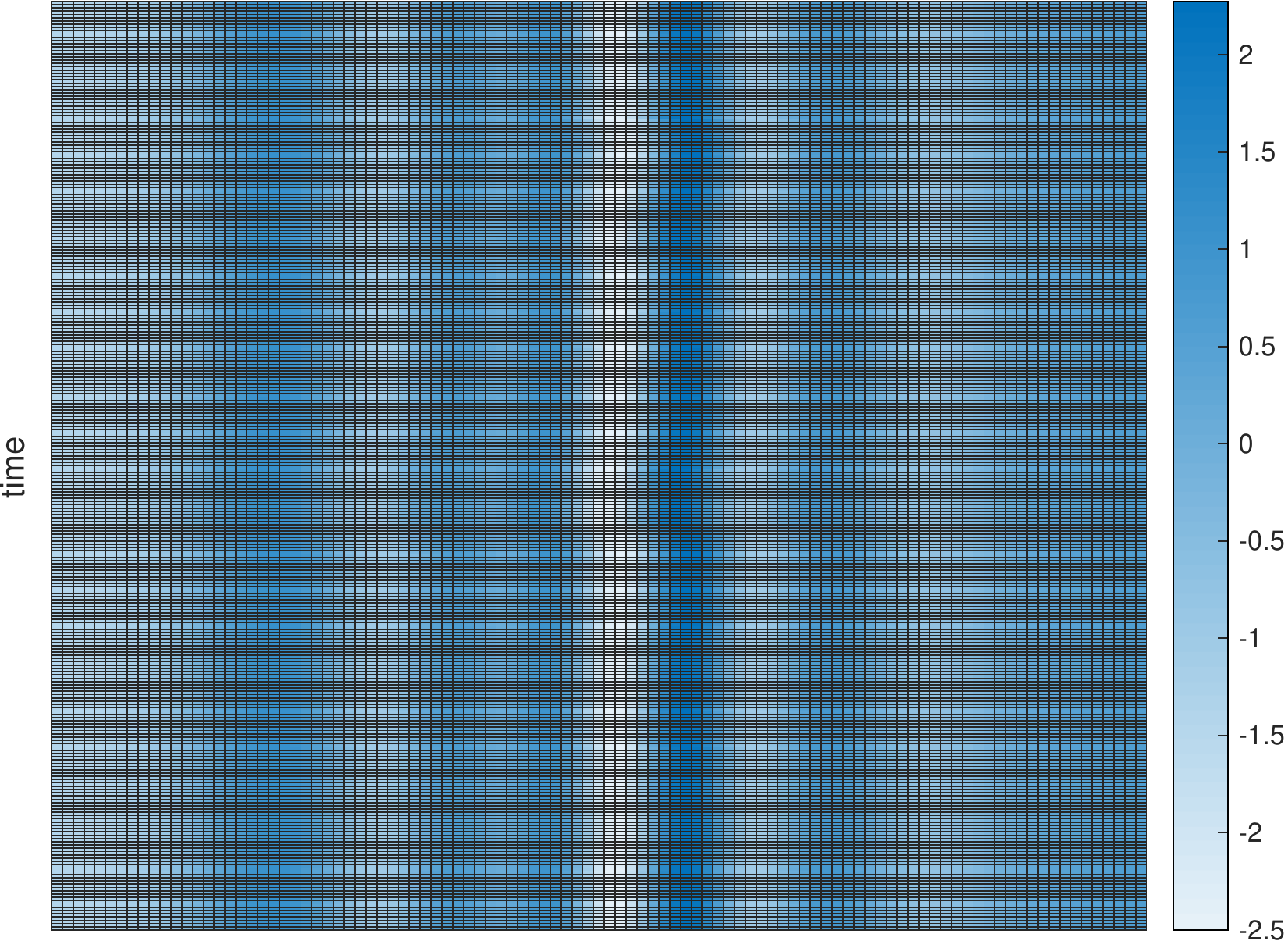}
		}
		\subfigure{
			\includegraphics[width=7.8cm]{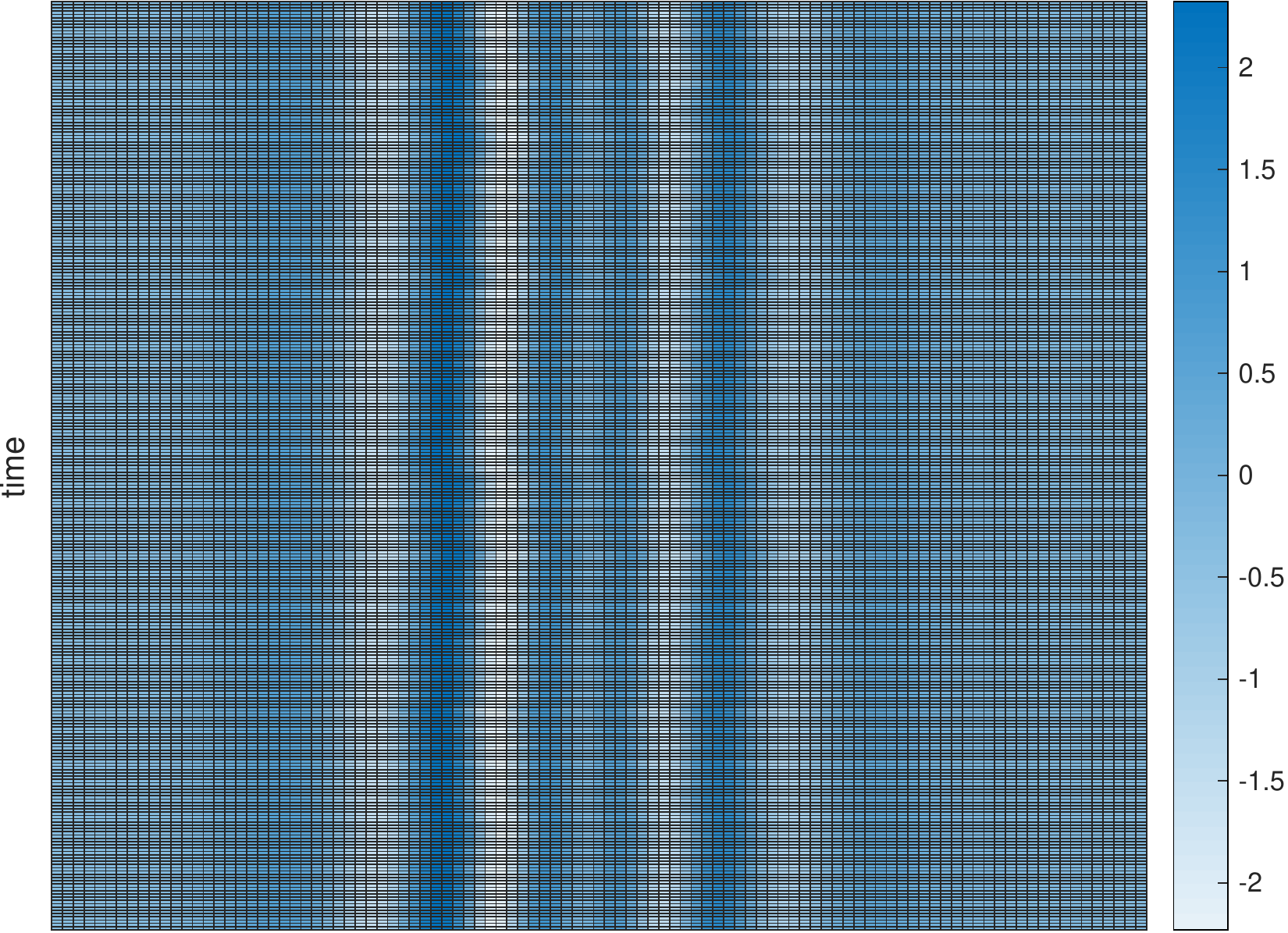}
		}
		\label{fig:weightfun1}
	\end{figure}
	
	\begin{appendices}
		\section{Additional simulation results}\label{sec:sim-add}
Figure \ref{fig:0.1case1} and	 \ref{fig:0.2case1} present the histograms for selected tuning parameters and behaviors of our proposed  estimators with  noise level $\sigma=0.1$ and $\sigma=0.2$.	
			\begin{figure}[htbp]
		\centering
		\includegraphics[width=\textwidth]{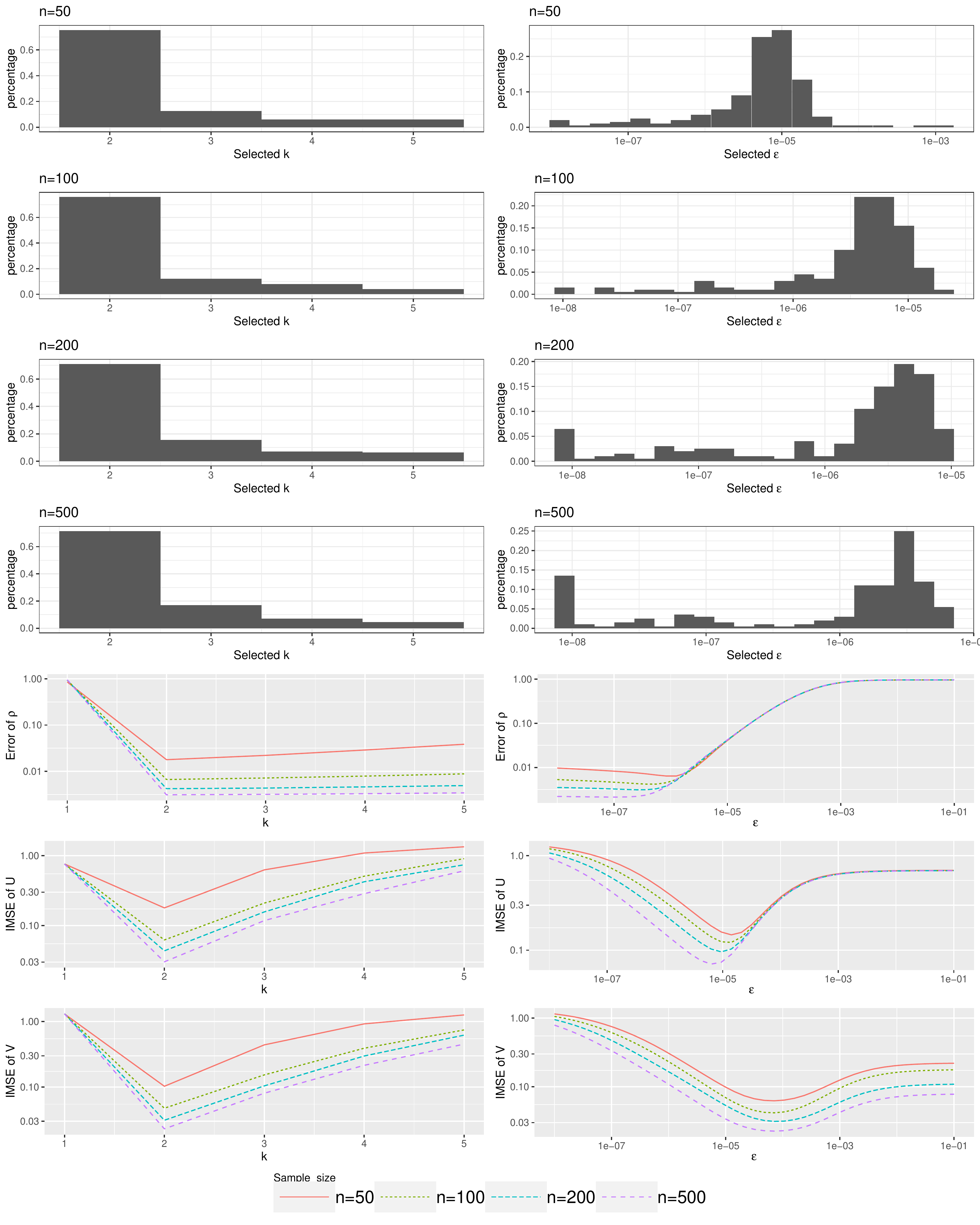}
		\caption{The first four rows present the selected tuning parameters  and the last three rows show the absolute error for $\hat\rho-\rho$ ($\tilde\rho-\rho$, respectively) and IMSE for $\hat{U},\hat{V} $ ($\tilde{U},\tilde{V}  $, respectively) on different tuning parameters for the FPCA (left column) and Tikhonov (right column) methods by the average of 200 Monte Carlo replicates with noise level $\sigma=0.1$ in Case 1.}	
		\label{fig:0.1case1}	
		\end{figure}
		\begin{figure}[htbp]
		\centering
		\includegraphics[width=\textwidth]{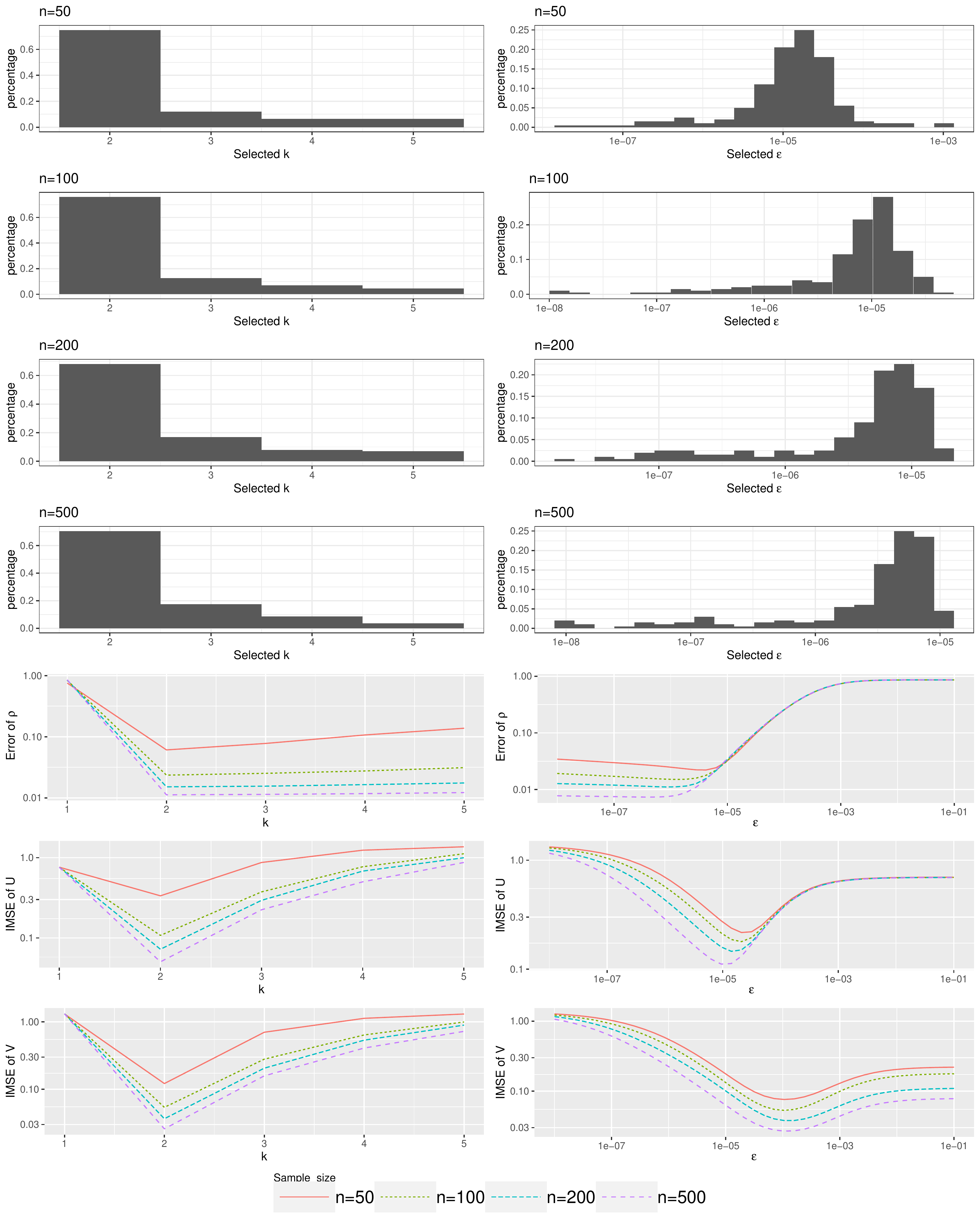}
		\caption{The first four rows present the selected tuning parameters  and the last three rows show the absolute error for $\hat\rho-\rho$ ($\tilde\rho-\rho$, respectively) and IMSE for $\hat{U},\hat{V} $ ($\tilde{U},\tilde{V}  $, respectively) on different tuning parameters for the FPCA (left column) and Tikhonov (right column) methods by the average of 200 Monte Carlo replicates with noise level $\sigma=0.2$ in Case 1.}	
		\label{fig:0.2case1}	
		\end{figure}
	
		\section{Proofs of main results}\label{sec:prof}
		\subsection{Proofs of Theorems and Propositions in Section \ref{sec:WFDA}}\label{sec:prof-WFDA}
		\begin{proof}[of Theorem \ref{thm:tensor}]
			To show that $\mathscr{T}(\mu) $ is a Hilbert space, it is sufficient to prove that $\mathscr{T}(\mu) $ is complete. Suppose that ${V_n}$ is a Cauchy sequence in $\mathscr{T}(\mu)$, i.e., $\lim_{N\rightarrow\infty}\sup_{m,n\geq N} \|V_{n}-V_{m} \|_{\mu}=0 $, where $\|\cdot \|_{\mu} $ denotes the norm induced by the inner product in $\mathscr{T}(\mu)$. From this sequence we choose a subsequence $n_{k} $ such that $\|V_{n_{k}}-V_{n_{k+1}} \|_{\mu}\leqslant2^{-k} $. For $U\in \mathscr{T}(\mu)$, by Cauchy--Schwarz inequality,
			$$\int_{\mathcal{T}} \|U(t) \|_{\mu(t)}\cdot \|V_{n_{k}}(t)-V_{n_{k+1}}(t) \|_{\mu(t)}\mathrm{d}t\leqslant\|U \|_{\mu}\|V_{n_{k}}-V_{n_{k+1}} \|_{\mu}\leqslant 2^{-k} \|U \|_{\mu},$$
			where $\|\cdot  \|_{\mu(t)}$ denotes the norm induced by the inner product in $\tang_{\mu(t)}$. Thus 
			$$\sum_{k=1}^{\infty}\int_{\mathcal{T}} \|U(t) \|_{\mu(t)}\cdot \|V_{n_{k}}(t)-V_{n_{k+1}}(t) \|_{\mu(t)}\mathrm{d}t\leqslant \|U \|_{\mu}<\infty,$$
			which implies 
			\begin{equation}\label{eq:tensor-1}
				\sum_{k=1}^{\infty}\left\|V_{n_{k+1}}(t)-V_{n_{k}}(t)\right\|_{\mu(t)}<\infty, \quad \text{a.e.}.
			\end{equation}
			Since $\tang_{\mu(t)}$ is complete, for those $t\in\tdomain$ such that \eqref{eq:tensor-1} holds,  the limit $V(t)=\lim_{k\rightarrow\infty}V_{n_{k}}(t)$ is defined and falls into $\tang_{\mu(t)}$. For any $\epsilon>0$, choose $N_{\epsilon}$ such that $n,m\geqslant N_{\epsilon} $ and $\|V_{n}-V_{m}\|_{\mu}\leqslant \epsilon $. Fatou’s lemma applied to the function $\|V_{n_{k}}(t)-V_{n_{k+1}}(t) \|_{\mu(t)}$ implies that if $m\geqslant N$, then $\left\|V-V_{m}\right\|_{\mu}^{2} \leq \liminf _{k \rightarrow \infty}\left\|V_{n_{k}}-V_{m}\right\|_{\mu}^{2} \leq \epsilon^{2} $. This shows that $V-V_{m}\in \mathscr{T}(\mu)$, thus $V=(V-V_{m})+V_{m} \in \mathscr{T}(\mu)$. The arbitrariness of $\epsilon$ implies that $\lim _{m \rightarrow \infty}\left\|V-V_{m}\right\|_{\mu}=0$. From the triangle inequality $\| V-V_{n}\|_{\mu}\leqslant\| V-V_{m}\|_{\mu}+\| V_{m}-V_{n}\|_{\mu}\leqslant2\epsilon $, we conclude that $V_{n}$ converges to $V$ in $\mathscr{T}(\mu)$, and further that $\mathscr{T}(\mu)$ is complete.
			
			To see $\mathscr{T}(\mu) $ is separable, we notice that for each $t\in \mathcal{T}$, $\tang_{\mu(t)}$ is a separable Hilbert space thus it has an orthonormal basis $\{\mathbf{\Phi}_{k}(t,s)\}_{k=1}^{\infty} $. Define $\mathbf{O}=\{\{\mathbf{\Phi}_{k}(t,s)\}_{k=1}^{\infty} |\forall t\in\mathcal{T}\}$, which can be regarded as an orthonormal frame along the curve $\mu(t)$. For every element $U\in \mathscr{T}(\mu)$, define $U_{\mathbf{O}}$ be the coordinate representation of $U$ with respect to $\mathbf{O}$. One can see that $U_{\mathbf{O}}$ is an element in the Hilbert space $\mathcal{L}^{2}(\mathcal{T}, l^{2} ) $ of square integrable $l^{2}$-valued measurable functions with norm $\|f\|_{\mathcal{L}^{2}}=\{\int |f(t)|_{l^{2}}^{2} dt \}^{1/2} $ for $ f\in \mathcal{L}^{2}(\mathcal{T},l^{2})$, where $l^{2}$ is the space of square-summable sequences endowed with the inner product $\langle\mathbf{a},\mathbf{b} \rangle_{l^{2}}=\sum_{j=1}^{\infty}a_{j}b_{j}$ and the induced norm $|\cdot|_{l^{2}}$, for each $\mathbf{a}=(a_{1},a_{2},\cdots) $ and $\mathbf{b}=(b_{1},b_{2},\cdots) $. If we define the map $\Upsilon:\mathscr{T}(\mu)\longmapsto \mathcal{L}^{2}(\mathcal{T},l^{2}) $ by $\Upsilon(U)=U_{\mathbf{O}(t)} $, we can immediately see that $\Upsilon$ is a linear map. It is also surjective, because for any $ f\in \mathcal{L}^{2}(\mathcal{T},l^{2})$, the vector field $U$ along $\mu$ given by $U(t)=\sum_{k=1}^{\infty}f_{k}(t)\mathbf{\Phi}_{k}(t,s)$ for $t\in \mathcal{T}$ is an element in $\mathscr{T}(\mu)$, where $f_{k}(t)$ denotes the $k$th component of $f(t)$. It can be verified that $\Upsilon$ preserves the inner product. Therefore, it is a Hilbertian isomorphism. Since $\mathcal{L}^{2}(\mathcal{T},l^{2}) $ is separable, the isomorphism between $\mathcal{L}^{2}(\mathcal{T},l^{2}) $ and $\mathscr{T}(\mu)$ implies separability of  $\mathscr{T}(\mu)$. \QED
		\end{proof}
		\begin{proof}[of Theorem \ref{thm:meancov}]
			By the definition \eqref{eq:logmap} and the equation \eqref{def:muest}, 
			$$
			\mathrm{Log}_{\mu(t)}\hat{\mu}(t)=F_{\hat\mu(t)}^{-1}\circ F_{\mu(t)}-\mathbf{id}=\frac{1}{n}\sum_{i=1}^{n}F_{X_{i}(t)}^{-1}\circ F_{\mu(t)}-\mathbf{id}=\frac{1}{n}\sum_{i=1}^{n}\mathrm{Log}_{\mu(t)}X_{i}(t).
			$$
			By Proposition 3.2.14 in \cite{panaretos2020}, we have $\mathbf{E}\mathrm{Log}_{\mu_{i}(t)}X_{i}(t)=0 $. Given this, the part (a) is verified by applying the central limit theorem in Hilbert space \citep{aldous1976} that asserts  convergence of the process $(\sqrt{n})^{-1}\sum_{i=1}^{\infty} \mathrm{Log}_{\mu(t)}X_{i}(t) $ to a Gaussian measure on tensor Hilbert space $\mathscr{T}(\mu)$ with covariance operator $\mathbf{C}'(\cdot)=\mathbf{E}(\llangle U,\cdot \rrangle_{\mu}U )$ defined via the random element $U=\mathrm{Log}_{\mu(t)}X_{i}(t) $ in $\mathscr{T}(\mu)$. The first statement of the part (b) is a corollary of (a), while the second statement follows from the first one and the compactness of $\mathcal{T}$. 
			
			For assertion  (c),  note that $\mu(t)$ is atomless for each $t\in\tdomain$, which implies that $\hat\mu(t)$ is also atomless \citep{chen2020}. Then we have
			$$\begin{aligned}
				\mathcal{P}_{\mathfrak{B(\hat\mu,\hat\mu)} }^{\mathfrak{B(\mu,\mu )} }\mathbf{\hat C} -\mathbf{C} =&\frac{1}{n}\sum_{i=1}^{n}\left(\mathcal{P}_{\hat\mu }^{\mu } \mathrm{Log}_{\hat\mu } X_{i} \right)\otimes\left(\mathcal{P}_{\hat\mu }^{\mu } \mathrm{Log}_{\hat\mu } X_{i} \right)-\mathbf{C} \\
				=&\frac{1}{n}\sum_{i=1}^{n}\left( \mathrm{Log}_{\mu } X_{i} \right)\otimes\left(\mathrm{Log}_{\mu } X_{i} \right)-\mathbf{C} \\
				&+\frac{1}{n}\sum_{i=1}^{n}\left(\mathcal{P}_{\hat\mu }^{\mu } \mathrm{Log}_{\hat\mu } X_{i}- \mathrm{Log}_{\mu } X_{i} \right)\otimes\left(\mathrm{Log}_{\mu } X_{i} \right)\\
				&+\frac{1}{n}\sum_{i=1}^{n}\left( \mathrm{Log}_{\mu } X_{i} \right)\otimes\left(\mathcal{P}_{\hat\mu }^{\mu } \mathrm{Log}_{\hat\mu } X_{i}-\mathrm{Log}_{\mu } X_{i} \right)\\
				&+\frac{1}{n}\sum_{i=1}^{n}\left(\mathcal{P}_{\hat\mu }^{\mu } \mathrm{Log}_{\hat\mu } X_{i}- \mathrm{Log}_{\mu } X_{i} \right)\otimes\left(\mathcal{P}_{\hat\mu }^{\mu } \mathrm{Log}_{\hat\mu } X_{i}-\mathrm{Log}_{\mu } X_{i} \right)\\
				:=&A_{1}+A_{2}+A_{3}+A_{4}.
			\end{aligned}$$
		For $A_2$, we further have
			\begin{equation}\label{eq:meancov-1}
				\begin{aligned}
					|||A_{2}|||_{\mathfrak{B}(\mu ,\mu )}^{2}\lesssim&\frac{1}{n^{2}}\sum_{i_{1}=1}^{n}\sum_{i_{2}=1}^{n}(\|\mathrm{Log}_{\mu }X_{i_{1}} \|_{\mu }^{2}+\|\mathrm{Log}_{\mu }X_{i_{2}} \|_{\mu }^{2} )\\
					&\times (\|\mathcal{P}_{\hat\mu }^{\mu }\mathrm{Log}_{\hat\mu }X_{i_{1}}-\mathrm{Log}_{\mu }X_{i_{1}} \|_{\mu }^{2} +\|\mathcal{P}_{\hat\mu }^{\mu }\mathrm{Log}_{\hat\mu }X_{i_{2}}-\mathrm{Log}_{\mu }X_{i_{2}} \|_{\mu }^{2})\\
					=&\frac{2}{n}\sum_{i=1}^{n}\|\mathrm{Log}_{\mu }X_{i} \|_{\mu }^{2}\|\mathcal{P}_{\hat\mu }^{\mu }\mathrm{Log}_{\hat\mu }X_{i}-\mathrm{Log}_{\mu }X_{i} \|_{\mu }^{2}\\
					&+\frac{2}{n^{2}}\sum_{i_{1}=1}^{n}\sum_{i_{2}=1}^{n}\|\mathrm{Log}_{\mu }X_{i_{1}} \|_{\mu }^{2}\|\mathcal{P}_{\hat\mu }^{\mu }\mathrm{Log}_{\hat\mu }X_{i_{2}}-\mathrm{Log}_{\mu }X_{i_{2}} \|_{\mu }^{2}.
				\end{aligned}
			\end{equation}
			For the first term on the right hand side of equation  \eqref{eq:meancov-1}, 
			$$
			\begin{aligned}
				&\frac{1}{n}\sum_{i=1}^{n}\|\mathrm{Log}_{\mu }X_{i} \|_{\mu }^{2}\|\mathcal{P}_{\hat\mu }^{\mu }\mathrm{Log}_{\hat\mu }X_{i}-\mathrm{Log}_{\mu }X_{i} \|_{\mu }^{2}\\
				=&\frac{1}{n}\sum_{i=1}^{n}\|\mathrm{Log}_{\mu }X_{i} \|_{\mu }^{2}\times
				\int\left\langle \mathcal{P}_{\hat\mu (t)}^{\mu (t)}\mathrm{Log}_{\hat\mu (t)}X_{i}(t)-\mathrm{Log}_{\mu (t)}X_{i}(t),\right.\\
				&\left.\mathcal{P}_{\hat\mu (t)}^{\mu (t)}\mathrm{Log}_{\hat\mu (t)}X_{i}(t)-\mathrm{Log}_{\mu (t)}X_{i}(t) \right\rangle_{\mu (t) }\mathrm{d}t	\\
				=&\iint\left(F_{\hat\mu (t)}^{-1}\circ F_{\mu (t)}(u)-u \right)^2\mathrm{d}F_{\mu(t)}(u)\mathrm{d}t\times \frac{1}{n}\sum_{i=1}^{n}\|\mathrm{Log}_{\mu }X_{i} \|_{\mu }^{2}\\
				=&\int_{\mathcal{T}} d^2(\hat\mu (t),\mu (t))\mathrm{d}t \times \frac{1}{n}\sum_{i=1}^{n}\|\mathrm{Log}_{\mu }X_{i} \|_{\mu }^{2}\\
				=&O_{p}\left(\frac{1}{n}\right),
			\end{aligned}
			$$
			where the second equality follows from $\mathcal{P}_{\hat\mu (t)}^{\mu (t)}\mathrm{Log}_{\hat\mu (t)}X_{i}(t)-\mathrm{Log}_{\mu (t)}X_{i}(t)=F_{\hat\mu (t)}^{-1}\circ F_{\mu (t)}-\mathbf{id} $, and the last equality is based on the part (b) and $n^{-1} \sum_{i=1}^{n}\|\mathrm{Log}_{\mu }X_{i} \|_{\mu }^{2}=O_{p}(1)$ by the law of large numbers. A similar argument shows that the second term in \eqref{eq:meancov-1} is of order $O_{p}(n^{-1}) $. 
			Thus $|||A_{2}|||_{\mathfrak{B}(\mu ,\mu )}^{2}=O_{p}(n^{-1})$. Analogous calculation shows that $|||A_{3}|||_{\mathfrak{B}(\mu ,\mu )}^{2}=O_{p}(n^{-1})$ and $|||A_{4}|||_{\mathfrak{B}(\mu ,\mu )}^{2}=O_{p}(n^{-2})$. According to \cite{dauxois1982}, $|||{n}^{-1}\sum_{i=1}^{n}\left( \mathrm{Log}_{\mu } X_{i} \right)\otimes\left(\mathrm{Log}_{\mu } X_{i} \right)-\mathbf{C} |||_{\mathfrak{B}(\mu ,\mu )}^{2}=O_{p}(n^{-1})$, and consequently, $\left|\left|\left|\mathcal{P}_{\mathfrak{B(\hat\mu ,\hat\mu )} }^{\mathfrak{B(\mu ,\mu )} }\mathbf{\hat C} -\mathbf{C} \right|\right|\right|^2_{\mathfrak{B(\mu ,\mu )}}=O_{p}(n^{-1})$. The result for $\hat\lambda_{k}$ follows from the perturbation argument in \cite{bosq2000}.
			
			For the part (e), we first note that 
			$$
			\begin{aligned}
				\mathcal{P}_{\mathfrak{B(\hat\mu ,\hat\mu )} }^{\mathfrak{B(\mu ,\mu )} }\mathbf{\hat C} =& \frac{1}{n}\sum_{i=1}^{n}\left(\mathcal{P}_{\hat\mu }^{\mu }\mathrm{Log}_{\hat\mu }X_{i} \right)\otimes\left(\mathcal{P}_{\hat\mu }^{\mu }\mathrm{Log}_{\hat\mu }X_{i} \right)\\
				=& \frac{1}{n}\sum_{i=1}^{n}\left(\mathrm{Log}_{\mu } X_{i}-\overline{\mathrm{Log}_{\mu } X } \right)\otimes\left(\mathrm{Log}_{\mu } X_{i}-\overline{\mathrm{Log}_{\mu } X } \right),
			\end{aligned}
			$$
			where $\overline{\mathrm{Log}_{\mu } X }=n^{-1}\sum_{i=1}^{n} \mathrm{Log}_{\mu} X_{i}$. 
By the proof of Theorem 5.1.18 in \cite{hsing2015}, for all $\{j:\lnorm\Delta\rnorm_{\mathfrak{B}(\mu,\mu) }\leqslant \eta_{j}/2 \}$, we have the following expansion, 
			$$
	\begin{aligned}
		\mathcal{P}_{\hat\mu }^{\mu }\mathbf{\hat\Phi} _{j}-\mathbf{\Phi} _{j}
	= & \sum_{k\neq j}\frac{ \llangle \Delta \mathbf{\Phi} _{j}, \mathbf{\Phi} _{k}\rrangle_{\mu}}{\left(\lambda_{j}-\lambda_{k}\right)}\mathbf{\Phi} _{k}+\sum_{k\neq j}\frac{\llangle\Delta(\mathcal{P}_{\hat\mu }^{\mu }\mathbf{\hat\Phi} _{j}- \mathbf{\Phi} _{j}), \mathbf{\Phi} _{k}\rrangle_{\mu}}{\left(\lambda_{j}-\lambda_{k}\right)}\mathbf{\Phi} _{k} \\
	&+\sum_{k\neq j}\sum_{s=1}^{\infty}\frac{(\lambda_{j}-\hat\lambda_{j})^s}{(\lambda_{j}-\lambda_{k} )^{s+1} }\llangle\Delta \mathcal{P}_{\hat\mu }^{\mu }\mathbf{\hat\Phi} _{j}, \mathbf{\Phi} _{k}\rrangle_{\mu} \mathbf{\Phi} _{k}+\llangle (\mathcal{P}_{\hat\mu }^{\mu }\mathbf{\hat\Phi} _{j}- \mathbf{\Phi} _{j}), \mathbf{\Phi} _{j}\rrangle_{\mu}\mathbf{\Phi} _{j},
	\end{aligned}
$$
where $\Delta=\mathbf{\mathcal{P}_{\mathfrak{B(\hat\mu ,\hat\mu )} }^{\mathfrak{B(\mu ,\mu )} }\hat C} -\mathbf{C} $. For all $k\neq j$, we have
\begin{equation}\label{eq:mean-cov-2}
	\begin{aligned}
		\mathbf{E}\llangle\Delta \mathbf{\Phi}_{j}, \mathbf{\Phi}_{k} \rrangle_{\mu }^2=&\mathbf{E} \left\llangle \mathcal{P}_{\mathfrak{B(\hat\mu ,\hat\mu )} }^{\mathfrak{B(\mu ,\mu )} }\mathbf{\hat C} \mathbf{\Phi}_{j}, \mathbf{\Phi}_{k} \right\rrangle_{\mu }^2\\
				=&\mathbf{E}\left\{\frac{1}{n}\sum_{i=1}^{n}\llangle \mathrm{Log}_{\mu }X_{i}-\overline{\mathrm{Log}_{\mu } X_{i} },\mathrm{\Phi}_{i} \rrangle_{\mu }\llangle \mathrm{Log}_{\mu }X_{i}-\overline{\mathrm{Log}_{\mu } X_{i} },\mathrm{\Phi}_{i} \rrangle_{\mu } \right\}^2\\
				=&\mathbf{E}\left\{\frac{1}{n}\sum_{i=1}^{n}(\xi_{i,k}-\overline{\xi}_{k})(\xi_{i,j}-\overline{\xi}_{j}) \right\}^{2}\\
				=&\frac{\lambda_{j}\lambda_{k}}{n}\left(1-\frac{1}{n} \right)^2,
	\end{aligned}
\end{equation}
			where $\overline{\xi}_{k}=n^{-1}\sum_{i=1}^{n}\xi_{i,k}$ and $\xi_{i,k} $ is the $k$-th component score of $\mathrm{Log}_{\mu }X_{i}$.  By equation \eqref{eq:mean-cov-2} and Lemma 7 in \cite{dou2012}, 
\begin{equation}\label{eq:mean-cov-3}
	\mathbf{E}\left\|\sum_{k\neq j}\frac{ \llangle \Delta \mathbf{\Phi} _{j}, \mathbf{\Phi} _{k}\rrangle_{\mu}}{\left(\lambda_{j}-\lambda_{k}\right)}\mathbf{\Phi} _{k} \right\|_{\mu}^2=\sum_{k\neq j}\frac{\mathbf{E}\llangle \Delta \mathbf{\Phi} _{j}, \mathbf{\Phi} _{k}\rrangle_{\mu}^2 }{\left(\lambda_{j}-\lambda_{k}\right)^2}=\frac{1}{n}\left(1-\frac{1}{n}\right)^2\sum_{k\neq j}\frac{\lambda_{j}\lambda_{k}}{\left(\lambda_{j}-\lambda_{k}\right)^2}\leqslant C \frac{j^{2}}{n},
\end{equation}			
	where $C$ is a constant does not depend on $j$. From Bessel's inequality and given that $|||\Delta |||_{\mathfrak{B(\mu ,\mu )}}<\eta_j/2$,	
	\begin{equation}\label{eq:mean-cov-3}
		\begin{aligned}
			\mathbf{E}\left\|\sum_{k\neq j}\frac{\llangle\Delta(\mathcal{P}_{\hat\mu }^{\mu }\mathbf{\hat\Phi} _{j}- \mathbf{\Phi} _{j}), \mathbf{\Phi} _{k}\rrangle_{\mu}}{\left(\lambda_{j}-\lambda_{k}\right)}\mathbf{\Phi} _{k} \right\|_{\mu}^2\leqslant&\mathbf{E}\frac{\lnorm\Delta \rnorm^2_{\mathfrak{B}(\mu,\mu)}\|\mathcal{P}_{\hat\mu }^{\mu }\mathbf{\hat\Phi} _{j}-\mathbf{\Phi} _{j}\|^2_{\mu}  }{(2\eta_{j})^2}\\
			<&\frac{1}{16}\mathbf{E}\|\mathcal{P}_{\hat\mu }^{\mu }\mathbf{\hat\Phi} _{j}-\mathbf{\Phi} _{j}\|^2_{\mu}. 
		\end{aligned}
	\end{equation}	
	Similarly,
	\begin{equation}\label{eq:mean-cov-4}
		\begin{aligned}
		&\mathbb{E}\left\|\sum_{k\neq j}\sum_{s=1}^{\infty}\frac{(\lambda_{j}-\hat\lambda_{j})^s}{(\lambda_{j}-\lambda_{k} )^{s+1} }\llangle\Delta \mathcal{P}_{\hat\mu }^{\mu }\mathbf{\hat\Phi} _{j}, \mathbf{\Phi} _{k}\rrangle_{\mu}  \mathbf{\Phi} _{k} \right\|_{\mu}^2\\
	 	=&\mathbb{E}\sum_{k\neq j}\left(\sum_{s=1}^{\infty}\frac{(\lambda_{j}-\hat\lambda_{j})^s}{(\lambda_{j}-\lambda_{k} )^{s+1} } \right)^{2}\llangle\Delta \mathcal{P}_{\hat\mu }^{\mu }\mathbf{\hat\Phi} _{j}, \mathbf{\Phi} _{k}\rrangle_{\mu}^2\\
	 	\leqslant&\mathbb{E}\frac{\lnorm\Delta \rnorm_{\mathfrak{B}(\mu,\mu)}^2}{(2\eta_{j}-\lnorm\Delta \rnorm_{\mathfrak{B}(\mu,\mu)})^2}\left\{2\sum_{k\neq j}\frac{\llangle\Delta \mathcal{P}_{\hat\mu }^{\mu }\mathbf{\Phi} _{j}, \mathbf{\Phi} _{k}\rrangle_{\mu}^2 }{\left(\lambda_{j}-\lambda_{k}\right)^{2}}\right.+2\left.\sum_{k\neq j}\frac{\llangle\Delta \mathcal{P}_{\hat\mu }^{\mu }(\hat{\mathbf{\Phi} }_{j}-\mathbf{\Phi} _{j}), \mathbf{\Phi} _{k}\rrangle_{\mu}^2 }{\left(\lambda_{j}-\lambda_{k}\right)^{2}} \right\}\\
	 	\leqslant& \frac{8}{9}\mathbb{E}\left[ \frac{\lnorm\Delta \rnorm_{\mathfrak{B}(\mu,\mu)}^2}{\eta_{j}^2} \sum_{k\neq j}\frac{\llangle\Delta \mathcal{P}_{\hat\mu }^{\mu }\mathbf{\Phi} _{j}, \mathbf{\Phi} _{k}\rrangle_{\mu}^2 }{\left(\lambda_{j}-\lambda_{k}\right)^{2}}+ \frac{\lnorm\Delta \rnorm_{\mathfrak{B}(\mu,\mu)}^4}{\eta_{j}^4}\|\mathcal{P}_{\hat\mu }^{\mu }\mathbf{\hat\Phi} _{j}-\mathbf{\Phi} _{j}\|^2_{\mu}\right]\\
	 	\leqslant& \frac{2}{9}\mathbb{E}\sum_{k\neq j}\frac{\llangle\Delta \mathcal{P}_{\hat\mu }^{\mu }\mathbf{\Phi} _{j}, \mathbf{\Phi} _{k}\rrangle_{\mu}^2 }{\left(\lambda_{j}-\lambda_{k}\right)^{2}}+\frac{1}{18}\mathbb{E}\|\mathcal{P}_{\hat\mu }^{\mu }\mathbf{\hat\Phi} _{j}-\mathbf{\Phi} _{j}\|^2_{\mu} 	 \end{aligned}.
	\end{equation}
	Combing equation \eqref{eq:mean-cov-2} to \eqref{eq:mean-cov-4}, the proof is completed by the fact  $\|\llangle (\mathcal{P}_{\hat\mu }^{\mu }\mathbf{\hat\Phi} _{j}- \mathbf{\Phi} _{j}), \mathbf{\Phi} _{j}\rrangle_{\mu}\mathbf{\Phi} _{j}\|_{\mu}=1/2\|\mathcal{P}_{\hat\mu }^{\mu }\mathbf{\hat\Phi} _{j}-\mathbf{\Phi} _{j}\|^2_{\mu}  $ and Cauchy-Schwarz inequality.
	
			\QED
		\end{proof}
		\subsection{Proofs of Theorems and Propositions in Section \ref{sec:GCA}}\label{sec:prof-GCA}
		\begin{proof}[of Theorem \ref{thm:fpc}]
			To reduce notational burden, we shall suppress the subscripts and superscripts from $\mathcal{P}_{\hat\mu_{X}}^{\mu_{X}} $ in the sequel. We first show that $\lnorm (\mathcal{P} \hat{\mathbf{C}}_{Y,k_{Y}}^{-1})(\mathcal{P}\hat{\mathbf{C}}_{YX})-\mathbf{C}_{Y}^{-1}\mathbf{C}_{YX} \rnorm^2_{\mathfrak{B}(\mu_{X},\mu_{Y} )}=O_{p}(n^{(1-2b_{Y})/(a_{Y}+2b_{Y})}) $, where  $\mathbf{C}_{Y,k_{Y}}^{-1}=\sum_{j=1}^{k_{Y}}\lambda_{Y,j}^{-1}\mathbf{\Phi}_{Y,j} $.  Given $h=\sum_{j=1}^{\infty}h_{j}\mathbf{\Phi}_{1,j}\in\mathscr{T}(\mu_{X}) $ with $\|h\|_{\mu_{X}}=1 $, we have 
			\begin{equation}\label{eq:fpc-1}
				\begin{aligned}
					(\mathcal{P} \hat{\mathbf{C}}_{Y,k_{Y}}^{-1})(\mathcal{P}\hat{\mathbf{C}}_{YX})h-\mathbf{C}_{Y}^{-1}\mathbf{C}_{YX} h
					=&(\mathcal{P} \hat{\mathbf{C}}_{Y,k_{Y}}^{-1})(\mathcal{P}\hat{\mathbf{C}}_{YX})h-{\mathbf{C}}_{Y,k_{Y}}^{-1}(\mathcal{P}\hat{\mathbf{C}}_{YX})h\\
					&+{\mathbf{C}}_{Y,k_{Y}}^{-1}(\mathcal{P}\hat{\mathbf{C}}_{YX})h-{\mathbf{C}}_{Y,k_{Y}}^{-1}{\mathbf{C}}_{YX}h\\
					&+{\mathbf{C}}_{Y,k_{Y}}^{-1}{\mathbf{C}}_{YX}h-\mathbf{C}_{Y}^{-1}\mathbf{C}_{YX} h\\
					:=&J_{1}+J_{2}+J_{3}.
				\end{aligned}
			\end{equation}
			For the first term in \eqref{eq:fpc-1}, we further have
			\begin{equation}\label{eq:fpc-2}
				\begin{aligned}
					J_{1}=&\sum_{j_{2}=1}^{k_{Y}}\hat{\lambda}_{Y,j_{2}}^{-1}\left\{\sum_{j_1=1}^{\infty}h_{j_{1}}\left(\frac{1}{n}\sum_{i=1}^{n}\xi_{ij_{1}}\eta_{ij_{2}} \right) \right\}\mathcal{P}\hat{\mathbf{\Phi}}_{Y,j_{2}} -\sum_{j_{2}=1}^{k_{Y} } {\lambda}_{Y,j_{2}}^{-1}\left\{ \sum_{j_1=1}^{\infty}h_{j_{1}} \left(\frac{1}{n}\sum_{i=1}^{n}\xi_{ij_{1}}\eta_{ij_{2}} \right) \right\} {\mathbf{\Phi}}_{Y,j_{2}}\\
					=&\sum_{j_{2}=1}^{k_{Y} } {\lambda}_{Y,j_{2}}^{-1}\left\{ \sum_{j_1=1}^{\infty}h_{j_{1}} \left(\gamma_{j_{1}j_{2}}- \frac{1}{n}\sum_{i=1}^{n}\xi_{ij_{1}}\eta_{ij_{2}} \right) \right\} {\mathbf{\Phi}}_{Y,j_{2}}\\
					&+\sum_{j_{2}=1}^{k_{Y} } \hat{\lambda}_{Y,j_{2}}^{-1} \left(\sum_{j_{1}=1}^{\infty}\gamma_{j_{1}j_{2}}h_{j_{1}}\right)\mathcal{P}\hat{\mathbf{\Phi}}_{Y,j_{2}} -\sum_{j_{2}=1}^{k_{Y} } {\lambda}_{Y,j_{2}}^{-1}\left(\sum_{j_{1}=1}^{\infty}\gamma_{j_{1}j_{2}}h_{j_{1}}\right) {\mathbf{\Phi}}_{Y,j_{2}}\\
					&+\sum_{j_{2}=1}^{k_{Y} } \hat{\lambda}_{Y,j_{2}}^{-1} \left\{ \sum_{j_1=1}^{\infty}h_{j_{1}} \left(\frac{1}{n}\sum_{i=1}^{n}\xi_{ij_{1}}\eta_{ij_{2}}- \gamma_{j_{1}j_{2}} \right) \right\}\mathcal{P}\hat{\mathbf{\Phi}}_{Y,j_{2}} \\
					:=&J_{11}+J_{12}+J_{13}.
				\end{aligned}
			\end{equation}
			For $J_{11}$, it is of order 
			\begin{equation}\label{eq:fpc-3}
				\begin{aligned}
					\mathbf{E}\|J_{11}\|_{\mu_{Y}}^2=&\sum_{j_{2}=1}^{k_{Y} } {\lambda}_{Y,j_{2}}^{-2}\mathbf{E}\left\{ \sum_{j_1=1}^{\infty}h_{j_{1}} \left(\gamma_{j_{1}j_{2}}- \frac{1}{n}\sum_{i=1}^{n}\xi_{ij_{1}}\eta_{ij_{2}} \right) \right\}^2\\
					\leqslant& \sum_{j_{2}=1}^{k_{Y} } {\lambda}_{Y,j_{2}}^{-2} \sum_{j_{1}=1}^{\infty}\mathbf{E}\left(\gamma_{j_{1}j_{2}}- \frac{1}{n}\sum_{i=1}^{n}\xi_{ij_{1}}\eta_{ij_{2}} \right) ^2\\
					\leqslant&\sum_{j_{2}=1}^{k_{Y} } {\lambda}_{Y,j_{2}}^{-2} \sum_{j_{1}=1}^{\infty}\frac{C\lambda_{X,j_{1}}\lambda_{Y,j_{2}}}{n}\\
					=&\frac{C}{n} \sum_{j_{2}=1}^{k_{Y} } {\lambda}_{Y,j_{2}}^{-1} \sum_{j_{1}=1}^{\infty}\lambda_{X,j_{1}}=O_{p}\left(\frac{k_{Y}^{a_{Y}+1}}{n} \right),
				\end{aligned}
			\end{equation}
			and note that under the assumption \hyperref[asm:b1]{B.1}, $k_{Y}^{a_{Y}+1}/n=n^{-(2b_{Y}-1)/(a_{Y}+2b_{Y})}$.

			For $J_{13}$, we have
			\begin{equation}\label{eq:fpc-6}
				\begin{aligned}
					\mathbf{E}\|J_{13}\|_{\mu_{Y}}^2=&\mathbf{E}\sum_{j_{2}=1}^{k_{Y} } \hat{\lambda}_{Y,j_{2}}^{-2}\sum_{j_1=1}^{\infty}\left(\frac{1}{n}\sum_{i=1}^{n}\xi_{ij_{1}}\eta_{ij_{2}}- \gamma_{j_{1}j_{2}} \right)^2\\
					\leqslant& 4 \sum_{j_{2}=1}^{k_{Y} } {\lambda}_{Y,j_{2}}^{-2}\sum_{j_1=1}^{\infty}\mathbf{E}\left(\frac{1}{n}\sum_{i=1}^{n}\xi_{ij_{1}}\eta_{ij_{2}}- \gamma_{j_{1}j_{2}} \right)^2\\
					\leqslant& 4C\sum_{j_{2}=1}^{k_{Y} } {\lambda}_{Y,j_{2}}^{-2}\sum_{j_1=1}^{\infty}\frac{\lambda_{X,j_{1}}\lambda_{Y,j_{2}}}{n}\\
					=&\frac{4C}{n} \sum_{j_{2}=1}^{k_{Y} } {\lambda}_{Y,j_{2}}^{-1} \sum_{j_{1}=1}^{\infty}\lambda_{X,j_{1}}=O_{p}\left(\frac{k_{Y}^{a_{Y}+1}}{n} \right).
				\end{aligned}
			\end{equation}
		
		For $J_{12}$, we divide it into two components by  
		\begin{align*}
			J_{12}=&\sum_{j_{2}=1}^{k_{Y} } \hat{\lambda}_{Y,j_{2}}^{-1} \left(\sum_{j_{1}=1}^{\infty}\gamma_{j_{1}j_{2}}h_{j_{1}}\right)\mathcal{P}\hat{\mathbf{\Phi}}_{Y,j_{2}} -\sum_{j_{2}=1}^{k_{Y} } {\lambda}_{Y,j_{2}}^{-1}\left(\sum_{j_{1}=1}^{\infty}\gamma_{j_{1}j_{2}}h_{j_{1}}\right) {\mathbf{\Phi}}_{Y,j_{2}}\\
			=&\sum_{j_{2}=1}^{k_{Y} } \hat{\lambda}_{Y,j_{2}}^{-1} \left(\sum_{j_{1}=1}^{\infty}\gamma_{j_{1}j_{2}}h_{j_{1}}\right)\mathcal{P}\hat{\mathbf{\Phi}}_{Y,j_{2}}-\sum_{j_{2}=1}^{k_{Y} }{\lambda}_{Y,j_{2}}^{-1} \left(\sum_{j_{1}=1}^{\infty}\gamma_{j_{1}j_{2}}h_{j_{1}}\right)\mathcal{P}\hat{\mathbf{\Phi}}_{Y,j_{2}}\\
			&+\sum_{j_{2}=1}^{k_{Y} }{\lambda}_{Y,j_{2}}^{-1} \left(\sum_{j_{1}=1}^{\infty}\gamma_{j_{1}j_{2}}h_{j_{1}}\right)\left(\mathcal{P}\hat{\mathbf{\Phi}}_{Y,j_{2}} -{\mathbf{\Phi}}_{Y,j_{2}}\right)\\
			:=&J_{121}+J_{122}. 
		\end{align*}
		For the first component,
		\begin{equation*}
			\begin{aligned}
				\|J_{121}\|_{\mu_{Y}}^2=&\sum_{j_{2}=1}^{k_{Y} }\left(\hat{\lambda}_{Y,j_{2}}^{-1} - {\lambda}_{Y,j_{2}}^{-1} \right)^2 \left(\sum_{j_1=1}^{\infty}h_{j_{1}} \gamma_{j_{1}j_{2}} \right) ^2\\
				\leqslant&\sum_{j_{2}=1}^{k_{Y} }\left(\hat{\lambda}_{Y,j_{2}}^{-1} - {\lambda}_{Y,j_{2}}^{-1} \right)^2 \sum_{j_{1}=1}^{\infty} \gamma_{j_{1}j_{2}}^2\\
				\leqslant&\sup_{j_{2}\leqslant k_{Y}} \left(\hat{\lambda}_{Y,j_{2}}^{-1} - {\lambda}_{Y,j_{2}}^{-1} \right)^2 \sum_{j_{2}=1}^{k_{Y} }j^{-2a_{Y}-2b_{Y}}\\
				=&O_{p}\left( \frac{k_{Y}^{4a_{Y}}}{n} \right)O(k_{Y}^{1-2a_{Y}-2b_{Y}}) =o_{p}\left( \frac{k^{a_{Y}+1}}{n} \right),
			\end{aligned}
		\end{equation*}
		where the last inequality is due to the second assertion in Lemma \ref{lem:eigvalue} and the assumption  \hyperref[asm:b1]{B.1}. For the second component, we have
		\begin{equation*}
			\begin{aligned}
				\mathbf{E}\|J_{122}\|_{\mu_{Y}}^2=&\mathbf{E}\left\|\sum_{j_{2}=1}^{k_{Y} }{\lambda}_{Y,j_{2}}^{-1} \left(\sum_{j_{1}=1}^{\infty}\gamma_{j_{1}j_{2}}h_{j_{1}}\right)\left(\mathcal{P}\hat{\mathbf{\Phi}}_{Y,j_{2}} -{\mathbf{\Phi}}_{Y,j_{2}}\right) \right\|_{\mu_{Y}}^2\\
				\leqslant& k_{Y}\sum_{j_{2}=1}^{k_{Y} } {\lambda}_{Y,j_{2}}^{-2} \left(\sum_{j_{1}=1}^{\infty}\gamma_{j_{1}j_{2}}h_{j_{1}}\right)^2 \mathbf{E}\left\|\mathcal{P}\hat{\mathbf{\Phi}}_{Y,j_{2}} -{\mathbf{\Phi}}_{Y,j_{2}} \right\|_{\mu_{Y}}^2\\
				\leqslant&C\frac{k_{Y}}{n}\sum_{j_{2}=1}^{k_{Y} } {\lambda}_{Y,j_{2}}^{-2} j_{2}^2\sum_{j_{1}=1}^{\infty}\gamma_{j_{1}j_{2}}^2\\
				=&O\left(\frac{k_{Y}^{4-2b_{Y}}}{n}\right) = o\left( \frac{k_{Y}^{a_{Y}+1}}{n} \right),
			\end{aligned}
		\end{equation*}
		where the last inequality follows from the part (d) of Theorem \ref{thm:meancov}, Lemma \ref{lem:eigap} and $b_{Y}>a_{Y}/2+1$. Together with the first component, this shows that $\|J_{12}\|^2_{\mu_{Y}}=o_p(k^{a_Y+1}/n)$. 			
			Combing this with \eqref{eq:fpc-2} to \eqref{eq:fpc-6}, we deduce that $\|J_{1}\|^2_{\mu_{Y}}=O_{p}\left({k^{(a_{Y}+1})}/{n} \right)=O_p\left(n^{-(2b_{Y}-1)/(a_{Y}+2b_{Y})}\right) $. 
			
			Note that $\|J_{2}\|_{\mu_{Y}}^2=\|J_{11}\|_{\mu_{Y}}^2=O_p\left(n^{-(2b_{Y}-1)/(a_{Y}+2b_{Y})}\right)$ and for $J_{3}$, 
			\begin{equation}\label{eq:fpc-7}
				\begin{aligned}
					\|J_{3}\|_{\mu_{Y}}^2=&\sum_{j_{2}=k_{Y}+1}^{\infty} {\lambda}_{Y,j_{2}}^{-2} \left( \sum_{j_{1}=1}^{\infty} h_{j_{1}} \gamma_{j_{1}j_{2}} \right)^2\\
					\leqslant& \sum_{j_{2}=k_{Y}+1}^{\infty} {\lambda}_{Y,j_{2}}^{-2} \sum_{j_{1}=1}^{\infty} \gamma_{j_{1}j_{2}}^2\\
					\leqslant& \sum_{j_{2}=k_{Y}+1}^{\infty} j_{2}^{2a_{Y}}j_{2}^{-2a_{Y}-2b_{Y}}=O\left(k_{Y}^{1-2b_{Y}} \right)=O\left(n^{-\frac{2b_{Y}-1}{a_{Y}+2b_{Y}}} \right) . 
				\end{aligned}	
			\end{equation}
		Consequently, $\lnorm(\mathcal{P} \hat{\mathbf{C}}_{Y,k_{Y}}^{-1})(\mathcal{P}\hat{\mathbf{C}}_{YX})-\mathbf{C}_{Y}^{-1}\mathbf{C}_{YX} \rnorm_{\mathfrak{B}(\mu_{X},\mu_{Y})}^2=O_{p}(n^{-(2b_{Y}-1)/(a_{Y}+2b_{Y})}) $ according to \eqref{eq:fpc-1}. Similarly, we can show $\lnorm(\mathcal{P} \hat{\mathbf{C}}_{X,k_{X}}^{-1})(\mathcal{P}\hat{\mathbf{C}}_{XY})-\mathbf{C}_{X}^{-1}\mathbf{C}_{XY} \rnorm_{\mathfrak{B}(\mu_{Y},\mu_{X})}^2=O_{p}(n^{-(2b_{X}-1)/(a_{X}+2b_{X})}) $. By the part (f) of Proposition \ref{prop:trans} and the fact that $\mathcal{P} \hat{\mathbf{C}}_{X,k_{X}}^{-1} \mathcal{P}\hat{\mathbf{C}}_{XY}=O_{p}(1)$, we further deduce that
			$$
			\begin{aligned}
				&\lnorm \mathbf{C}_{X}^{-1}\mathbf{C}_{XY}\mathbf{C}_{Y}^{-1}\mathbf{C}_{YX}-\mathcal{P}\hat{\mathbf{C}}_{X,k}^{-1}\hat{\mathbf{C}}_{XY}\hat{\mathbf{C}}_{Y,k}^{-1}\hat{\mathbf{C}}_{YX}\rnorm^{2}_{\mathfrak{B}(\mu_{X},\mu_{X})}\\=
				&\lnorm\mathbf{C}_{X}^{-1} \mathbf{C}_{XY}\mathbf{C}_{Y}^{-1} \mathbf{C}_{YX} -\mathcal{P} \hat{\mathbf{C}}_{X,k_{X}}^{-1} \mathcal{P}\hat{\mathbf{C}}_{XY}\mathcal{P} \hat{\mathbf{C}}_{Y,k_{Y}}^{-1} \mathcal{P}\hat{\mathbf{C}}_{YX}\rnorm_{\mathfrak{B}(\mu_{X},\mu_{X})}^{2} \\
				\leq & 2\lnorm\mathbf{C}_{X}^{-1} \mathbf{C}_{XY}\mathbf{C}_{Y}^{-1} \mathbf{C}_{YX} -\mathcal{P} \hat{\mathbf{C}}_{X,k_{X}}^{-1} \mathcal{P}\hat{\mathbf{C}}_{XY} \mathbf{C}_{Y}^{-1} \mathbf{C}_{YX} \rnorm_{\mathfrak{B}(\mu_{X},\mu_{X})}^{2} \\
				&+2\lnorm\mathcal{P} \hat{\mathbf{C}}_{X,k_{X}}^{-1} \mathcal{P}\hat{\mathbf{C}}_{XY} \mathbf{C}_{Y}^{-1} \mathbf{C}_{YX} -\mathcal{P} \hat{\mathbf{C}}_{X,k_{X}}^{-1} \mathcal{P}\hat{\mathbf{C}}_{XY}\mathcal{P} \hat{\mathbf{C}}_{Y,k_{Y}}^{-1} \mathcal{P}\hat{\mathbf{C}}_{YX}\rnorm_{\mathfrak{B}(\mu_{X},\mu_{X})}^{2} \\
				\leq & 2\lnorm\mathbf{C}_{X}^{-1} \mathbf{C}_{XY}-\mathcal{P} \hat{\mathbf{C}}_{X,k_{X}}^{-1} \mathcal{P}\hat{\mathbf{C}}_{XY}\rnorm^{2}_{\mathfrak{B}(\mu_{Y},\mu_{X})}\lnorm\mathbf{C}_{Y}^{-1} \mathbf{C}_{YX} \rnorm^{2}_{\mathfrak{B}(\mu_{X},\mu_{Y})} \\
				&+2\lnorm\mathcal{P} \hat{\mathbf{C}}_{X,k_{X}}^{-1} \mathcal{P}\hat{\mathbf{C}}_{XY}\rnorm^{2}_{\mathfrak{B}(\mu_{Y},\mu_{X})}\lnorm\mathbf{C}_{Y}^{-1} \mathbf{C}_{YX} -\mathcal{P} \hat{\mathbf{C}}_{Y,k_{Y}}^{-1} \mathcal{P}\hat{\mathbf{C}}_{YX}\rnorm_{\mathfrak{B}(\mu_{X},\mu_{Y})}^{2} \\
				=& O_{p}\left(\max \left\{n^{-\left(2 b_{X} -1\right) /\left( a_{X} +2 b_{X} \right)}, n^{-\left(2 b_{Y} -1\right) /\left( a_{Y} +2 b_{Y} \right)}\right\}\right).
			\end{aligned}
			$$
			Now we adopt the perturbation argument in \cite{bosq2000} to establish $$\|\mathcal{P} \hat{U}-U\|_{\mu_{X}}^2 =O_{p}\left(\max \left\{n^{-\left(2 b_{X} -1\right) /\left( a_{X} +2 b_{X} \right)}, n^{-\left(2 b_{Y} -1\right) /\left( a_{Y} +2 b_{Y} \right)}\right\}\right),$$ and complete the proof by 
			\begin{align*}
				&\|\mathcal{P}\hat{\mathbf{C}}_{Y,k_{Y}}^{-1}\hat{\mathbf{C}}_{YX}\mathcal{P}\hat{U}-\mathbf{C}_{Y}^{-1}\mathbf{C}_{YX}U \|^2_{\mu_{Y}}\\
				\leqslant&2\lnorm(\mathcal{P} \hat{\mathbf{C}}_{Y,k_{Y}}^{-1})(\mathcal{P}\hat{\mathbf{C}}_{YX})-\mathbf{C}_{Y}^{-1}\mathbf{C}_{YX} \rnorm^{2}_{\mathfrak{B}(\mu_{X},\mu_{Y})} \|\mathcal{P}\hat{U} \|_{\mu_{X}}^2+2\lnorm\mathbf{C}_{Y}^{-1}\mathbf{C}_{YX} \rnorm^2_{\mathfrak{B}(\mu_{X},\mu_{Y})}\|\mathcal{P}\hat{U}-U \|_{\mu_{X}}^2\\
				=&O_{p}\left(\max \left\{n^{-\left(2 b_{X} -1\right) /\left( a_{X} +2 b_{X} \right)}, n^{-\left(2 b_{Y} -1\right) /\left( a_{Y} +2 b_{Y} \right)}\right\}\right).
			\end{align*}\QED	
		\end{proof}
		\begin{proof}[of Theorem \ref{thm:ridge}]
			By the fact $\mathcal{P}\hat{\mathbf{id}}_{Y}=\mathbf{id}_{Y}$, it is sufficient to show 
			\begin{equation}\label{eq:ridge-1}
				\sup_{h\in \mathscr{T}(\mu_{X})\atop \|h\|_{\mu_{X}}=1 } \left\|(\mathcal{P}\hat{\mathbf{C}}_{Y}+\epsilon_{Y}{\mathbf{id}}_{Y})^{-1}(\mathcal{P}\hat{\mathbf{C}}_{YX})h-\mathbf{C}_{Y}^{-1}\mathbf{C}_{YX}h \right\|^{2}_{\mu_{Y} }=O_{p}(n^{-(2b_{Y}-1)/(a_{Y}+2b_{Y})} ), 
			\end{equation}
			as then the proof can be completed in analogy to the proof of Theorem \ref{thm:fpc}.
			We start with the following decomposition 
			\begin{equation}\label{eq:ridge-2}
				\begin{aligned}
					&(\mathcal{P}\hat{\mathbf{C}}_{Y}+\epsilon_{Y}{\mathbf{id}}_{Y})^{-1}(\mathcal{P}\hat{\mathbf{C}}_{YX})h-\mathbf{C}_{Y}^{-1}\mathbf{C}_{YX}h \\=&(\mathcal{P}\hat{\mathbf{C}}_{Y}+\epsilon_{Y}{\mathbf{id}}_{Y})^{-1}(\mathcal{P}\hat{\mathbf{C}}_{YX})h-({\mathbf{C}}_{Y}+\epsilon_{Y}{\mathbf{id}}_{Y})^{-1}(\mathcal{P}\hat{\mathbf{C}}_{YX})h\\
					&+({\mathbf{C}}_{Y}+\epsilon_{Y}{\mathbf{id}}_{Y})^{-1}(\mathcal{P}\hat{\mathbf{C}}_{YX})h-({\mathbf{C}}_{Y}+\epsilon_{Y}{\mathbf{id}}_{Y})^{-1}({\mathbf{C}}_{YX})h\\
					&+({\mathbf{C}}_{Y}+\epsilon_{Y}{\mathbf{id}}_{Y})^{-1}({\mathbf{C}}_{YX})h-\mathbf{C}_{Y}^{-1}\mathbf{C}_{YX}h \\
					:=& K_{1}+K_{2}+K_{3}.
				\end{aligned}
			\end{equation}
			
			Since 
			$$
			\begin{aligned}
				&\mathbf{E}\left\|({\mathbf{C}}_{Y}+\epsilon_{Y}{\mathbf{id}}_{Y})^{-1}({\mathbf{C}}_{Y} -\mathcal{P}\hat{\mathbf{C}}_{Y}) g \right\|_{\mu_{Y}}^2\\
				=&\mathbf{E}\left\|\sum_{j =1}^{\infty} \frac{1}{\lambda_{Y,j }+\epsilon_{Y}}g_{j }\lambda_{j }\mathbf{\Phi}_{Y,j }-\sum_{j_{1}=1}^{\infty} \frac{1}{\lambda_{Y,j_{1} }+\epsilon_{Y}}\sum_{j_{2}=1}^{\infty}\frac{1}{n}\sum_{i=1}^{n}\eta_{ij_{1}}\eta_{ij_{2}}g_{j_{2}}\mathbf{\Phi}_{Y,j_{1}} \right\|_{\mu_{Y}}^2\\
				=&\sum_{j_{1}=1}^{\infty}\left(\frac{1}{\lambda_{Y,j_{1} }+\epsilon_{Y}}\right)^{2}\mathbf{E}\left\{ \left(\sum_{j_{2}=1}^{\infty}\frac{1}{n}\sum_{i=1}^{n}\eta_{ij_{1}}\eta_{ij_{2}}g_{j_{2}} \right)-g_{j_{1}}\lambda_{Y,j_{1}} \right\}^2\\
				=&\sum_{j_{1}=1}^{\infty}\left(\frac{1}{\lambda_{Y,j_{1} }+\epsilon_{Y}}\right)^{2}\left\{\frac{2\lambda_{Y,j_{1}}^2g_{j_{1}}^2}{n}+\sum_{j_{2}=1}^{\infty}\frac{\lambda_{Y,j_{1}}\lambda_{Y,j_{2}}}{n}g_{j_{2}}^2 \right\}\\
				\lesssim&\frac{1}{n} \sum_{j=1}^{\infty}\frac{\lambda_{Y,j} }{(\lambda_{Y,j }+\epsilon_{Y})^2}=O_{p}\left(\epsilon_{Y}^{-(1+{1}/{a_{Y}})}/n \right),
			\end{aligned}
			$$
			for $g=\sum_{j=1}^{\infty}g_{j}\mathbf{\Phi}_{Y,j}$ with $\sum_{j=1}^{\infty}g_{j}^2 =1$ and due to the fact that $(\mathcal{P}\hat{\mathbf{C}}_{Y}+\epsilon_{Y}{\mathbf{id}}_{Y})^{-1}\hat{\mathbf{C}}_{YX} $ is bounded in probability, we have $\|K_{1}\|_{\mu_{Y}}^2= O_{p}\left(\epsilon_{Y}^{-(1+{1}/{a_{Y}})} /n\right)$.
			
			For the second term, 
			\begin{equation}\label{eq:ridge-K2}
				\begin{aligned}
					\mathbf{E}\|K_{2}\|_{\mu_{Y}}^2=&\mathbf{E} \left\|({\mathbf{C}}_{Y}+\epsilon_{Y}{\mathbf{id}}_{Y})^{-1}(\mathcal{P}\hat{\mathbf{C}}_{YX})h-({\mathbf{C}}_{Y}+\epsilon_{Y}{\mathbf{id}}_{Y})^{-1}({\mathbf{C}}_{YX})h \right\|_{\mu_{Y}}^{2}\\
					=&\mathbf{E}\left\|\sum_{j_2=1}^{\infty}\frac{1}{\lambda_{Y,j_{2}}+\epsilon_{Y}}\sum_{j_{1}=1}^{\infty}\left(\frac{1}{n}\sum_{i=1}^{n}\xi_{ij_{1}}\eta_{ij_{2}}-\gamma_{j_{1}j_{2}} \right)h_{j_{1}}\mathbf{\Phi}_{Y,j_{2}} \right\|_{\mu_{Y}}^{2}\\
					=&\sum_{j_2=1}^{\infty}\left(\frac{1}{\lambda_{Y,j_{2} }+\epsilon_{Y}}\right)^{2}\mathbf{E}\left\{\sum_{j_{1}=1}^{\infty}\left(\frac{1}{n}\sum_{i=1}^{n}\xi_{ij_{1}}\eta_{ij_{2}}-\gamma_{j_{1}j_{2}} \right)h_{j_{1}} \right\}^{2}\\
					\leqslant& 	\sum_{j_2=1}^{\infty}\left(\frac{1}{\lambda_{Y,j_{2} }+\epsilon_{Y}}\right)^{2}\mathbf{E}	\sum_{j_{1}=1}^{\infty}\left(\frac{1}{n}\sum_{i=1}^{n}\xi_{ij_{1}}\eta_{ij_{2}}-\gamma_{j_{1}j_{2}} \right)^2	\sum_{j_{1}=1}^{\infty}h_{j_{1}}^2\\
					\leqslant& C\sum_{j_2=1}^{\infty}\left(\frac{1}{\lambda_{Y,j_{2} }+\epsilon_{Y}}\right)^{2}\sum_{j_{1}=1}^{\infty}\frac{\lambda_{X,j_{1}}\lambda_{Y,j_{2}}}{n}\\
					\leqslant&\frac{2C}{n}\sum_{j_{1}=1}^{\infty}\sum_{j_{2}=1}^{\infty}\frac{\lambda_{X,j_{1}}\lambda_{Y,j_{2}}}{(\lambda_{Y,j_{2} }+\epsilon_{Y})^2}=O_{p}\left(\epsilon_{Y}^{-(1+{1}/{a_{Y}})}/n \right),
				\end{aligned}
			\end{equation}
			The second inequality in the above follows from the condition  \hyperref[asm:a2]{$\textup{A.2}$} and Cauchy--Schwarz inequality that together imply
					$$\mathbf{E}\left(\frac{1}{n}\sum_{i=1}^{n}\xi_{ij_{1}}\eta_{ij_{2}}-\gamma_{j_{1}j_{2}} \right)^2=\frac{\mathbf{E}(\xi_{ij_{1}}\eta_{ij_{2}})^2}{n}\leqslant\frac{1}{n}\sqrt{\mathbb{E}\xi_{ij_{1}}^{4}\eta_{ij_{2}}^{4}}\leqslant\frac{C}{n}\lambda_{X,j_{1}}\lambda_{Y,j_{2}}.	 $$

			Finally, 
			$$
			\begin{aligned}
				\|K_{3}\|_{\mu_{Y}}^{2}=&\|({\mathbf{C}}_{Y}+\epsilon_{Y}{\mathbf{id}}_{Y})^{-1}({\mathbf{C}}_{YX})h-\mathbf{C}_{Y}^{-1}\mathbf{C}_{YX}h \|_{\mu_{Y}}^2\\
				=&\left\| \sum_{j_{1}=1}^{\infty}\sum_{j_{2}=1}^{\infty}\left(\frac{1}{\lambda_{Y,j_{2}}+\epsilon_{Y}}-\frac{1}{\lambda_{Y,j_{2}}} \right)\gamma_{j_{1}j_{2}}h_{j_{1}}\mathbf{\Phi}_{Y,j_{2}} \right\|_{\mu_{Y}}^2\\
				=&\sum_{j_{2}=1}^{\infty} \left(\frac{1}{\lambda_{Y,j_{2}}+\epsilon_{Y}}-\frac{1}{\lambda_{Y,j_{2}}} \right)^2\left(\sum_{j_{1}=1} ^{\infty}\gamma_{j_{1}j_{2}}h_{j_{1}} \right)^{2}\\
				\leqslant&\sum_{j_{2}=1}^{\infty}\frac{\epsilon_{Y}^2}{(\lambda_{Y,j_{2}}+\epsilon_{Y} )^2\lambda_{Y,j_{2}}^2 }\sum_{j_{1}=1} ^{\infty}\gamma_{j_{1}j_{2}}^2\\
				\leqslant& C\epsilon_{Y}^2 \sum_{j_{2}=1}^{\infty}\frac{j_{2}^{-2b_Y}}{(\lambda_{Y,j_{2}}+\epsilon_{Y} )^2}=O\left(\epsilon_{Y}^{({2b_{Y}-1)}/{a_{Y}}} \right). 
			\end{aligned}
			$$
			When $\epsilon_{Y}\asymp n^{-a_{Y}/(a_{Y}+2b_{Y})}$, we have $\epsilon_{Y}^{-(1+{1}/{a_{Y}})}/n =\epsilon_{Y}^{({2b_{Y}-1)}/{a_{Y}}}$, and then $K_{1},K_{2}$ and $K_{3}$ are of the same order $n^{-(2b_{Y}-1)/(a_{Y}+2b_{Y})} $. This establishes  \eqref{eq:ridge-1} and further
			$$\lnorm (\mathcal{P}\hat{\mathbf{C}}_{Y}+\epsilon_{Y}{\mathbf{id}}_{Y})^{-1}(\mathcal{P}\hat{\mathbf{C}}_{YX})-\mathbf{C}_{Y}^{-1}\mathbf{C}_{YX} \rnorm_{\mathfrak{B}(\mu_{X},\mu_{Y})} ^{2}=O_{p}(n^{-(2b_{Y}-1)/(a_{Y}+2b_{Y})} ).$$
			\QED
		\end{proof}
	
	\begin{remark}\label{rem:minimax} 
		In the equation (7) in \cite{lian2014} and the proof of Theorem 2 in \cite{zhou2020}, both about canonical correlation analysis for Euclidean functional data,  the following unjustified inequality
		$$\mathbf{E} \left\|({\mathbf{C}}_{Y}+\epsilon_{Y}{\mathbf{id}}_{Y})^{-1}(\hat{\mathbf{C}}_{YX})h-({\mathbf{C}}_{Y}+\epsilon_{Y}{\mathbf{id}}_{Y})^{-1}({\mathbf{C}}_{YX})h \right\|^{2}\leqslant\frac{1}{n}\mathbf{E}\left\|({\mathbf{C}}_{Y}+\epsilon_{Y}{\mathbf{id}}_{Y})^{-1}({\mathbf{C}}_{YX})h \right\|^{2},$$
		translated from  \cite{lian2014} and \cite{zhou2020} into our language and notation,  taking into account that parallel transport is not needed for the Euclidean case, 
		is used to derive $$\mathbf{E}\|K_{2}\|^2\leqslant\frac{1}{n}\sum_{j_{1}=1}^{\infty}\sum_{j_{2}=1}^{\infty}\frac{\gamma_{j_{1}j_{2}}^{2} }{(\lambda_{Y,j_{2} }+\epsilon_{Y})^2}. $$
		However, this seems incorrect. For a counterexample, take $(\xi_{j_{1}},\eta_{j_{2}})$ to follow a joint Gaussian distribution with zero mean and covariance matrix $$\begin{pmatrix}
			\lambda_{X,j_{1}},\gamma_{j_{1}j_{2}}\\\gamma
			_{j_{1}j_{2}},\lambda_{Y,j_{2}}
		\end{pmatrix}. $$ Let $h=\mathbf{\Phi}_{X,1}$. Then 
		$\mathbf{E}\|K_{2}\|^2=\sum_{j_{1}=1}^{\infty}\sum_{j_{2}=1}^{\infty}({\lambda_{X,j_1}\lambda_{Y,j_{2}}+2\gamma_{j_1j_{2}}^2})/{(\lambda_{Y,j_{2} }+\epsilon_{Y})^2} $, {and under the condition \hyperref[asm:b2']{$\textup{B.2}^{'}$}, the term  $n^{-1}\sum_{j_1=1}^\infty\sum_{j_{2}=1}^{\infty}{\lambda_{X,j_1}\lambda_{Y,j_{2}}}/{(\lambda_{Y,j_{2} }+\epsilon_{Y})^2} $ is asymptotically strictly larger than $\frac{1}{n}\sum_{j_{1}=1}^{\infty}\sum_{j_{2}=1}^{\infty}\frac{\gamma_{j_{1}j_{2}}^{2} }{(\lambda_{Y,j_{2} }+\epsilon_{Y})^2}$}. 
		Consequently, there is a gap in the proof of Theorem 1 of \cite{lian2014}, which however may be filled by using our arguments in the above. In contrast, the proof of Theorem 2  in both \cite{lian2014} and \cite{zhou2020} about the convergence rate with respect to an RHKS (reproducing kernel Hilbert space) norm, may not be fixed in the same way. To see this, let $\mathbb{G}_{X}$ be the RKHS generated by the kernel $\mathbf{C}_{X}$. Then a bound on $\mathbf{E}\|K_{2}\|^{2}_{\mathbb{G}_{X}}$ is given by
		$$\begin{aligned}
			\mathbf{E}\|K_{2}\|^{2}_{\mathbb{G}_{X}}\leqslant&\sum_{j_{1}=1}^{\infty}\sum_{j_2=1}^{\infty}\frac{1}{\lambda_{X,j_{1}}}\left(\frac{1}{\lambda_{Y,j_{2} }+\epsilon_{Y}}\right)^{2}\mathbf{E}	\left(\frac{1}{n}\sum_{i=1}^{n}\xi_{ij_{1}}\eta_{ij_{2}}-\gamma_{j_{1}j_{2}} \right)^2	\\
			\leqslant&\frac{1}{n}\sum_{j_{1}=1}^{\infty}\sum_{j_2=1}^{\infty}\frac{\lambda_{X,j_{1}}\lambda_{Y,{j_{2}}} }{\lambda_{X,j_{1}}}\left(\frac{1}{\lambda_{Y,j_{2} }+\epsilon_{Y}}\right)^{2},
		\end{aligned}$$
		which diverges with respect to $j_{1}$, contrasting the finite bound made in those works. 
	\end{remark}

		\section{Ancillary Lemmas}
		\begin{lemma}\label{lem:eigvalue}
			Suppose that $(\hat{\lambda}_{X,j},\mathbf{\hat{\Phi}}_{X,j})$ and $(\hat{\lambda}_{Y,j},\mathbf{\hat{\Phi}}_{Y,j}) $ are the eigen-systems of $\mathbf{\hat{C}}_{X}$ and $\mathbf{\hat{C}}_{Y} $, respectively. Under the assumptions \hyperref[asm:a1]{\textup{A.1}}--\hyperref[asm:a2]{\textup{A.2}}, we have 
			$$\sup_{j\leqslant n^{1/2a_{X}}}j^{-2a_{X}}\left|\hat{\lambda}_{X,j}^{-{1}}-\lambda_{X,j}^{-1} \right|=O_{p}(1/\sqrt{n}) \text{ and }\sup_{j\leqslant n^{1/2a_{Y}}}j^{-2a_{Y}}\left|\hat{\lambda}_{Y,j}^{-{1}}-\lambda_{Y,j}^{-1} \right|=O_{p}(1/\sqrt{n}).  $$
		\end{lemma}
		\begin{proof}[of Lemma \ref{lem:eigvalue}]
			By Taylor expansion, for any real numbers $x$ and $x_0$,
			$$
			x^{-1 }-x_{0}^{-1 }=- x_{0}^{- 2}\left(x-x_{0}\right)+ \hat{x}^{-3}\left(x-x_{0}\right)^{2},
			$$
			where $\hat{x}$ is some value between $x$ and $x_{0}$. Using this fact, we obtain
			$$
			\begin{aligned}
				\left|\hat{\lambda}_{X,j}^{-1 }-\lambda_{X,j}^{-1 }\right|=& \lambda_{X,j}^{- 2} O_{p}\left(\left|\hat{\lambda}_{X,j}-\lambda_{X,j}\right|\right) +\left(\lambda_{X,j}-O_{p}\left(\lnorm \Delta_{X} \rnorm \right)\right)^{-3} O_{p}\left(\left|\hat{\lambda}_{X,j}-\lambda_{X,j}\right|^{2}\right)\\
				=&O_{p}(j^{2a_{X}}/\sqrt{n} )+ O_{p}(j^{3a_{X}}/n ) ,
			\end{aligned}
			$$
			where $\Delta_{X}=\mathcal{P}_{\mathfrak{B(\hat\mu_{X},\hat\mu_{X})} }^{\mathfrak{B(\mu_{X},\mu_{X} )} }\mathbf{\hat C}_{X} -\mathbf{C}_{X} $ and the last equality follows from the part (d) of Theorem \ref{thm:meancov}. Under the condition $j\leqslant n^{1/2a_{X}} $, we have $j^{3a_{X}}/n =o(j^{2a_{X}}/\sqrt{n} ) $ and the first assertion follows. The second assertion is proved by a similar argument.\QED
		\end{proof}
		\begin{lemma}\label{lem:eigap}
			For $\eta_{X,k}=\inf_{j\neq k}|\lambda_{X,k}- \lambda_{X,j}| $ and $\eta_{Y,k}=\inf_{j\neq k}|\lambda_{Y,k}- \lambda_{Y,j}| $, under the  assumption \hyperref[asm:b1]{\textup{B.1}}, one has
			$$\mathbf{P}\left(\frac{1}{2}\eta_{X,k_{X}}>\lnorm \mathcal{P}\hat{\mathbf{C}}_{X}-\mathbf{C}_{X}\rnorm \right) \rightarrow 1\qquad\text{and}\qquad\mathbf{P}\left(\frac{1}{2}\eta_{Y,k_{Y}}>\lnorm \mathcal{P}\hat{\mathbf{C}}_{Y}-\mathbf{C}_{Y}\rnorm \right)\rightarrow 1.$$
		\end{lemma}
		\begin{proof}[of Lemma \ref{lem:eigap}]
			According to the part (d) of Theorem \ref{thm:meancov}, we have $\lnorm \mathcal{P}\hat{\mathbf{C}}_{X}-\mathbf{C}_{X}\rnorm=O_{p}(n^{-1/2})$. Under the assumption \hyperref[asm:b1]{B.1}, $\eta_{X,k_{X}}\leqslant k_{X}^{-a_{X}-1}$ and $n^{1/2}k_{X}^{-a_{X}-1}\rightarrow\infty$, which implies that  $\mathbf{P}\left(\frac{1}{2}\eta_{X,k_{X}}>\lnorm \mathcal{P}\hat{\mathbf{C}}_{X}-\mathbf{C}_{X}\rnorm \right)\rightarrow 1$. The second statement follows analogously.
			\QED
		\end{proof}

	\end{appendices}

	\bigskip
	\bibliographystyle{Chicago}
	\bibliography{CCA-JASA-5_18}

\begin{thebibliography}{}

\bibitem[\protect\citeauthoryear{Aldous}{Aldous}{1976}]{aldous1976}
Aldous, D.~J. (1976).
\newblock A characterisation of hilbert space using the central limit theorem.
\newblock {\em Journal of the London Mathematical Society\/}~{\em 2\/}(2),
  376--380.

\bibitem[\protect\citeauthoryear{Ambrosio, Gigli, and Savar{\'e}}{Ambrosio
  et~al.}{2008}]{ambrosio2008}
Ambrosio, L., N.~Gigli, and G.~Savar{\'e} (2008).
\newblock {\em Gradient flows: in metric spaces and in the space of probability
  measures}.
\newblock Springer Science \& Business Media.

\bibitem[\protect\citeauthoryear{Benamou and Brenier}{Benamou and
  Brenier}{2000}]{benamou2000}
Benamou, J.-D. and Y.~Brenier (2000).
\newblock A computational fluid mechanics solution to the monge-kantorovich
  mass transfer problem.
\newblock {\em Numerische Mathematik\/}~{\em 84\/}(3), 375--393.

\bibitem[\protect\citeauthoryear{Bigot, Gouet, Klein, L{\'o}pez, et~al.}{Bigot
  et~al.}{2017}]{bigot2017}
Bigot, J., R.~Gouet, T.~Klein, A.~L{\'o}pez, et~al. (2017).
\newblock Geodesic {PCA} in the wasserstein space by convex {PCA}.
\newblock In {\em Annales de l'Institut Henri Poincar{\'e}, Probabilit{\'e}s et
  Statistiques}, Volume~53, pp.\  1--26. Institut Henri Poincar{\'e}.

\bibitem[\protect\citeauthoryear{Bosq}{Bosq}{2000}]{bosq2000}
Bosq, D. (2000).
\newblock {\em Linear processes in function spaces: theory and applications},
  Volume 149.
\newblock Springer Science \& Business Media.

\bibitem[\protect\citeauthoryear{Chen, Lin, and M{\"u}ller}{Chen
  et~al.}{2020}]{chen2020}
Chen, Y., Z.~Lin, and H.-G. M{\"u}ller (2020).
\newblock Wasserstein regression.
\newblock {\em arXiv preprint arXiv:2006.09660\/}.

\bibitem[\protect\citeauthoryear{Dai}{Dai}{2021}]{dai2021}
Dai, X. (2021).
\newblock Statistical inference on the hilbert sphere with application to
  random densities.
\newblock {\em arXiv preprint arXiv:2101.00527\/}.

\bibitem[\protect\citeauthoryear{Dai, Lin, and M{\"u}ller}{Dai
  et~al.}{2020}]{Dai2020}
Dai, X., Z.~Lin, and H.-G. M{\"u}ller (2020).
\newblock Modeling sparse longitudinal data on riemannian manifolds.
\newblock {\em Biometrics\/}.

\bibitem[\protect\citeauthoryear{Dai and M{\"u}ller}{Dai and
  M{\"u}ller}{2018}]{dai2018}
Dai, X. and H.-G. M{\"u}ller (2018).
\newblock {Principal component analysis for functional data on Riemannian
  manifolds and spheres}.
\newblock {\em The Annals of Statistics\/}~{\em 46\/}(6B), 3334 -- 3361.

\bibitem[\protect\citeauthoryear{Dauxois, Pousse, and Romain}{Dauxois
  et~al.}{1982}]{dauxois1982}
Dauxois, J., A.~Pousse, and Y.~Romain (1982).
\newblock Asymptotic theory for the principal component analysis of a vector
  random function: some applications to statistical inference.
\newblock {\em Journal of multivariate analysis\/}~{\em 12\/}(1), 136--154.

\bibitem[\protect\citeauthoryear{Delaigle and Hall}{Delaigle and
  Hall}{2012}]{delaigle2012}
Delaigle, A. and P.~Hall (2012).
\newblock Achieving near perfect classification for functional data.
\newblock {\em Journal of the Royal Statistical Society: Series B (Statistical
  Methodology)\/}~{\em 74\/}(2), 267--286.

\bibitem[\protect\citeauthoryear{Dou, Pollard, Zhou, et~al.}{Dou
  et~al.}{2012}]{dou2012}
Dou, W.~W., D.~Pollard, H.~H. Zhou, et~al. (2012).
\newblock Estimation in functional regression for general exponential families.
\newblock {\em The Annals of Statistics\/}~{\em 40\/}(5), 2421--2451.

\bibitem[\protect\citeauthoryear{Eubank and Hsing}{Eubank and
  Hsing}{2008}]{eubank2008}
Eubank, R. and T.~Hsing (2008).
\newblock Canonical correlation for stochastic processes.
\newblock {\em Stochastic Processes and their Applications\/}~{\em 118\/}(9),
  1634--1661.

\bibitem[\protect\citeauthoryear{Ferraty and Vieu}{Ferraty and
  Vieu}{2006}]{Ferraty2006}
Ferraty, F. and P.~Vieu (2006).
\newblock {\em Nonparametric functional data analysis: theory and practice}.
\newblock Springer Science \& Business Media.

\bibitem[\protect\citeauthoryear{Fr{\'e}chet}{Fr{\'e}chet}{1948}]{frechet1948}
Fr{\'e}chet, M. (1948).
\newblock Les {\'e}l{\'e}ments al{\'e}atoires de nature quelconque dans un
  espace distanci{\'e}.
\newblock In {\em Annales de l'institut Henri Poincar{\'e}}, Volume~10, pp.\
  215--310.

\bibitem[\protect\citeauthoryear{Gangbo and McCann}{Gangbo and
  McCann}{1996}]{gangbo1996}
Gangbo, W. and R.~J. McCann (1996).
\newblock The geometry of optimal transportation.
\newblock {\em Acta Mathematica\/}~{\em 177\/}(2), 113--161.

\bibitem[\protect\citeauthoryear{Hall and Horowitz}{Hall and
  Horowitz}{2007}]{hall2007}
Hall, P. and J.~L. Horowitz (2007).
\newblock Methodology and convergence rates for functional linear regression.
\newblock {\em The Annals of Statistics\/}~{\em 35\/}(1), 70--91.

\bibitem[\protect\citeauthoryear{Hall and Hosseini-Nasab}{Hall and
  Hosseini-Nasab}{2006}]{hall2006}
Hall, P. and M.~Hosseini-Nasab (2006).
\newblock On properties of functional principal components analysis.
\newblock {\em Journal of the Royal Statistical Society. Series B. Statistical
  Methodology\/}~{\em 68\/}(1), 109--126.

\bibitem[\protect\citeauthoryear{He, M{\"u}ller, and Wang}{He
  et~al.}{2003}]{he2003}
He, G., H.-G. M{\"u}ller, and J.-L. Wang (2003).
\newblock Functional canonical analysis for square integrable stochastic
  processes.
\newblock {\em Journal of Multivariate Analysis\/}~{\em 85\/}(1), 54--77.

\bibitem[\protect\citeauthoryear{Horv{\'a}th and Kokoszka}{Horv{\'a}th and
  Kokoszka}{2012}]{Horvath2012}
Horv{\'a}th, L. and P.~Kokoszka (2012).
\newblock {\em Inference for functional data with applications}, Volume 200.
\newblock Springer Science \& Business Media.

\bibitem[\protect\citeauthoryear{Hsing and Eubank}{Hsing and
  Eubank}{2015}]{hsing2015}
Hsing, T. and R.~Eubank (2015).
\newblock {\em Theoretical foundations of functional data analysis, with an
  introduction to linear operators}, Volume 997.
\newblock John Wiley \& Sons.

\bibitem[\protect\citeauthoryear{James and Sugar}{James and
  Sugar}{2003}]{james2003}
James, G.~M. and C.~A. Sugar (2003).
\newblock Clustering for sparsely sampled functional data.
\newblock {\em Journal of the American Statistical Association\/}~{\em
  98\/}(462), 397--408.

\bibitem[\protect\citeauthoryear{Kantorovich}{Kantorovich}{2006}]{kantorovich2006}
Kantorovich, L.~V. (2006).
\newblock On a problem of monge.
\newblock {\em J. Math. Sci.(NY)\/}~{\em 133}, 1383.

\bibitem[\protect\citeauthoryear{Kokoszka and Reimherr}{Kokoszka and
  Reimherr}{2017}]{Kokoszka2017}
Kokoszka, P. and M.~Reimherr (2017).
\newblock {\em Introduction to functional data analysis}.
\newblock CRC press.

\bibitem[\protect\citeauthoryear{Leurgans, Moyeed, and Silverman}{Leurgans
  et~al.}{1993}]{leurgans1993}
Leurgans, S.~E., R.~A. Moyeed, and B.~W. Silverman (1993).
\newblock Canonical correlation analysis when the data are curves.
\newblock {\em Journal of the Royal Statistical Society. Series B.
  Methodological\/}~{\em 55\/}(3), 725--740.

\bibitem[\protect\citeauthoryear{Lian}{Lian}{2014}]{lian2014}
Lian, H. (2014).
\newblock Some asymptotic properties for functional canonical correlation
  analysis.
\newblock {\em Journal of Statistical Planning and Inference\/}~{\em 153},
  1--10.

\bibitem[\protect\citeauthoryear{Lin, Kong, and Wang}{Lin
  et~al.}{2021}]{lin-causal-inference}
Lin, Z., D.~Kong, and L.~Wang (2021).
\newblock Causal inference on non-linear spaces: Distribution functions and
  beyond.
\newblock {\em arXiv preprint arXiv:2101.01599\/}.

\bibitem[\protect\citeauthoryear{Lin and M{\"u}ller}{Lin and
  M{\"u}ller}{2019}]{lin-metric-tv}
Lin, Z. and H.-G. M{\"u}ller (2019).
\newblock Total variation regularized fr$\backslash$'echet regression for
  metric-space valued data.
\newblock {\em arXiv preprint arXiv:1904.09647\/}.

\bibitem[\protect\citeauthoryear{Lin, Shao, and Yao}{Lin
  et~al.}{2020}]{lin2020}
Lin, Z., L.~Shao, and F.~Yao (2020).
\newblock Intrinsic riemannian functional data analysis for sparse longitudinal
  observations.
\newblock {\em arXiv preprint arXiv:2009.07427\/}.

\bibitem[\protect\citeauthoryear{Lin and Yao}{Lin and Yao}{2019}]{lin2019}
Lin, Z. and F.~Yao (2019).
\newblock {Intrinsic Riemannian functional data analysis}.
\newblock {\em The Annals of Statistics\/}~{\em 47\/}(6), 3533 -- 3577.

\bibitem[\protect\citeauthoryear{Lin and Yao}{Lin and Yao}{2021}]{lin2021}
Lin, Z. and F.~Yao (2021).
\newblock Functional regression on the manifold with contamination.
\newblock {\em Biometrika\/}~{\em 108\/}(1), 167--181.

\bibitem[\protect\citeauthoryear{McCann}{McCann}{1997}]{mccann1997}
McCann, R.~J. (1997).
\newblock A convexity principle for interacting gases.
\newblock {\em Advances in mathematics\/}~{\em 128\/}(1), 153--179.

\bibitem[\protect\citeauthoryear{Panaretos and Zemel}{Panaretos and
  Zemel}{2020}]{panaretos2020}
Panaretos, V.~M. and Y.~Zemel (2020).
\newblock {\em An invitation to statistics in Wasserstein space}.
\newblock Springer Nature.

\bibitem[\protect\citeauthoryear{Petersen and M{\"u}ller}{Petersen and
  M{\"u}ller}{2019}]{petersen2019bmka}
Petersen, A. and H.-G. M{\"u}ller (2019).
\newblock Wasserstein covariance for multiple random densities.
\newblock {\em Biometrika\/}~{\em 106\/}(2), 339--351.

\bibitem[\protect\citeauthoryear{Petersen, M{\"u}ller, et~al.}{Petersen
  et~al.}{2016}]{petersen2016}
Petersen, A., H.-G. M{\"u}ller, et~al. (2016).
\newblock Functional data analysis for density functions by transformation to a
  hilbert space.
\newblock {\em Annals of Statistics\/}~{\em 44\/}(1), 183--218.

\bibitem[\protect\citeauthoryear{Ramsay and Silverman}{Ramsay and
  Silverman}{2006}]{ramsay2005}
Ramsay, J. and B.~Silverman (2006).
\newblock {\em Functional Data Analysis}.
\newblock Springer Science \& Business Media.

\bibitem[\protect\citeauthoryear{Van~Essen, Smith, Barch, Behrens, Yacoub,
  Ugurbil, Consortium, et~al.}{Van~Essen et~al.}{2013}]{van2013}
Van~Essen, D.~C., S.~M. Smith, D.~M. Barch, T.~E. Behrens, E.~Yacoub,
  K.~Ugurbil, W.-M.~H. Consortium, et~al. (2013).
\newblock The wu-minn human connectome project: an overview.
\newblock {\em Neuroimage\/}~{\em 80}, 62--79.

\bibitem[\protect\citeauthoryear{Wahl}{Wahl}{2020}]{wahl2020}
Wahl, M. (2020).
\newblock Information inequalities for the estimation of principal components.
\newblock {\em arXiv preprint arXiv:2005.06869\/}.

\bibitem[\protect\citeauthoryear{Yang, M{\"u}ller, and Stadtm{\"u}ller}{Yang
  et~al.}{2011}]{yang2011}
Yang, W., H.-G. M{\"u}ller, and U.~Stadtm{\"u}ller (2011).
\newblock Functional singular component analysis.
\newblock {\em Journal of the Royal Statistical Society: Series B (Statistical
  Methodology)\/}~{\em 73\/}(3), 303--324.

\bibitem[\protect\citeauthoryear{Yao, M{\"u}ller, and Wang}{Yao
  et~al.}{2005a}]{yao2005jasa}
Yao, F., H.-G. M{\"u}ller, and J.-L. Wang (2005a).
\newblock Functional data analysis for sparse longitudinal data.
\newblock {\em Journal of the American statistical association\/}~{\em
  100\/}(470), 577--590.

\bibitem[\protect\citeauthoryear{Yao, M{\"u}ller, and Wang}{Yao
  et~al.}{2005b}]{yao2005aos}
Yao, F., H.-G. M{\"u}ller, and J.-L. Wang (2005b).
\newblock Functional linear regression analysis for longitudinal data.
\newblock {\em The Annals of Statistics\/}, 2873--2903.

\bibitem[\protect\citeauthoryear{Yuan and Cai}{Yuan and Cai}{2010}]{yuan2010}
Yuan, M. and T.~T. Cai (2010).
\newblock A reproducing kernel hilbert space approach to functional linear
  regression.
\newblock {\em The Annals of Statistics\/}~{\em 38\/}(6), 3412--3444.

\bibitem[\protect\citeauthoryear{Zhou and Chen}{Zhou and Chen}{2020}]{zhou2020}
Zhou, Y. and D.-R. Chen (2020).
\newblock The optimal rate of canonical correlation analysis for stochastic
  processes.
\newblock {\em Journal of Statistical Planning and Inference\/}~{\em 207},
  276--287.

\end{thebibliography}

\end{document}